\documentclass[sigplan,10pt]{acmart}
\usepackage[utf8x]{inputenc}
\usepackage{graphicx}
\usepackage{caption}
\usepackage{subcaption}
\usepackage{ragged2e}
\usepackage{float}
\usepackage{pgf,pgfarrows,pgfnodes,pgfautomata,pgfheaps}
\usepackage{pstricks}
\usepackage{sidecap}
\usepackage{etex}
\usepackage{paralist}
\usepackage{colortbl}
\usepackage{xcolor}
\usepackage{pst-node}
\usepackage{pst-grad}
\usepackage{pst-text}
\usepackage{pst-rel-points}
\usepackage{amsmath}
\usepackage{amsthm}
\usepackage{cancel}
\usepackage{amsfonts}
\usepackage{epsfig}
\usepackage{calc}
\usepackage{ifthen}
\usepackage{rotating}
\usepackage{multido}
\usepackage{multirow}
\usepackage{hhline}
\usepackage{ar}
\usepackage{hyperref}
\usepackage{array}
\usepackage{longtable}
\usepackage[normalem]{ulem}
\usepackage{arydshln}
\usepackage{stmaryrd}
\usepackage{fancyhdr}
\usepackage{lmodern}
\usepackage{algorithm}
\usepackage{xspace}
\usepackage[noend]{algpseudocode}

\usepackage{booktabs}
\usepackage{paralist}
\usepackage{listings}
\emergencystretch=\maxdimen
\theoremstyle{definition}
\newtheorem{myexample}{Example}
\newtheorem{thm}{Theorem}

\acmConference[]{}{}{}
\acmYear{2018}
\acmISBN{} 
\acmDOI{} 
\startPage{1}

\setcopyright{none}

\usepackage{booktabs}   
\usepackage{subcaption} 

\begin{document}

\title[Which Part of the Heap is Useful?]{Which Part of the Heap is Useful?\\
Improving Heap Liveness Analysis}

\author{Vini Kanvar}
\authornote{Partially supported by a TCS fellowship}          
\affiliation{
  \institution{Indian Institute of Technology Bombay}            
  \country{India}
 }
\email{vkanvar@gmail.com}          

\author{Uday P. Khedker}
\affiliation{
  \institution{Indian Institute of Technology Bombay}           
  \country{India}
  }
\email{uday@cse.iitb.ac.in}         


\begin{abstract}

With the growing sizes of data structures allocated in heap, understanding the
actual use of heap memory is critically important for minimizing cache misses
and reclaiming unused memory.  A static analysis aimed at this is difficult
because the heap locations are unnamed. Using allocation sites to name them
creates very few distinctions making it difficult to identify 
allocated heap locations that are  not used. Heap liveness analysis
using access graphs solves this problem by 
\begin{inparaenum}[\em (a)]
\item using  a storeless model of heap memory by naming the locations with access paths,  and
\item representing the unbounded sets of access paths (which are regular languages) as finite automata.
\end{inparaenum}

We improve the scalability and efficiency of heap liveness analysis, and reduce the
amount of computed heap liveness information by using {\em deterministic
automata} and by {\em minimizing the inclusion of aliased access paths} in the
language. Practically, our field-, flow-, context-sensitive liveness analysis
on SPEC CPU2006 benchmarks scales to 36 kLoC (existing analysis scales to 10.5
kLoC) and improves efficiency even up to 99\%. For some of the benchmarks, our
technique shows multifold reduction in the computed liveness information,
ranging from 2 to 100 times (in terms of the number of live access paths),
without compromising on soundness. 
\end{abstract}



\maketitle

\newcommand{\mycomment}[1]{}

\newcommand{\eqdef}{\overset{\sf def}{=}}

\newcommand{\true}{\mbox{\text{\sf true}}\xspace}
\newcommand{\new}{\mbox{\text{\sf new}}\xspace}
\newcommand{\Null}{\mbox{\text{\sf null}}\xspace}
\newcommand{\Inn}{{\text{\sf in}}\xspace}
\newcommand{\Outn}{{\text{\sf out}}\xspace}
\newcommand{\Gen}{{\text{\sf gen}}\xspace}
\newcommand{\Kill}{{\text{\sf kill}}\xspace}
\newcommand{\MustLnA}{{\text{\sf MustLnA}}\xspace}
\newcommand{\MayLnA}{{\text{\sf MayLnA}}\xspace}
\newcommand{\LnA}{{\text{$\overline{\text{\sf LnA}}$}}\xspace}
\newcommand{\Temp}{{\text{\sf temp}}\xspace}
\newcommand{\Succ}{{\text{\sf succ}}\xspace}
\newcommand{\END}{{\text{\sf end}}\xspace}
\newcommand{\deref}{{\text{\sf deref}}\xspace}

\newcommand{\START}{{\text{\sf start}}\xspace} 
\newcommand{\stack}{{\text{\sf stack}}\xspace} 
\renewcommand{\SS}{\text{\sf S}\xspace}
\renewcommand{\u}{\text{\sf u}\xspace}
\newcommand{\RR}{\text{\sf R}\xspace}
\newcommand{\TT}{\text{\sf T}\xspace}
\newcommand{\Q}{\text{\sf Q}\xspace}
\newcommand{\LL}{\text{\sf L}\xspace}
\newcommand{\Liv}{\text{\sf Liv}\xspace}
\newcommand{\G}{\text{\sf G}\xspace}
\newcommand{\W}{\text{\em w}\xspace}
\newcommand{\X}{\text{\em x}\xspace}
\newcommand{\Y}{\text{\em y}\xspace}
\newcommand{\Z}{\text{\em z}\xspace}
\newcommand{\f}{\text{\em f}\xspace}
\newcommand{\g}{\text{\em g}\xspace}
\newcommand{\h}{\text{\em h}\xspace}
\newcommand{\n}{\text{\em n}\xspace}
\renewcommand{\a}{\text{\em a}\xspace}
\renewcommand{\v}{\text{\em v}\xspace}
\newcommand{\N}{\text{\sf N}\xspace}
\newcommand{\E}{\text{\sf E}\xspace}
\newcommand{\F}{\text{\sf F}\xspace}
\newcommand{\V}{\text{\sf V}\xspace}
\newcommand{\q}{\text{\em q}\xspace}
\newcommand{\p}{\text{\em p}\xspace}
\renewcommand{\c}{\text{\em c}\xspace}
\newcommand{\s}{\text{\em s}\xspace}
\renewcommand{\k}{\text{\em k}\xspace}
\renewcommand{\a}{\text{\em a}\xspace}
\renewcommand{\t}{\text{\em t}\xspace}

\newcommand{\node}{\text{\em node}\xspace}
\newcommand{\tmp}{\text{\em tmp}\xspace}
\newcommand{\data}{\text{\em data}\xspace}
\newcommand{\head}{\text{\em head}\xspace}
\newcommand{\lft}{\text{\em lft}\xspace}
\newcommand{\rgt}{\text{\em rgt}\xspace}
\newcommand{\rooot}{\text{\em root}\xspace}
\newcommand{\fhead}{\text{\em fhead}\xspace}
\newcommand{\unext}{\text{\em unext}\xspace}
\newcommand{\fnext}{\text{\em fnext}\xspace}
\newcommand{\fn}{\text{\em fn}\xspace}
\newcommand{\un}{\text{\em un}\xspace}
\newcommand{\fdata}{\text{\em fdata}\xspace}
\newcommand{\fd}{\text{\em d}\xspace}
\newcommand{\user}{\text{\em user}\xspace}
\newcommand{\file}{\text{\em file}\xspace}
\newcommand{\Next}{\text{\em next}\xspace}

\newcommand{\Gmemalgoab}
{{
{\psset{unit=.3mm}
\psset{linewidth=.2mm,arrowsize=4pt,arrowinset=0}
			\begin{pspicture}(-15,-25)(140,17)

		\psrelpoint{origin}{q}{0}{-10}
		\rput(\x{q},\y{q}){\rnode{q}{$\tmp$}}
		\psrelpoint{origin}{w}{0}{10}
		\rput(\x{w},\y{w}){\rnode{q}{$\head$}}

		\psrelpoint{q}{s}{35}{8}
		\rput(\x{s},\y{s}){\rnode{s}{
				\begin{tabular}{@{}c@{}}
				\psframebox[fillstyle=solid,
				fillcolor=white,framesep=0]{
					\begin{tabular}{@{}c@{}}
					$\Next$ 
					\\
					\psframebox[fillstyle=solid,
					fillcolor=white,framesep=3]{\white $o$}
					\end{tabular}
				} \end{tabular}}}
		{
		\psrelpoint{s}{r}{41}{0}
		\rput(\x{r},\y{r}){\rnode{r}{
				\begin{tabular}{@{}c@{}}
				\psframebox[fillstyle=solid,
				fillcolor=white,framesep=0]{
					\begin{tabular}{@{}c@{}}
					$\Next$ 
					\\
					\psframebox[fillstyle=solid,
					fillcolor=white,framesep=3]{\white $o$}
					\end{tabular}
				} \end{tabular}}}
		}{}
		{
		\psrelpoint{r}{t}{41}{0}
		\rput(\x{t},\y{t}){\rnode{t}{
				\begin{tabular}{@{}c@{}}
				\psframebox[fillstyle=solid,
				fillcolor=white,framesep=0]{
					\begin{tabular}{@{}c@{}}
					$\Next$
					\\
					\psframebox[fillstyle=solid,
					fillcolor=white,framesep=3]{\white $o$}
					\end{tabular}
				} \end{tabular}}}
		}{}
		\putnode{k}{q}{7}{0}{}
		\putnode{z}{s}{-17}{-8}{}
		\ncline[nodesepA=4,nodesepB=-6]{->}{k}{z}
		\putnode{z}{s}{8}{-22}{}

		\putnode{a}{s}{0}{-8}{}
		\putnode{b}{r}{-13}{-8}{}
		\ncline{->}{a}{b}

		\putnode{a}{r}{0}{-8}{}
		\putnode{b}{t}{-13}{-8}{}
		\ncline{->}{a}{b}

		\putnode{k}{w}{7}{0}{}
		\putnode{z}{s}{-17}{3}{}
		\ncline[nodesepA=4,nodesepB=-6]{->}{k}{z}
		\putnode{z}{s}{8}{-22}{}

		\end{pspicture}
		}
}}

\newcommand{\Gmemalgoaa}
{{
{\psset{unit=.35mm}
\psset{linewidth=.2mm,arrowsize=4pt,arrowinset=0}
			\begin{pspicture}(-15,-25)(140,17)

		\psrelpoint{origin}{q}{0}{-10}
		\rput(\x{q},\y{q}){\rnode{q}{$\tmp$}}

		\psrelpoint{q}{s}{35}{8}
		\rput(\x{s},\y{s}){\rnode{s}{
				\begin{tabular}{@{}c@{}}
				\psframebox[fillstyle=solid,
				fillcolor=white,framesep=0]{
					\begin{tabular}{@{}c@{}}
					$\Next$ 
					\\
					\psframebox[fillstyle=solid,
					fillcolor=white,framesep=3]{\white $o$}
					\end{tabular}
				} \end{tabular}}}
		{
		\psrelpoint{s}{r}{41}{0}
		\rput(\x{r},\y{r}){\rnode{r}{
				\begin{tabular}{@{}c@{}}
				\psframebox[fillstyle=solid,
				fillcolor=white,framesep=0]{
					\begin{tabular}{@{}c@{}}
					$\Next$ 
					\\
					\psframebox[fillstyle=solid,
					fillcolor=white,framesep=3]{\white $o$}
					\end{tabular}
				} \end{tabular}}}
		}{}
		{
		\psrelpoint{r}{t}{41}{0}
		\rput(\x{t},\y{t}){\rnode{t}{
				\begin{tabular}{@{}c@{}}
				\psframebox[fillstyle=solid,
				fillcolor=white,framesep=0]{
					\begin{tabular}{@{}c@{}}
					$\Next$
					\\
					\psframebox[fillstyle=solid,
					fillcolor=white,framesep=3]{\white $o$}
					\end{tabular}
				} \end{tabular}}}
		}{}
		\putnode{k}{q}{7}{0}{}
		\putnode{z}{s}{-17}{-8}{}
		\ncline[nodesepA=4,nodesepB=-6]{->}{k}{z}
		\putnode{z}{s}{8}{-22}{}

		\putnode{a}{s}{0}{-8}{}
		\putnode{b}{r}{-13}{-8}{}
		\ncline{->}{a}{b}

		\putnode{a}{r}{0}{-8}{}
		\putnode{b}{t}{-13}{-8}{}
		\ncline{->}{a}{b}

		\end{pspicture}
		}
}}

\newcommand{\Gmemalgoac}
{{
{\psset{unit=.35mm}
\psset{linewidth=.2mm,arrowsize=4pt,arrowinset=0}
			\begin{pspicture}(-15,-25)(140,17)

		\psrelpoint{origin}{q}{0}{-10}
		\rput(\x{q},\y{q}){\rnode{q}{$\head$}}

		\psrelpoint{q}{s}{35}{8}
		\rput(\x{s},\y{s}){\rnode{s}{
				\begin{tabular}{@{}c@{}}
				\psframebox[fillstyle=solid,
				fillcolor=white,framesep=0]{
					\begin{tabular}{@{}c@{}}
					$\Next$ 
					\\
					\psframebox[fillstyle=solid,
					fillcolor=white,framesep=3]{\white $o$}
					\end{tabular}
				} \end{tabular}}}
		{
		\psrelpoint{s}{r}{41}{0}
		\rput(\x{r},\y{r}){\rnode{r}{
				\begin{tabular}{@{}c@{}}
				\psframebox[fillstyle=solid,
				fillcolor=white,framesep=0]{
					\begin{tabular}{@{}c@{}}
					$\Next$ 
					\\
					\psframebox[fillstyle=solid,
					fillcolor=white,framesep=3]{\white $o$}
					\end{tabular}
				} \end{tabular}}}
		}{}
		{
		\psrelpoint{r}{t}{41}{0}
		\rput(\x{t},\y{t}){\rnode{t}{
				\begin{tabular}{@{}c@{}}
				\psframebox[fillstyle=solid,
				fillcolor=white,framesep=0]{
					\begin{tabular}{@{}c@{}}
					$\Next$
					\\
					\psframebox[fillstyle=solid,
					fillcolor=white,framesep=3]{\white $o$}
					\end{tabular}
				} \end{tabular}}}
		}{}
		\putnode{k}{q}{7}{0}{}
		\putnode{z}{s}{-17}{-8}{}
		\ncline[nodesepA=4,nodesepB=-6]{->}{k}{z}
		\putnode{z}{s}{8}{-22}{}

		\putnode{a}{s}{0}{-8}{}
		\putnode{b}{r}{-13}{-8}{}
		\ncline{->}{a}{b}

		\putnode{a}{r}{0}{-8}{}
		\putnode{b}{t}{-13}{-8}{}
		\ncline{->}{a}{b}

		\end{pspicture}
		}
}}

\newcommand{\Gmemalgob}
{{
{\psset{unit=.35mm}
\psset{linewidth=.2mm,arrowsize=4pt,arrowinset=0}
			\begin{pspicture}(-15,-185)(140,17)

		\psrelpoint{origin}{q}{0}{-10}
		\rput(\x{q},\y{q}){\rnode{q}{\em root}}
		\psrelpoint{origin}{w}{0}{10}
		\rput(\x{w},\y{w}){\rnode{q}{$\user$}}

		\psrelpoint{q}{s}{45}{8}
		\rput(\x{s},\y{s}){\rnode{s}{
				\begin{tabular}{@{}c@{}}
				\white $l_{9}$
				\\
				\psframebox[fillstyle=solid,
				fillcolor=white,framesep=0]{
					\begin{tabular}{@{}l@{}}
					$\unext$ \psframebox[fillstyle=solid,
					fillcolor=white,framesep=2]{\white $o$}
					\\
					$\fhead$ \psframebox[fillstyle=solid,
					fillcolor=white,framesep=2]{\white $o$}
					\end{tabular}
				} \end{tabular}}}
		{
		\psrelpoint{s}{r}{58}{0}
		\rput(\x{r},\y{r}){\rnode{r}{
				\begin{tabular}{@{}c@{}}
				\white $l_{10}$
				\\
				\psframebox[fillstyle=solid,
				fillcolor=white,framesep=0]{
					\begin{tabular}{@{}l@{}}
					$\unext$ \psframebox[fillstyle=solid,
					fillcolor=white,framesep=2]{\white $o$}
					\\
					$\fhead$ \psframebox[fillstyle=solid,
					fillcolor=white,framesep=2]{\white $o$}
					\end{tabular}
				} \end{tabular}}}
		}{}
		{
		\psrelpoint{r}{t}{58}{0}
		\rput(\x{t},\y{t}){\rnode{t}{
				\begin{tabular}{@{}c@{}}
				\white $l_{11}$
				\\
				\psframebox[fillstyle=solid,
				fillcolor=white,framesep=0]{
					\begin{tabular}{@{}l@{}}
					$\unext$ \psframebox[fillstyle=solid,
					fillcolor=white,framesep=2]{\white $o$}
					\\
					$\fhead$ \psframebox[fillstyle=solid,
					fillcolor=white,framesep=2]{\white $o$}
					\end{tabular}
				} \end{tabular}}}
		}{}
		\putnode{k}{q}{7}{0}{}
		\putnode{z}{s}{-17}{-8}{}
		\ncline[nodesepA=4,nodesepB=2]{->}{k}{z}
		\putnode{z}{s}{8}{-22}{}

		\putnode{a}{s}{12}{-8}{}
		\putnode{b}{r}{-20}{-8}{}
		\ncline{->}{a}{b}

		\putnode{a}{r}{12}{-8}{}
		\putnode{b}{t}{-20}{-8}{}
		\ncline{->}{a}{b}

		\putnode{k}{w}{7}{0}{}
		\putnode{z}{s}{-17}{3}{}
		\ncline[nodesepA=4,nodesepB=2]{->}{k}{z}
		\putnode{z}{s}{8}{-22}{}


		\psrelpoint{origin}{w1}{0}{-65}
		\rput(\x{w1},\y{w1}){\rnode{w1}{$\file$}}

		\psrelpoint{w1}{s1}{45}{8}
		\rput(\x{s1},\y{s1}){\rnode{s1}{
				\begin{tabular}{@{}c@{}}
				\white $l_{9}$
				\\
				\psframebox[fillstyle=solid,
				fillcolor=white,framesep=0]{
					\begin{tabular}{@{}l@{}}
					$\fdata$ \psframebox[fillstyle=solid,
					fillcolor=white,framesep=2]{\white $o$}
					\\
					$\fnext$ \psframebox[fillstyle=solid,
					fillcolor=white,framesep=2]{\white $o$}
					\end{tabular}
				} \end{tabular}}}
		{
		\psrelpoint{s1}{r1}{58}{0}
		\rput(\x{r1},\y{r1}){\rnode{r1}{
				\begin{tabular}{@{}c@{}}
				\white $l_{10}$
				\\
				\psframebox[fillstyle=solid,
				fillcolor=white,framesep=0]{
					\begin{tabular}{@{}l@{}}
					$\fdata$ \psframebox[fillstyle=solid,
					fillcolor=white,framesep=2]{\white $o$}
					\\
					$\fnext$ \psframebox[fillstyle=solid,
					fillcolor=white,framesep=2]{\white $o$}
					\end{tabular}
				} \end{tabular}}}
		}{}
		{
		\psrelpoint{r1}{t1}{58}{0}
		\rput(\x{t1},\y{t1}){\rnode{t1}{
				\begin{tabular}{@{}c@{}}
				\white $l_{11}$
				\\
				\psframebox[fillstyle=solid,
				fillcolor=white,framesep=0]{
					\begin{tabular}{@{}l@{}}
					$\fdata$ \psframebox[fillstyle=solid,
					fillcolor=white,framesep=2]{\white $o$}
					\\
					$\fnext$ \psframebox[fillstyle=solid,
					fillcolor=white,framesep=2]{\white $o$}
					\end{tabular}
				} \end{tabular}}}
		}{}

		\putnode{k1}{w1}{7}{0}{}
		\putnode{z1}{s1}{-17}{3}{}
		\ncline[nodesepA=4,nodesepB=2]{->}{k1}{z1}
		\putnode{z1}{s1}{8}{-22}{}

		\putnode{c1}{s}{12}{-18}{}
		\putnode{d1}{s1}{12}{11}{}
		\ncline[nodesepB=-4]{->}{c1}{d1}

		\putnode{e1}{r}{12}{-18}{}
		\putnode{f1}{r1}{12}{11}{}
		\ncline[nodesepB=-4]{->}{e1}{f1}

		\putnode{g1}{t}{12}{-18}{}
		\putnode{h1}{t1}{12}{11}{}
		\ncline[nodesepB=-4]{->}{g1}{h1}


		\psrelpoint{w1}{w2}{0}{-47}
		\rput(\x{w2},\y{w2}){\rnode{w2}{}}

		\psrelpoint{w2}{s2}{45}{8}
		\rput(\x{s2},\y{s2}){\rnode{s2}{
				\begin{tabular}{@{}c@{}}
				\white $l_{9}$
				\\
				\psframebox[fillstyle=solid,
				fillcolor=white,framesep=0]{
					\begin{tabular}{@{}l@{}}
					$\fdata$ \psframebox[fillstyle=solid,
					fillcolor=white,framesep=2]{\white $o$}
					\\
					$\fnext$ \psframebox[fillstyle=solid,
					fillcolor=white,framesep=2]{\white $o$}
					\end{tabular}
				} \end{tabular}}}
		{
		\psrelpoint{s2}{r2}{58}{0}
		\rput(\x{r2},\y{r2}){\rnode{r2}{
				\begin{tabular}{@{}c@{}}
				\white $l_{10}$
				\\
				\psframebox[fillstyle=solid,
				fillcolor=white,framesep=0]{
					\begin{tabular}{@{}l@{}}
					$\fdata$ \psframebox[fillstyle=solid,
					fillcolor=white,framesep=2]{\white $o$}
					\\
					$\fnext$ \psframebox[fillstyle=solid,
					fillcolor=white,framesep=2]{\white $o$}
					\end{tabular}
				} \end{tabular}}}
		}{}
		{
		\psrelpoint{r2}{t2}{58}{0}
		\rput(\x{t2},\y{t2}){\rnode{t2}{
				\begin{tabular}{@{}c@{}}
				\white $l_{11}$
				\\
				\psframebox[fillstyle=solid,
				fillcolor=white,framesep=0]{
					\begin{tabular}{@{}l@{}}
					$\fdata$ \psframebox[fillstyle=solid,
					fillcolor=white,framesep=2]{\white $o$}
					\\
					$\fnext$ \psframebox[fillstyle=solid,
					fillcolor=white,framesep=2]{\white $o$}
					\end{tabular}
				} \end{tabular}}}
		}{}

		\putnode{c2}{s1}{12}{-18}{}
		\putnode{d2}{s2}{12}{11}{}
		\ncline[nodesepB=-4]{->}{c2}{d2}

		\putnode{e2}{r1}{12}{-18}{}
		\putnode{f2}{r2}{12}{11}{}
		\ncline[nodesepB=-4]{->}{e2}{f2}

		\putnode{g2}{t1}{12}{-18}{}
		\putnode{h2}{t2}{12}{11}{}
		\ncline[nodesepB=-4]{->}{g2}{h2}


		\psrelpoint{w2}{w3}{0}{-47}
		\rput(\x{w3},\y{w3}){\rnode{w3}{}}

		\psrelpoint{w3}{s3}{45}{8}
		\rput(\x{s3},\y{s3}){\rnode{s3}{
				\begin{tabular}{@{}c@{}}
				\white $l_{9}$
				\\
				\psframebox[fillstyle=solid,
				fillcolor=white,framesep=0]{
					\begin{tabular}{@{}l@{}}
					$\fdata$ \psframebox[fillstyle=solid,
					fillcolor=white,framesep=2]{\white $o$}
					\\
					$\fnext$ \psframebox[fillstyle=solid,
					fillcolor=white,framesep=2]{\white $o$}
					\end{tabular}
				} \end{tabular}}}
		{
		\psrelpoint{s3}{r3}{58}{0}
		\rput(\x{r3},\y{r3}){\rnode{r3}{
				\begin{tabular}{@{}c@{}}
				\white $l_{10}$
				\\
				\psframebox[fillstyle=solid,
				fillcolor=white,framesep=0]{
					\begin{tabular}{@{}l@{}}
					$\fdata$ \psframebox[fillstyle=solid,
					fillcolor=white,framesep=2]{\white $o$}
					\\
					$\fnext$ \psframebox[fillstyle=solid,
					fillcolor=white,framesep=2]{\white $o$}
					\end{tabular}
				} \end{tabular}}}
		}{}
		{
		\psrelpoint{r3}{t3}{58}{0}
		\rput(\x{t3},\y{t3}){\rnode{t3}{
				\begin{tabular}{@{}c@{}}
				\white $l_{11}$
				\\
				\psframebox[fillstyle=solid,
				fillcolor=white,framesep=0]{
					\begin{tabular}{@{}l@{}}
					$\fdata$ \psframebox[fillstyle=solid,
					fillcolor=white,framesep=2]{\white $o$}
					\\
					$\fnext$ \psframebox[fillstyle=solid,
					fillcolor=white,framesep=2]{\white $o$}
					\end{tabular}
				} \end{tabular}}}
		}{}

		\putnode{c3}{s2}{12}{-18}{}
		\putnode{d3}{s3}{12}{11}{}
		\ncline[nodesepB=-4]{->}{c3}{d3}

		\putnode{e3}{r2}{12}{-18}{}
		\putnode{f3}{r3}{12}{11}{}
		\ncline[nodesepB=-4]{->}{e3}{f3}

		\putnode{g3}{t2}{12}{-18}{}
		\putnode{h3}{t3}{12}{11}{}
		\ncline[nodesepB=-4]{->}{g3}{h3}

		\end{pspicture}
		}
}}

\newcommand{\Greprab}
{{
{\psset{unit=.35mm}
\psset{linewidth=.2mm,arrowsize=4pt,arrowinset=0}
			\begin{pspicture}(-15,-25)(140,17)

		\psrelpoint{origin}{q}{0}{-10}
		\rput(\x{q},\y{q}){\rnode{q}{}}

		\psrelpoint{q}{s}{35}{8}
		\rput(\x{s},\y{s}){\rnode{s}{
				\begin{tabular}{@{}c@{}}
				\white $l_{9}$
				\\
				\psframebox[fillstyle=solid,
				fillcolor=white,framesep=3]{
					$\f$ \psframebox[fillstyle=solid,
					fillcolor=white,framesep=3]{\white $l$}
				} \end{tabular}}}
		{
		\psrelpoint{s}{r}{41}{0}
		\rput(\x{r},\y{r}){\rnode{r}{
				\begin{tabular}{@{}c@{}}
				\white $l_{10}$
				\\
				\psframebox[fillstyle=solid,
				fillcolor=white,framesep=3]{
					\ \ \psframebox[fillstyle=solid,
					fillcolor=white,framesep=3]{\white $l$}
				} \end{tabular}}}
		}{}
		\putnode{k}{q}{7}{0}{}
		\putnode{z}{s}{-17}{-8}{}
	    	\nccurve[angleA=30,angleB=135,linearc=2,nodesepA=-6,nodesepB=-4]{->}{k}{z}
		\putnode{z}{s}{8}{-22}{}

		\putnode{a}{s}{10}{-8}{}
		\putnode{b}{r}{-13}{-8}{}
		\ncline{->}{a}{b}

		\end{pspicture}
		}
}}

\newcommand{\Grepraa}
{{
{\psset{unit=.35mm}
\psset{linewidth=.2mm,arrowsize=4pt,arrowinset=0}
			\begin{pspicture}(-15,-25)(140,17)

		\psrelpoint{origin}{q}{0}{-10}
		\rput(\x{q},\y{q}){\rnode{q}{}}
		\psrelpoint{origin}{w}{0}{10}
		\rput(\x{w},\y{w}){\rnode{q}{}}

		\psrelpoint{q}{s}{35}{8}
		\rput(\x{s},\y{s}){\rnode{s}{
				\begin{tabular}{@{}c@{}}
				\white $l_{9}$
				\\
				\psframebox[fillstyle=solid,
				fillcolor=white,framesep=3]{
					$\f$ \psframebox[fillstyle=solid,
					fillcolor=white,framesep=3]{\white $l$}
				} \end{tabular}}}
		{
		\psrelpoint{s}{r}{41}{0}
		\rput(\x{r},\y{r}){\rnode{r}{
				\begin{tabular}{@{}c@{}}
				\white $l_{10}$
				\\
				\psframebox[fillstyle=solid,
				fillcolor=white,framesep=3]{
					\ \ \psframebox[fillstyle=solid,
					fillcolor=white,framesep=3]{\white $l$}
				} \end{tabular}}}
		}{}
		\putnode{k}{q}{7}{0}{}
		\putnode{z}{s}{-17}{-8}{}
	    	\nccurve[angleA=30,angleB=135,linearc=2,nodesepA=-6,nodesepB=-4]{->}{k}{z}
		\putnode{z}{s}{8}{-22}{}

		\putnode{a}{s}{10}{-8}{}
		\putnode{b}{r}{-13}{-8}{}
		\ncline{->}{a}{b}

		\putnode{k}{w}{7}{-14}{}
		\putnode{z}{s}{-17}{-2}{}
	    	\nccurve[angleA=-30,angleB=210,linearc=2,nodesepA=-6,nodesepB=-4]{->}{k}{z}
		\putnode{z}{s}{8}{-22}{}

		\end{pspicture}
		}
}}

\newcommand{\repra}
{
\renewcommand{\arraystretch}{.6}
  \scalebox{.9}
  {
    \psset{unit=.75mm}
    \psset{linewidth=.3mm}
    \psset{linewidth=.3mm,arrowsize=5pt,arrowinset=0}
    \begin{pspicture}(-45,-60)(75,70)

    \putnode{f}{origin}{17}{60}{\START \psframebox{\begin{tabular}{@{}l@{}}$\quad {\white f} \ \ \ {\white f} \quad$\end{tabular}}\white\ \START}
    \putnode{a}{f}{0}{-28}{\psframebox{\begin{tabular}{@{}l@{}}$\quad {\white f} \dots {\white f} \quad$\end{tabular}}}
    \putnode{b}{a}{-15}{-28}{\psframebox{\begin{tabular}{@{}l@{}}$\quad {\white f} \dots {\white f} \quad$\end{tabular}}}
    \putnode{c}{a}{15}{-28}{\psframebox{\begin{tabular}{@{}l@{}}$\quad {\white f} \dots {\white f} \quad$\end{tabular}}}
    \putnode{d}{a}{0}{-56}{\psframebox{\begin{tabular}{@{}l@{}}Dereference $\f$\end{tabular}}\white\ }
    \putnode{e}{d}{0}{-28}{\END \psframebox{\begin{tabular}{@{}l@{}}$\quad {\white f} \ \ \ {\white f} \quad$\end{tabular}}\white\ \END}

    \nccurve[angleA=270,angleB=90,loopsize=20,linearc=2,offset=0,arm=5]{->}{a}{b}
    \nccurve[angleA=270,angleB=90,loopsize=-20,linearc=2,offset=0,arm=5]{->}{b}{d}
    \nccurve[angleA=270,angleB=90,loopsize=-20,linearc=2,offset=0,arm=5]{->}{a}{c}
    \nccurve[angleA=270,angleB=90,loopsize=20,linearc=2,offset=0,arm=5]{->}{c}{d}
	\ncline{->}{f}{a}
	\ncline{->}{d}{e}

    \putnode{x}{d}{-30}{12}{\Grepraa}
    \putnode{y}{d}{40}{12}{\Greprab}

    \putnode{z}{d}{5}{8}{$\q$}

    \putnode{p}{a}{-90}{77}{${\mathbf --}$}
    \putnode{q}{p}{0}{-125}{${\mathbf --}$}
	\ncline[linewidth=1mm,linestyle=dotted]{-}{p}{q}
	\ncput*{$\pi_1$}

    \putnode{p}{a}{30}{77}{${\mathbf --}$}
    \putnode{q}{p}{0}{-125}{${\mathbf --}$}
	\ncline[linewidth=1mm,linestyle=dotted]{-}{p}{q}
	\ncput*{$\pi_2$}

   \end{pspicture}
   }
} 

\newcommand{\Greprbb}
{{
{\psset{unit=.35mm}
\psset{linewidth=.2mm,arrowsize=4pt,arrowinset=0}
			\begin{pspicture}(-15,-25)(140,17)

		\psrelpoint{origin}{q}{0}{-10}
		\rput(\x{q},\y{q}){\rnode{q}{}}

		\psrelpoint{q}{s}{35}{8}
		\rput(\x{s},\y{s}){\rnode{s}{
				\begin{tabular}{@{}c@{}}
				\white $l_{9}$
				\\
				\psframebox[fillstyle=solid,
				fillcolor=white,framesep=3]{
					$\f$ \psframebox[fillstyle=solid,
					fillcolor=white,framesep=3]{\white $l$}
				} \end{tabular}}}
		{
		\psrelpoint{s}{r}{41}{0}
		\rput(\x{r},\y{r}){\rnode{r}{
				\begin{tabular}{@{}c@{}}
				\white $l_{10}$
				\\
				\psframebox[fillstyle=solid,
				fillcolor=white,framesep=3]{
					\ \ \psframebox[fillstyle=solid,
					fillcolor=white,framesep=3]{\white $l$}
				} \end{tabular}}}
		}{}
		\putnode{k}{q}{7}{0}{}
		\putnode{z}{s}{-17}{-8}{}
	    	\nccurve[angleA=30,angleB=135,linearc=2,nodesepA=-6,nodesepB=-4]{->}{k}{z}
		\putnode{z}{s}{8}{-22}{}

		\putnode{a}{s}{4}{-8}{}
		\putnode{b}{r}{-13}{-8}{}
		\ncline{->}{a}{b}

		\putnode{k}{r}{10}{-8}{}
		\putnode{z}{k}{15}{8}{}
	    	\nccurve[angleA=-30,angleB=210,linearc=2,nodesepA=-6,nodesepB=-4]{->}{k}{z}

		\putnode{k}{r}{10}{-8}{}
		\putnode{z}{k}{15}{-8}{}
	    	\nccurve[angleA=30,angleB=135,linearc=2,nodesepA=-6,nodesepB=-4]{->}{k}{z}

		\end{pspicture}
		}
}}

\newcommand{\Greprba}
{{
{\psset{unit=.35mm}
\psset{linewidth=.2mm,arrowsize=4pt,arrowinset=0}
			\begin{pspicture}(-15,-25)(140,17)

		\psrelpoint{origin}{q}{0}{-10}
		\rput(\x{q},\y{q}){\rnode{q}{}}
		\psrelpoint{origin}{w}{0}{10}
		\rput(\x{w},\y{w}){\rnode{q}{}}

		\psrelpoint{q}{s}{35}{8}
		\rput(\x{s},\y{s}){\rnode{s}{
				\begin{tabular}{@{}c@{}}
				\white $l_{9}$
				\\
				\psframebox[fillstyle=solid,
				fillcolor=white,framesep=3]{
					$\f$ \psframebox[fillstyle=solid,
					fillcolor=white,framesep=3]{\white $l$}
				} \end{tabular}}}
		{
		\psrelpoint{s}{r}{41}{0}
		\rput(\x{r},\y{r}){\rnode{r}{
				\begin{tabular}{@{}c@{}}
				\white $l_{10}$
				\\
				\psframebox[fillstyle=solid,
				fillcolor=white,framesep=3]{
					\ \ \psframebox[fillstyle=solid,
					fillcolor=white,framesep=3]{\white $l$}
				} \end{tabular}}}
		}{}
		\putnode{k}{q}{7}{0}{}
		\putnode{z}{s}{-17}{-8}{}
	    	\nccurve[angleA=30,angleB=135,linearc=2,nodesepA=-6,nodesepB=-4]{->}{k}{z}
		\putnode{z}{s}{8}{-22}{}

		\putnode{a}{s}{4}{-8}{}
		\putnode{b}{r}{-13}{-8}{}
		\ncline{->}{a}{b}

		\putnode{k}{w}{7}{-14}{}
		\putnode{z}{s}{-17}{-2}{}
	    	\nccurve[angleA=-30,angleB=210,linearc=2,nodesepA=-6,nodesepB=-4]{->}{k}{z}
		\putnode{z}{s}{8}{-22}{}

		\putnode{k}{r}{10}{-8}{}
		\putnode{z}{k}{15}{8}{}
	    	\nccurve[angleA=-30,angleB=210,linearc=2,nodesepA=-6,nodesepB=-4]{->}{k}{z}

		\end{pspicture}
		}
}}

\newcommand{\reprb}
{
\renewcommand{\arraystretch}{.6}
  \scalebox{.9}
  {
    \psset{unit=.75mm}
    \psset{linewidth=.3mm}
    \psset{linewidth=.3mm,arrowsize=5pt,arrowinset=0}
    \begin{pspicture}(-45,-60)(75,70)

    \putnode{f}{origin}{15}{60}{\START \psframebox{\begin{tabular}{@{}l@{}}$\quad {\white f} \ \ \ {\white f} \quad$\end{tabular}}\white\ \START}
    \putnode{a}{f}{0}{-28}{\psframebox{\begin{tabular}{@{}l@{}}$\quad {\white f} \dots {\white f} \quad$\end{tabular}}}
    \putnode{b}{a}{-15}{-35}{\psframebox{\begin{tabular}{@{}l@{}}$\quad {\white f} \dots {\white f} \quad$\end{tabular}}}
    \putnode{c}{a}{15}{-35}{\psframebox{\begin{tabular}{@{}l@{}}$\quad {\white f} \dots {\white f} \quad$\end{tabular}}}
    \putnode{d}{a}{0}{-56}{\psframebox{\begin{tabular}{@{}l@{}}Dereference $\f$\end{tabular}}\white\ }
    \putnode{e}{d}{0}{-28}{\END \psframebox{\begin{tabular}{@{}l@{}}$\quad {\white f} \ \ \ {\white f} \quad$\end{tabular}}\white\ \END}

    \nccurve[angleA=270,angleB=90,loopsize=20,linearc=2,offset=0,arm=5]{->}{a}{b}
    \nccurve[angleA=270,angleB=90,loopsize=-20,linearc=2,offset=0,arm=5]{->}{b}{d}
    \nccurve[angleA=270,angleB=90,loopsize=-20,linearc=2,offset=0,arm=5]{->}{a}{c}
    \nccurve[angleA=270,angleB=90,loopsize=20,linearc=2,offset=0,arm=5]{->}{c}{d}
	\ncline{->}{f}{a}
	\ncline{->}{d}{e}

    \putnode{x}{d}{-27}{45}{\Greprba}
    \putnode{y}{d}{37}{45}{\Greprbb}

    \putnode{z}{d}{10}{8}{$\q$}
    \putnode{z}{d}{10}{47}{$\q'$}

    \putnode{p}{a}{-85}{77}{${\mathbf --}$}
    \putnode{q}{p}{0}{-125}{${\mathbf --}$}
	\ncline[linewidth=1mm,linestyle=dotted]{-}{p}{q}
	\ncput*{$\pi_1$}

    \putnode{p}{a}{35}{77}{${\mathbf --}$}
    \putnode{q}{p}{0}{-125}{${\mathbf --}$}
	\ncline[linewidth=1mm,linestyle=dotted]{-}{p}{q}
	\ncput*{$\pi_2$}

   \end{pspicture}
   }
}

\newcommand{\Gdfs}
{
  \scalebox{.85}
  {
    \psset{unit=.85mm}
    \psset{linewidth=.3mm}
    \psset{linewidth=.3mm,arrowsize=5pt,arrowinset=0}
    \begin{pspicture}(0,-205)(90,-50)

    \putnode{x}{origin}{35}{-60}{1 \psframebox{$\quad{\white f} \dots {\white f}\quad$}\white\ 1}
    \putnode{y}{x}{0}{-13}{3 \psframebox{$\text{\em list\_alloc}(\tmp)$}\white\ 1}
    \putnode{a}{y}{24}{-15}{9 \psframebox{$\t = \tmp$}\white\ 1}
    \putnode{b}{a}{0}{-13}{10 \psframebox{$\head = \text{\ttfamily new}$}\white\ 1}
    \putnode{c}{b}{0}{-13}{11 \psframebox{$\h = \head$}\white\ 1}
    \putnode{d}{c}{0}{-15}{13 \psframebox{$\h \texttt{->} \data = \t \texttt{->} \data$}\white\ 1}
    \putnode{e}{d}{0}{-13}{14 \psframebox{$\h \texttt{->} \Next = \text{\ttfamily{new}}$}\white\ 1}
    \putnode{f}{e}{0}{-13}{15 \psframebox{$\h = \h \texttt{->} \Next$}\white\ 1}
    \putnode{g}{f}{0}{-13}{16 \psframebox{$\t = \t \texttt{->} \Next$}\white\ 1}
    \putnode{j}{g}{0}{-15}{17 \psframebox{\white$\quad f \quad \quad$}\white\ 1}
    \putnode{h}{y}{-24}{-55}{6 \psframebox{$\head = \tmp$}\white\ 1}
    \putnode{i}{j}{-12}{-15}{21 \psframebox{$\head = \head \texttt{->} \Next$}\white\ 1}

    \ncline{->}{x}{y}
    \ncline{->}{y}{a}
    \ncline{->}{a}{b}
    \ncline{->}{b}{c}
    \ncline{->}{c}{d}
    \ncline{->}{d}{e}
    \ncline{->}{e}{f}
    \ncline{->}{f}{g}
    \ncline{->}{g}{j}
    \ncline{->}{y}{h}
    \ncline{->}{h}{i}
    \ncline{->}{j}{i}
    \ncloop[angleA=270,angleB=90,loopsize=-23,linearc=2,offset=4,arm=5]{->}{g}{d}
    \ncloop[angleA=270,angleB=90,loopsize=15,linearc=2,offsetA=-14,offsetB=-14,arm=5.2]{->}{i}{i}

   \end{pspicture}
  }
}

\newcommand{\Gdir}
{
\renewcommand{\arraystretch}{.6}
  \scalebox{.85}
  {
    \psset{unit=.75mm}
    \psset{linewidth=.3mm}
    \psset{linewidth=.3mm,arrowsize=5pt,arrowinset=0}
    \begin{pspicture}(0,-50)(75,70)

    \putnode{a}{origin}{40}{60}{8 \psframebox{\begin{tabular}{@{}l@{}}$\quad \dots \quad$\end{tabular}}\white\ 1}
    \putnode{b}{a}{0}{-20}{10 \psframebox{\begin{tabular}{@{}l@{}}$\user = \user \texttt{->} \unext$\end{tabular}}\white\ 2}
    \putnode{c}{b}{0}{-20}{11 \psframebox{\begin{tabular}{@{}l@{}}$\file =  \user \texttt{->} \fhead$\end{tabular}}\white\ 2}
    \putnode{d}{c}{0}{-20}{13 \psframebox{\begin{tabular}{@{}l@{}}$\file = \file \texttt{->} \fnext$\end{tabular}}\white\ 2}
    \putnode{e}{d}{0}{-20}{14 \psframebox{\begin{tabular}{@{}l@{}}$\file \texttt{->} \fnext \texttt{->}\fdata = \text{\ttfamily new}$\end{tabular}}\white\ 3}
    \putnode{f}{e}{0}{-20}{16 \psframebox{\begin{tabular}{@{}l@{}}$\text{\ttfamily use } (\file \texttt{->} \fnext \texttt{->}\fdata)$\end{tabular}}\white\ 3}

    \ncline{->}{a}{b}
    \ncline{->}{b}{c}
    \ncline{->}{c}{d}
    \ncline{->}{d}{e}
    \ncline{->}{e}{f}
    \ncloop[angleA=270,angleB=90,loopsize=40,linearc=2,offsetA=-4,offsetB=-16,arm=5]{->}{a}{f}
    \ncloop[angleA=270,angleB=90,loopsize=-15,linearc=2,offsetA=15,offsetB=15,arm=5.2]{->}{b}{b}
    \ncloop[angleA=270,angleB=90,loopsize=-15,linearc=2,offsetA=16,offsetB=16,arm=5.2]{->}{d}{d}

    \end{pspicture}
   }
}

\newcommand{\Galias}
{
  \scalebox{.85}
  {
    \psset{unit=.75mm}
    \psset{linewidth=.3mm}
    \psset{linewidth=.3mm,arrowsize=5pt,arrowinset=0}
    \begin{pspicture}(5,-125)(60,-50)

    \putnode{a}{origin}{35}{-60}{1 \psframebox{$\quad{\white f} \dots {\white f}\quad$}\white\ 1}
    \putnode{b}{a}{-15}{-15}{2 \psframebox{$\X = \Y$}\white\ 3}
    \putnode{c}{a}{15}{-20}{4 \psframebox{$\X = \W$}\white\ 3}
    \putnode{d}{b}{0}{-15}{3 \psframebox{$\X \texttt{->} \f = \Z$}\white\ 4}
    \putnode{e}{d}{15}{-15}{5 \psframebox{\white $\X = \Y$}\white\ 5}
    \putnode{f}{e}{0}{-15}{6 \psframebox{$\t = \Y \texttt{->} \f$}\white\ 6}
    \putnode{g}{f}{0}{-15}{7 \psframebox{$\text{use } (\t \texttt{->} \g)$}\white\ 6}

    \ncline{->}{a}{b}
    \ncline{->}{a}{c}
    \ncline{->}{c}{e}
    \ncline{->}{b}{d}
    \ncline{->}{d}{e}
    \ncline{->}{e}{f}
    \ncline{->}{f}{g}

   \end{pspicture}
  }
}

\newcommand{\Gsumm}
{
  \scalebox{.83}
  {
    \psset{unit=.75mm}
    \psset{linewidth=.3mm}
    \psset{linewidth=.3mm,arrowsize=5pt,arrowinset=0}
    \begin{pspicture}(15,-20)(55,55)

    \putnode{a}{origin}{40}{45}{1 \psframebox{\quad {\white \f} \dots \quad {\white \f}}\white\ 1}
    \putnode{z}{a}{0}{-15}{2 \psframebox{$\t = \X \texttt{->} \f$}\white\ 2}
    \putnode{b}{z}{0}{-15}{3 \psframebox{$\t \texttt{->} \f =  \text{\ttfamily new}$}\white\ 2}
    \putnode{c}{b}{0}{-15}{4 \psframebox{$\X = \X \texttt{->} \f$}\white\ 3}
    \putnode{d}{c}{0}{-15}{5 \psframebox{$\text{\ttfamily use } (\X)$}\white\ 4}

    \ncline{->}{a}{z}
    \ncline{->}{z}{b}
    \ncline{->}{b}{c}
    \ncline{->}{c}{d}
    \ncloop[angleA=270,angleB=90,loopsize=20,linearc=2,arm=5,offsetB=-4,offsetA=-4]{->}{c}{z}

    \end{pspicture}
   }
}

\newcommand{\Gmerge}
{
  \scalebox{.85}
  {
    \psset{unit=.75mm}
    \psset{linewidth=.3mm}
    \psset{linewidth=.3mm,arrowsize=5pt,arrowinset=0}
    \begin{pspicture}(10,-20)(60,70)

    \putnode{a}{origin}{40}{60}{1 \psframebox{\begin{tabular}{@{}l@{}}$\quad {\white f} \dots {\white f} \quad$\end{tabular}}\white\ 1}
    \putnode{b}{a}{0}{-15}{2 \psframebox{\begin{tabular}{@{}l@{}}$\X = \X \texttt{->} \h$\end{tabular}}\white\ 2}
    \putnode{c}{b}{0}{-15}{3 \psframebox{\begin{tabular}{@{}l@{}}$\t = \X \texttt{->} \f$\end{tabular}}\white\ 3}
    \putnode{f}{c}{0}{-15}{4 \psframebox{\begin{tabular}{@{}l@{}}$\t\texttt{->} \g = \text{\ttfamily new}$\end{tabular}}\white\ 3}
    \putnode{d}{f}{0}{-15}{5 \psframebox{\begin{tabular}{@{}l@{}}$\Y = \X \texttt{->} \f$\end{tabular}}\white\ 3}
    \putnode{e}{d}{0}{-15}{6 \psframebox{\begin{tabular}{@{}l@{}}$\text{\ttfamily use } (\Y \texttt{->} \g)$\end{tabular}}\white\ 3}

    \ncline{->}{a}{b}
    \ncline{->}{b}{c}
    \ncline{->}{c}{f}
    \ncline{->}{f}{d}
    \ncline{->}{d}{e}
    \ncloop[angleA=270,angleB=90,loopsize=20,linearc=2,offset=-4,arm=5]{->}{a}{d}

    \end{pspicture}
   }
}

\newcommand{\Goldfour}
{\scalebox{.85}{
\large
{\psset{unit=.9mm}
\psset{linewidth=.3mm,arrowsize=5pt,arrowinset=0}
                \psrelpoint{origin}{q}{10}{-2}
                \begin{pspicture}(10,-8)(43,12)
                \rput(\x{q},\y{q}){\rnode{q}{\pscirclebox[fillstyle=solid,linecolor=black,
                                fillcolor=white,framesep=1.7]{$\X$}}}
                \psrelpoint{q}{r}{16}{0}
                \rput(\x{r},\y{r}){\rnode{r}{\pscirclebox[fillstyle=solid,linecolor=black,
                                fillcolor=white,framesep=1.7]{$\f_{4}$}}}
                \ncline{->}{q}{r}
		\naput{$\f$}
                \nccurve[angleA=45,angleB=135,nodesep=-1,ncurv=1.7]{->}{r}{r}
		\nbput{$\f$}
                \psrelpoint{r}{m}{16}{0}
                \rput(\x{m},\y{m}){\rnode{m}{\pscirclebox[fillstyle=solid,linecolor=black,
                                fillcolor=white,framesep=1.7]{$\f_2$}}}
                \ncline{->}{r}{m}
		\naput{$\f$}
                \nccurve[angleA=-45,angleB=-135,nodesep=-1,ncurv=.6]{->}{q}{m}

                \end{pspicture}
                }
}}
\newcommand{\Goldthree}
{\scalebox{.85}{
\large
{\psset{unit=.9mm}
\psset{linewidth=.3mm,arrowsize=5pt,arrowinset=0}
                \psrelpoint{origin}{q}{10}{-2}
                \begin{pspicture}(10,-8)(43,10)
                \rput(\x{q},\y{q}){\rnode{q}{\pscirclebox[fillstyle=solid,linecolor=black,
                                fillcolor=white,framesep=2]{$\X$}}}
                \psrelpoint{q}{r}{14}{0}
                \rput(\x{r},\y{r}){\rnode{r}{\pscirclebox[fillstyle=solid,linecolor=black,
                                fillcolor=white,framesep=1.7]{$\f_{4}$}}}
                \ncline{->}{q}{r}
		\naput{$\f$}
                \psrelpoint{r}{m}{14}{-5}
                \rput(\x{m},\y{m}){\rnode{m}{\pscirclebox[fillstyle=solid,linecolor=black,
                                fillcolor=white,framesep=1.7]{$\f_2$}}}
                \ncline{->}{r}{m}
		\nbput{$\f$}
                \psrelpoint{r}{n}{14}{5}
                \rput(\x{n},\y{n}){\rnode{n}{\pscirclebox[fillstyle=solid,linecolor=black,
                                fillcolor=white,framesep=1.7]{$\f_4$}}}
                \ncline{->}{r}{n}
		\naput{$\f$}

                \end{pspicture}
                }
}}

\newcommand{\Goldthreefinal}
{\scalebox{.85}{
\large
{\psset{unit=.9mm}
\psset{linewidth=.3mm,arrowsize=5pt,arrowinset=0}
                \psrelpoint{origin}{q}{10}{-2}
                \begin{pspicture}(10,-8)(43,12)
                \rput(\x{q},\y{q}){\rnode{q}{\pscirclebox[fillstyle=solid,linecolor=black,
                                fillcolor=white,framesep=2]{$\X$}}}
                \psrelpoint{q}{r}{16}{0}
                \rput(\x{r},\y{r}){\rnode{r}{\pscirclebox[fillstyle=solid,linecolor=black,
                                fillcolor=white,framesep=1.7]{$\f_{4}$}}}
                \ncline{->}{q}{r}
		\naput{$\f$}
                \nccurve[angleA=45,angleB=135,nodesep=-1,ncurv=1.7]{->}{r}{r}
		\nbput{$\f$}
                \psrelpoint{r}{m}{16}{0}
                \rput(\x{m},\y{m}){\rnode{m}{\pscirclebox[fillstyle=solid,linecolor=black,
                                fillcolor=white,framesep=1.7]{$\f_2$}}}
                \ncline{->}{r}{m}
		\naput{$\f$}

                \end{pspicture}
                }
}}

\newcommand{\Goldt}
{\scalebox{.85}{
\large
{\psset{unit=.9mm}
\psset{linewidth=.3mm,arrowsize=5pt,arrowinset=0}
                \begin{pspicture}(5,-7)(15,3)
                \psrelpoint{origin}{q}{10}{-2}
                \rput(\x{q},\y{q}){\rnode{q}{\pscirclebox[fillstyle=solid,linecolor=black,
                                fillcolor=white,framesep=2]{$\t$}}}

                \end{pspicture}
                }
}}

\newcommand{\Goldx}
{\scalebox{.85}{
\large
{\psset{unit=.9mm}
\psset{linewidth=.3mm,arrowsize=5pt,arrowinset=0}
                \begin{pspicture}(8,-7)(12,3)
                \psrelpoint{origin}{q}{10}{-2}
                \rput(\x{q},\y{q}){\rnode{q}{\pscirclebox[fillstyle=solid,linecolor=black,
                                fillcolor=white,framesep=2]{$\X$}}}

                \end{pspicture}
                }
}}

\newcommand{\Goldtwo}
{\scalebox{.85}{
\large
{\psset{unit=.9mm}
\psset{linewidth=.3mm,arrowsize=5pt,arrowinset=0}
                \begin{pspicture}(8,-7)(25,4)
                \psrelpoint{origin}{q}{10}{-2}
                \rput(\x{q},\y{q}){\rnode{q}{\pscirclebox[fillstyle=solid,linecolor=black,
                                fillcolor=white,framesep=2]{$\X$}}}
                \psrelpoint{q}{r}{14}{0}
                \rput(\x{r},\y{r}){\rnode{r}{\pscirclebox[fillstyle=solid,linecolor=black,
                                fillcolor=white,framesep=1.7]{$\f_{4}$}}}
                \ncline{->}{q}{r}
		\naput{$\f$}

                \end{pspicture}
                }
}}

\newcommand{\Goldone}
{\scalebox{.85}{
\large
{\psset{unit=.9mm}
\psset{linewidth=.3mm,arrowsize=5pt,arrowinset=0}
                \begin{pspicture}(8,-11)(25,7)
                \psrelpoint{origin}{q}{10}{-2}
                \rput(\x{q},\y{q}){\rnode{q}{\pscirclebox[fillstyle=solid,linecolor=black,
                                fillcolor=white,framesep=2]{$\X$}}}
                \psrelpoint{q}{r}{14}{-5}
                \rput(\x{r},\y{r}){\rnode{r}{\pscirclebox[fillstyle=solid,linecolor=black,
                                fillcolor=white,framesep=1.7]{$\f_{2}$}}}
                \ncline{->}{q}{r}
		\nbput{$\f$}
                \psrelpoint{q}{m}{14}{5}
                \rput(\x{m},\y{m}){\rnode{m}{\pscirclebox[fillstyle=solid,linecolor=black,
                                fillcolor=white,framesep=1.7]{$\f_4$}}}
                \ncline{->}{q}{m}
		\naput{$\f$}

                \end{pspicture}
                }
}}

\newcommand{\Gnewtwo}
{\scalebox{.85}{
\large
{\psset{unit=.9mm}
\psset{linewidth=.3mm,arrowsize=5pt,arrowinset=0}
                \begin{pspicture}(25,-8)(48,4)
                \psrelpoint{origin}{q}{10}{-2}
                \psrelpoint{q}{r}{19}{0}
                \rput(\x{r},\y{r}){\rnode{r}{\pscirclebox[fillstyle=solid,linecolor=black,
                                fillcolor=white,framesep=2]{$\X$}}}
                \psrelpoint{r}{s}{18}{0}
                \rput(\x{s},\y{s}){\rnode{s}{\pscirclebox[fillstyle=solid,linecolor=black,
                                fillcolor=white,framesep=0]{
                                $\f_{\{4\}}$
                                }}}
                \ncline{->}{r}{s}
                \naput{$\f$}
                \end{pspicture}
                }
}}

\newcommand{\Gnewone}
{\scalebox{.85}{
\large
{\psset{unit=.9mm}
\psset{linewidth=.3mm,arrowsize=5pt,arrowinset=0}
                \begin{pspicture}(25,-8)(45,4)
                \psrelpoint{origin}{q}{10}{-2}
                \psrelpoint{q}{r}{19}{0}
                \rput(\x{r},\y{r}){\rnode{r}{\pscirclebox[fillstyle=solid,linecolor=black,
                                fillcolor=white,framesep=2]{$\X$}}}
                \psrelpoint{r}{s}{18}{0}
                \rput(\x{s},\y{s}){\rnode{s}{\pscirclebox[fillstyle=solid,linecolor=black,
                                fillcolor=white,framesep=-1]{
                                $\f_{\{2,4\}}$
                                }}}
                \ncline{->}{r}{s}
                \naput{$\f$}
                \end{pspicture}
                }
}}

\newcommand{\Gnewthree}
{\scalebox{.85}{
\large
{\psset{unit=.9mm}
\psset{linewidth=.3mm,arrowsize=5pt,arrowinset=0}
                \begin{pspicture}(25,-8)(65,4)
                \psrelpoint{origin}{q}{10}{-2}
                \psrelpoint{q}{r}{19}{0}
                \rput(\x{r},\y{r}){\rnode{r}{\pscirclebox[fillstyle=solid,linecolor=black,
                                fillcolor=white,framesep=2]{$\X$}}}
                \psrelpoint{r}{s}{15}{0}
                \rput(\x{s},\y{s}){\rnode{s}{\pscirclebox[fillstyle=solid,linecolor=black,
                                fillcolor=white,framesep=0]{
                                $\f_{\{4\}}$
                                }}}
                \ncline{->}{r}{s}
                \naput{$\f$}
                \psrelpoint{s}{t}{17}{0}
                \rput(\x{t},\y{t}){\rnode{t}{\pscirclebox[fillstyle=solid,linecolor=black,
                                fillcolor=white,framesep=-1]{
                                $\f_{\{2,4\}}$
                                }}}
                \ncline{->}{s}{t}
                \naput{$\f$}
                \end{pspicture}
                }
}}

\newcommand{\Gzero}
{\scalebox{.85}{
\large
{\psset{unit=.9mm}
\psset{linewidth=.3mm,arrowsize=5pt,arrowinset=0}
                \begin{pspicture}(5,-9)(32,5)
                \psrelpoint{origin}{q}{10}{-2}
                \rput(\x{q},\y{q}){\rnode{q}{\pscirclebox[fillstyle=solid,linecolor=black,
                                fillcolor=white,framesep=2]{$\Y$}}}
                \psrelpoint{q}{t}{14}{0}
                \rput(\x{t},\y{t}){\rnode{t}{\pscirclebox[fillstyle=solid,linecolor=black,
                                fillcolor=white,framesep=1.7]{$\g_{5}$}}}
                \ncline{->}{q}{t}
		\naput{$\g$}

                \end{pspicture}
                }
}}

\newcommand{\Gone}
{\scalebox{.85}{
\large
{\psset{unit=.9mm}
\psset{linewidth=.3mm,arrowsize=5pt,arrowinset=0}
                \begin{pspicture}(5,-9)(42,5)
                \psrelpoint{origin}{q}{10}{-2}
                \rput(\x{q},\y{q}){\rnode{q}{\pscirclebox[fillstyle=solid,linecolor=black,
                                fillcolor=white,framesep=2]{$\X$}}}
                \psrelpoint{q}{t}{14}{0}
                \rput(\x{t},\y{t}){\rnode{t}{\pscirclebox[fillstyle=solid,linecolor=black,
                                fillcolor=white,framesep=1.7]{$\f_{5}$}}}
                \ncline{->}{q}{t}
		\naput{$\f$}
                \psrelpoint{t}{u}{14}{0}
                \rput(\x{u},\y{u}){\rnode{u}{\pscirclebox[fillstyle=solid,linecolor=black,
                                fillcolor=white,framesep=1.7]{$\g_{6}$}}}
                \ncline{->}{t}{u}
		\naput{$\g$}

                \end{pspicture}
                }
}}
\newcommand{\Gtwo}
{\scalebox{.85}{
\large
{\psset{unit=.9mm}
\psset{linewidth=.3mm,arrowsize=5pt,arrowinset=0}
                \begin{pspicture}(5,-12)(28,10)
                \psrelpoint{origin}{q}{10}{-2}
                \rput(\x{q},\y{q}){\rnode{q}{\pscirclebox[fillstyle=solid,linecolor=black,
                                fillcolor=white,framesep=2]{$\X$}}}
                \psrelpoint{q}{s}{14}{7}
                \rput(\x{s},\y{s}){\rnode{s}{\pscirclebox[fillstyle=solid,linecolor=black,
                                fillcolor=white,framesep=1.7]{$\f_{3}$}}}
                \ncline[nodesep=-1]{->}{q}{s}
		\naput{$\f$}
                \psrelpoint{q}{t}{14}{-7}
                \rput(\x{t},\y{t}){\rnode{t}{\pscirclebox[fillstyle=solid,linecolor=black,
                                fillcolor=white,framesep=1.7]{$\f_{4}$}}}
                \ncline[nodesep=-1]{->}{q}{t}
		\nbput{$\f$}

                \end{pspicture}
                }
}}

\newcommand{\Gthree}
{\scalebox{.85}{
\large
{\psset{unit=.9mm}
\psset{linewidth=.3mm,arrowsize=5pt,arrowinset=0}
                \begin{pspicture}(5,-12)(42,10)
                \psrelpoint{origin}{q}{10}{-2}
                \rput(\x{q},\y{q}){\rnode{q}{\pscirclebox[fillstyle=solid,linecolor=black,
                                fillcolor=white,framesep=2]{$\X$}}}
                \psrelpoint{q}{r}{14}{0}
                \rput(\x{r},\y{r}){\rnode{r}{\pscirclebox[fillstyle=solid,linecolor=black,
                                fillcolor=white,framesep=1.7]{$\h_{2}$}}}
                \ncline{->}{q}{r}
		\naput{$\h$}
                \psrelpoint{r}{s}{14}{7}
                \rput(\x{s},\y{s}){\rnode{s}{\pscirclebox[fillstyle=solid,linecolor=black,
                                fillcolor=white,framesep=1.7]{$\f_{3}$}}}
                \ncline[nodesep=-1]{->}{r}{s}
		\naput{$\f$}
                \psrelpoint{r}{t}{14}{-7}
                \rput(\x{t},\y{t}){\rnode{t}{\pscirclebox[fillstyle=solid,linecolor=black,
                                fillcolor=white,framesep=1.7]{$\f_{5}$}}}
                \ncline[nodesep=-1]{->}{r}{t}
		\nbput{$\f$}

                \end{pspicture}
                }
}}

\newcommand{\Gfour}
{\scalebox{.85}{
\large
{\psset{unit=.9mm}
\psset{linewidth=.3mm,arrowsize=5pt,arrowinset=0}
                \begin{pspicture}(9,-18)(52,10)
                \psrelpoint{origin}{q}{10}{-2}
                \rput(\x{q},\y{q}){\rnode{q}{\pscirclebox[fillstyle=solid,linecolor=black,
                                fillcolor=white,framesep=2]{$\X$}}}
                \psrelpoint{q}{r}{14}{0}
                \rput(\x{r},\y{r}){\rnode{r}{\pscirclebox[fillstyle=solid,linecolor=black,
                                fillcolor=white,framesep=1.7]{$\h_{2}$}}}
                \ncline{->}{q}{r}
		\naput{$\h$}
                \psrelpoint{r}{s}{14}{7}
                \rput(\x{s},\y{s}){\rnode{s}{\pscirclebox[fillstyle=solid,linecolor=black,
                                fillcolor=white,framesep=1.7]{$\f_{3}$}}}
                \ncline[nodesep=-1]{->}{r}{s}
		\naput{$\f$}
                \psrelpoint{r}{t}{14}{-7}
                \rput(\x{t},\y{t}){\rnode{t}{\pscirclebox[fillstyle=solid,linecolor=black,
                                fillcolor=white,framesep=1.7]{$\f_{5}$}}}
                \ncline[nodesep=-1]{->}{r}{t}
		\nbput{$\f$}
                \psrelpoint{t}{u}{14}{0}
                \rput(\x{u},\y{u}){\rnode{u}{\pscirclebox[fillstyle=solid,linecolor=black,
                                fillcolor=white,framesep=1.7]{$\g_{6}$}}}
                \ncline{->}{t}{u}
		\naput{$\g$}

                \nccurve[angleA=270,angleB=220,nodesepB=-1,ncurv=0.4]{->}{q}{t}
		\nbput{$\f$}

                \end{pspicture}
                }
}}

\newcommand{\Gnewmergezero}
{\scalebox{.85}{
\large
{\psset{unit=.9mm}
\psset{linewidth=.3mm,arrowsize=5pt,arrowinset=0}
                \begin{pspicture}(5,-9)(35,5)
                \psrelpoint{origin}{q}{10}{-2}
                \rput(\x{q},\y{q}){\rnode{q}{\pscirclebox[fillstyle=solid,linecolor=black,
                                fillcolor=white,framesep=2]{$\Y$}}}
                \psrelpoint{q}{t}{15}{0}
                \rput(\x{t},\y{t}){\rnode{t}{\pscirclebox[fillstyle=solid,linecolor=black,
                                fillcolor=white,framesep=.5]{$\g_{\{5\}}$}}}
                \ncline{->}{q}{t}
		\naput{$\g$}
                \end{pspicture}
                }
}}

\newcommand{\Gnewmergeone}
{\scalebox{.85}{
\large
{\psset{unit=.9mm}
\psset{linewidth=.3mm,arrowsize=5pt,arrowinset=0}
                \begin{pspicture}(8,-9)(40,5)
                \psrelpoint{origin}{q}{10}{-2}
                \rput(\x{q},\y{q}){\rnode{q}{\pscirclebox[fillstyle=solid,linecolor=black,
                                fillcolor=white,framesep=2]{$\X$}}}
                \psrelpoint{q}{t}{14}{0}
                \rput(\x{t},\y{t}){\rnode{t}{\pscirclebox[fillstyle=solid,linecolor=black,
                                fillcolor=white,framesep=.5]{$\f_{\{5\}}$}}}
                \ncline{->}{q}{t}
		\naput{$\f$}
                \psrelpoint{t}{u}{14}{0}
                \rput(\x{u},\y{u}){\rnode{u}{\pscirclebox[fillstyle=solid,linecolor=black,
                                fillcolor=white,framesep=.5]{$\g_{\{6\}}$}}}
                \ncline{->}{t}{u}
		\naput{$\g$}
                \end{pspicture}
                }
}}

\newcommand{\Gnewmergetwo}
{\scalebox{.85}{
\large
{\psset{unit=.9mm}
\psset{linewidth=.3mm,arrowsize=5pt,arrowinset=0}
                \begin{pspicture}(5,-9)(35,5)
                \psrelpoint{origin}{q}{10}{-2}
                \rput(\x{q},\y{q}){\rnode{q}{\pscirclebox[fillstyle=solid,linecolor=black,
                                fillcolor=white,framesep=2]{$\X$}}}
                \psrelpoint{q}{s}{15}{0}
                \rput(\x{s},\y{s}){\rnode{s}{\pscirclebox[fillstyle=solid,linecolor=black,
                                fillcolor=white,framesep=0]{$\f_{\{3,4\}}$}}}
                \ncline{->}{q}{s}
		\naput{$\f$}

                \end{pspicture}
                }
}}

\newcommand{\Gnewmergethree}
{\scalebox{.85}{
\large
{\psset{unit=.9mm}
\psset{linewidth=.3mm,arrowsize=5pt,arrowinset=0}
                \begin{pspicture}(8,-9)(40,5)
                \psrelpoint{origin}{q}{10}{-2}
                \rput(\x{q},\y{q}){\rnode{q}{\pscirclebox[fillstyle=solid,linecolor=black,
                                fillcolor=white,framesep=2]{$\X$}}}
                \psrelpoint{q}{r}{13}{0}
                \rput(\x{r},\y{r}){\rnode{r}{\pscirclebox[fillstyle=solid,linecolor=black,
                                fillcolor=white,framesep=.5]{$\h_{\{2\}}$}}}
                \ncline{->}{q}{r}
		\naput{$\h$}
                \psrelpoint{r}{s}{15}{0}
                \rput(\x{s},\y{s}){\rnode{s}{\pscirclebox[fillstyle=solid,linecolor=black,
                                fillcolor=white,framesep=0]{$\f_{\{3,5\}}$}}}
                \ncline{->}{r}{s}
		\naput{$\f$}
                \end{pspicture}
                }
}}

\newcommand{\Gnewmergefour}
{\scalebox{.85}{
\large
{\psset{unit=.9mm}
\psset{linewidth=.3mm,arrowsize=5pt,arrowinset=0}
                \begin{pspicture}(28,-15)(60,10)
                \psrelpoint{origin}{q}{10}{-2}
                \psrelpoint{q}{p}{19}{0}
                \rput(\x{p},\y{p}){\rnode{p}{\pscirclebox[fillstyle=solid,linecolor=black,
                                fillcolor=white,framesep=2]{$\X$}}}
                \psrelpoint{p}{r}{14}{-6}
                \rput(\x{r},\y{r}){\rnode{r}{\pscirclebox[fillstyle=solid,linecolor=black,
                                fillcolor=white,framesep=-.5]{
                                $\h_{\{2\}}$
                                }}}
                \ncline[nodesep=-1]{->}{p}{r}
                \nbput{$\h$}
                \psrelpoint{r}{s}{15}{0}
                \rput(\x{s},\y{s}){\rnode{s}{\pscirclebox[fillstyle=solid,linecolor=black,
                                fillcolor=white,framesep=-1]{
                                $\f_{\{3,5\}}$
                                }}}
                \ncline{->}{r}{s}
                \nbput{$\f$}
                \psrelpoint{p}{v}{14}{6}
                \rput(\x{v},\y{v}){\rnode{v}{\pscirclebox[fillstyle=solid,linecolor=black,
                                fillcolor=white,framesep=-.5]{
                                $\f_{\{5\}}$
                                }}}
                \ncline[nodesep=-1]{->}{p}{v}
                \naput{$\f$}
                \psrelpoint{v}{u}{15}{0}
                \rput(\x{u},\y{u}){\rnode{u}{\pscirclebox[fillstyle=solid,linecolor=black,
                                fillcolor=white,framesep=-.5]{
                                $\g_{\{6\}}$
                                }}}
                \ncline{->}{v}{u}
                \naput{$\g$}

                \end{pspicture}
                }
}}

\newcommand{\Gfull}
{\scalebox{.85}{
\large
{\psset{unit=.9mm}
\psset{linewidth=.3mm,arrowsize=5pt,arrowinset=0}
                \begin{pspicture}(7,-20)(57,6)
                \psrelpoint{origin}{q}{10}{-2}
                \rput(\x{q},\y{q}){\rnode{q}{\pscirclebox[fillstyle=solid,linecolor=black,
                                fillcolor=white,framesep=1.5]{$\X$}}}
                \psrelpoint{q}{s}{12}{0}
                \rput(\x{s},\y{s}){\rnode{s}{\pscirclebox[fillstyle=solid,linecolor=black,
                                fillcolor=white,framesep=1]{$\f_{2}$}}}
                \ncline{->}{q}{s}
		\naput{$\f$}
                \psrelpoint{s}{r}{12}{0}
                \rput(\x{r},\y{r}){\rnode{r}{\pscirclebox[fillstyle=solid,linecolor=black,
                                fillcolor=white,framesep=1]{$\f_{2}$}}}
                \ncline{->}{s}{r}
		\naput{$\f$}
                \psrelpoint{r}{t}{12}{0}
                \rput(\x{t},\y{t}){\rnode{t}{\pscirclebox[fillstyle=solid,linecolor=black,
                                fillcolor=white,framesep=1]{$\f_{2}$}}}
                \ncline{->}{r}{t}
		\naput{$\f$}
                \psrelpoint{q}{u}{7}{-15}
                \rput(\x{u},\y{u}){\rnode{u}{\pscirclebox[fillstyle=solid,linecolor=black,
                                fillcolor=white,framesep=1]{$\g_{3}$}}}
                \ncline[nodesep=-1]{->}{q}{u}
		\nbput{$\g$}
                \psrelpoint{s}{v}{7}{-15}
                \rput(\x{v},\y{v}){\rnode{v}{\pscirclebox[fillstyle=solid,linecolor=black,
                                fillcolor=white,framesep=1]{$\g_{3}$}}}
                \ncline[nodesep=-1]{->}{s}{v}
		\nbput{$\g$}
                \psrelpoint{r}{w}{7}{-15}
                \rput(\x{w},\y{w}){\rnode{w}{\pscirclebox[fillstyle=solid,linecolor=black,
                                fillcolor=white,framesep=1]{$\g_{3}$}}}
                \ncline[nodesep=-1]{->}{r}{w}
		\nbput{$\g$}
                \psrelpoint{t}{x}{7}{-15}
                \rput(\x{x},\y{x}){\rnode{x}{\pscirclebox[fillstyle=solid,linecolor=black,
                                fillcolor=white,framesep=1]{$\g_{3}$}}}
                \ncline[nodesep=-1]{->}{t}{x}
		\nbput{$\g$}

                \end{pspicture}
                }
}}
\newcommand{\Gcompact}
{\scalebox{.85}{
\large
{\psset{unit=.9mm}
\psset{linewidth=.3mm,arrowsize=5pt,arrowinset=0}
                \begin{pspicture}(7,-25)(35,10)
                \psrelpoint{origin}{q}{10}{-2}
                \rput(\x{q},\y{q}){\rnode{q}{\pscirclebox[fillstyle=solid,linecolor=black,
                                fillcolor=white,framesep=1.8]{$\X$}}}
                \psrelpoint{q}{s}{12}{0}
                \rput(\x{s},\y{s}){\rnode{s}{\pscirclebox[fillstyle=solid,linecolor=black,
                                fillcolor=white,framesep=1]{$\f_{2}$}}}
                \ncline{->}{q}{s}
		\naput{$\f$}
                \nccurve[angleA=60,angleB=120,nodesep=-1,ncurv=2.9]{->}{s}{s}
		\nbput{$\f$}
                \psrelpoint{s}{t}{9}{-15}
                \rput(\x{t},\y{t}){\rnode{t}{\pscirclebox[fillstyle=solid,linecolor=black,
                                fillcolor=white,framesep=1]{$\g_{3}$}}}
                \ncline[nodesep=-1]{->}{s}{t}
		\naput{$\g$}
                \nccurve[angleA=270,angleB=170,nodesepB=0,ncurv=0.4]{->}{q}{t}
		\nbput{$\g$}

                \end{pspicture}
                }
}}

\newcommand{\introcfg}
{
  \scalebox{.83}
  {
    \psset{unit=.75mm}
    \psset{linewidth=.3mm}
        \psset{linewidth=.3mm,arrowsize=5pt,arrowinset=0}
    \begin{pspicture}(22,25)(52,65)

    \putnode{a}{origin}{40}{60}{1 \psframebox{${\white \f}\quad \dots \quad {\white \f}$}\white\ 1}
    \putnode{b}{a}{0}{-15}{2 \psframebox{$\X = \X \texttt{->} \f$}\white\ 2}
    \putnode{c}{b}{0}{-15}{3 \psframebox{$\text{use } (\X \texttt{->} \g)$}\white\ 3}

    \ncline{->}{a}{b}
    \ncline{->}{b}{c}
    \ncloop[angleA=270,angleB=90,loopsize=12,linearc=2,offsetA=-5,offsetB=-6,arm=5.8]{->}{b}{b}

    \end{pspicture}
   }
}

\newcommand{\Gexec}
{{
{\psset{unit=.28mm}
\psset{linewidth=.2mm,arrowsize=4pt,arrowinset=0}
			\begin{pspicture}(-15,-25)(180,17)

		\psrelpoint{origin}{q}{0}{-10}
		\rput(\x{q},\y{q}){\rnode{q}{$\X$}}
		\psrelpoint{q}{s}{35}{8}
		\rput(\x{s},\y{s}){\rnode{s}{
				\begin{tabular}{@{}c@{}}
				$l_1$
				\\
				\psframebox[fillstyle=solid,
				fillcolor=white,framesep=3]{
					$\f$ \psframebox[fillstyle=solid,
					fillcolor=white,framesep=3]{\white $l$}
				} \end{tabular}}}
		{
		\psrelpoint{s}{r}{41}{0}
		\rput(\x{r},\y{r}){\rnode{r}{
				\begin{tabular}{@{}c@{}}
				$l_2$
				\\
				\psframebox[fillstyle=solid,
				fillcolor=white,framesep=3]{
					$\f$ \psframebox[fillstyle=solid,
					fillcolor=white,framesep=3]{\white $l$}
				} \end{tabular}}}
		}{}
		{
		\psrelpoint{r}{t}{41}{0}
		\rput(\x{t},\y{t}){\rnode{t}{
				\begin{tabular}{@{}c@{}}
				$l_3$
				\\
				\psframebox[fillstyle=solid,
				fillcolor=white,framesep=3]{
					$\f$ \psframebox[fillstyle=solid,
					fillcolor=white,framesep=3]{\white $l$}
				} \end{tabular}}}
		}{}
		\putnode{k}{q}{7}{0}{}
		\putnode{z}{s}{-17}{-8}{}
		\ncline[nodesepA=0,nodesepB=-4]{->}{k}{z}
		\putnode{z}{s}{8}{-22}{}

		\putnode{a}{s}{10}{-8}{}
		\putnode{b}{r}{-13}{-8}{}
		\ncline{->}{a}{b}

		\putnode{a}{r}{10}{-8}{}
		\putnode{b}{t}{-13}{-8}{}
		\ncline{->}{a}{b}
			
		\putnode{u}{t}{41}{0}{$\dots$}
                \putnode{a}{t}{10}{-8}{}
                \putnode{b}{u}{-13}{-8}{}
                \ncline{->}{a}{b}{}

		\end{pspicture}
		}
}}

\newcommand{\Gexecone}
{{
{\psset{unit=.28mm}
\psset{linewidth=.2mm,arrowsize=4pt,arrowinset=0}
			\begin{pspicture}(0,-50)(135,17)

		\psrelpoint{origin}{q}{0}{-10}
		\rput(\x{q},\y{q}){\rnode{q}{$\X$}}
		\psrelpoint{q}{s}{30}{8}
		\rput(\x{s},\y{s}){\rnode{s}{
				\begin{tabular}{@{}c@{}}
				$l_1$
				\\
				\psframebox[fillstyle=solid,
				fillcolor=white,framesep=3]{
					$\f$ \psframebox[fillstyle=solid,
					fillcolor=white,framesep=3]{\white $l$}
				} \end{tabular}}}
		{
		\psrelpoint{s}{r}{41}{0}
		\rput(\x{r},\y{r}){\rnode{r}{
				\begin{tabular}{@{}c@{}}
				$l_2$
				\\
				\psframebox[fillstyle=solid,
				fillcolor=white,framesep=3]{
					$\f$ \psframebox[fillstyle=solid,
					fillcolor=white,framesep=3]{\white $l$}
				} \end{tabular}}}
		}{}
		\psrelpoint{r}{t}{0}{-40}
		\rput(\x{t},\y{t}){\rnode{t}{$\t$}}
		\ncline{->}{t}{r}
		{
		\psrelpoint{r}{v}{44}{-30}
		\rput(\x{v},\y{v}){\rnode{v}{
				\begin{tabular}{@{}c@{}}
				\psframebox[fillstyle=solid,
				fillcolor=white,framesep=3]{
					$\f$ \psframebox[fillstyle=solid,
					fillcolor=white,framesep=3]{\white $l$}
				} 
				\\
				$l_4$
				\end{tabular}}}
		}{}

		\putnode{k}{q}{3}{0}{}
		\putnode{z}{s}{-17}{-8}{}
		\ncline[nodesepA=0,nodesepB=-4]{->}{k}{z}
		\putnode{z}{s}{8}{-22}{}

		\putnode{a}{s}{10}{-8}{}
		\putnode{b}{r}{-13}{-8}{}
		\ncline{->}{a}{b}

		\putnode{a}{r}{10}{-8}{}
		\putnode{b}{v}{-13}{8}{}
		\ncline{->}{a}{b}

%
		\end{pspicture}
		}
}}

\newcommand{\Gexectwo}
{{
{\psset{unit=.28mm}
\psset{linewidth=.2mm,arrowsize=4pt,arrowinset=0}
			\begin{pspicture}(60,-50)(155,17)

		\psrelpoint{origin}{q}{45}{-10}
		\rput(\x{q},\y{q}){\rnode{q}{$\X$}}
		{
		\psrelpoint{q}{r}{30}{8}
		\rput(\x{r},\y{r}){\rnode{r}{
				\begin{tabular}{@{}c@{}}
				$l_2$
				\\
				\psframebox[fillstyle=solid,
				fillcolor=white,framesep=3]{
					$\f$ \psframebox[fillstyle=solid,
					fillcolor=white,framesep=3]{\white $l$}
				} \end{tabular}}}
		}{}

		\psrelpoint{r}{t}{0}{-40}
		\rput(\x{t},\y{t}){\rnode{t}{$\t$}}
		\ncline{->}{t}{r}
		{
		\psrelpoint{r}{v}{44}{-30}
		\rput(\x{v},\y{v}){\rnode{v}{
				\begin{tabular}{@{}c@{}}
				\psframebox[fillstyle=solid,
				fillcolor=white,framesep=3]{
					$\f$ \psframebox[fillstyle=solid,
					fillcolor=white,framesep=3]{\white $l$}
				} 
				\\
				$l_4$
				\end{tabular}}}
		}{}

		\putnode{a}{s}{17}{-8}{}
		\putnode{b}{r}{-13}{-8}{}
		\ncline{->}{a}{b}

		\putnode{a}{r}{10}{-8}{}
		\putnode{b}{v}{-13}{8}{}
		\ncline{->}{a}{b}


		\end{pspicture}
		}
}}

\newcommand{\Gexecthree}
{{
{\psset{unit=.28mm}
\psset{linewidth=.2mm,arrowsize=4pt,arrowinset=0}
			\begin{pspicture}(45,-65)(175,17)

		\psrelpoint{origin}{q}{45}{-10}
		\rput(\x{q},\y{q}){\rnode{q}{$\X$}}
		{
		\psrelpoint{q}{r}{30}{8}
		\rput(\x{r},\y{r}){\rnode{r}{
				\begin{tabular}{@{}c@{}}
				$l_2$
				\\
				\psframebox[fillstyle=solid,
				fillcolor=white,framesep=3]{
					$\f$ \psframebox[fillstyle=solid,
					fillcolor=white,framesep=3]{\white $l$}
				} \end{tabular}}}
		}{}
		{
		\psrelpoint{r}{v}{41}{-45}
		\rput(\x{v},\y{v}){\rnode{v}{
				\begin{tabular}{@{}c@{}}
				\psframebox[fillstyle=solid,
				fillcolor=white,framesep=3]{
					$\f$ \psframebox[fillstyle=solid,
					fillcolor=white,framesep=3]{\white $l$}
				} 
				\\
				$l_4$
				\end{tabular}}}
		}{}
		{
		\psrelpoint{v}{w}{41}{0}
		\rput(\x{w},\y{w}){\rnode{w}{
				\begin{tabular}{@{}c@{}}
				\psframebox[fillstyle=solid,
				fillcolor=white,framesep=3]{
					$\f$ \psframebox[fillstyle=solid,
					fillcolor=white,framesep=3]{\white $l$}
				} 
				\\
				$l_5$
				\end{tabular}}}
		}{}

		\putnode{a}{s}{17}{-8}{}
		\putnode{b}{r}{-13}{-8}{}
		\ncline{->}{a}{b}

		\putnode{a}{r}{10}{-8}{}
		\putnode{b}{v}{-13}{8}{}
		\ncline{->}{a}{b}

		\putnode{a}{v}{10}{8}{}
		\putnode{b}{w}{-13}{8}{}
		\ncline{->}{a}{b}

%
		\end{pspicture}
		}
}}

\newcommand{\Gexecfour}
{{
{\psset{unit=.28mm}
\psset{linewidth=.2mm,arrowsize=4pt,arrowinset=0}
			\begin{pspicture}(80,-65)(175,17)

		\psrelpoint{origin}{q}{82}{-38}
		\rput(\x{q},\y{q}){\rnode{q}{$\X$}}
		{
		\psrelpoint{r}{v}{41}{-45}
		\rput(\x{v},\y{v}){\rnode{v}{
				\begin{tabular}{@{}c@{}}
				\psframebox[fillstyle=solid,
				fillcolor=white,framesep=3]{
					$\f$ \psframebox[fillstyle=solid,
					fillcolor=white,framesep=3]{\white $l$}
				} 
				\\
				$l_4$
				\end{tabular}}}
		}{}
		{
		\psrelpoint{v}{w}{41}{0}
		\rput(\x{w},\y{w}){\rnode{w}{
				\begin{tabular}{@{}c@{}}
				\psframebox[fillstyle=solid,
				fillcolor=white,framesep=3]{
					$\f$ \psframebox[fillstyle=solid,
					fillcolor=white,framesep=3]{\white $l$}
				} 
				\\
				$l_5$
				\end{tabular}}}
		}{}

		\putnode{a}{q}{6}{-1}{}
		\putnode{b}{v}{-13}{8}{}
		\ncline{->}{a}{b}

		\putnode{a}{v}{10}{8}{}
		\putnode{b}{w}{-13}{8}{}
		\ncline{->}{a}{b}

%
		\end{pspicture}
		}
}}

\newcommand{\Gmemone}
{{
{\psset{unit=.28mm}
\psset{linewidth=.2mm,arrowsize=4pt,arrowinset=0}
			\begin{pspicture}(-15,-40)(140,25)

		\psrelpoint{origin}{q}{0}{-20}
		\rput(\x{q},\y{q}){\rnode{q}{$\X$}}
		\psrelpoint{origin}{w}{0}{0}
		\rput(\x{w},\y{w}){\rnode{q}{$\Y$}}
		\psrelpoint{origin}{z}{55}{20}
		\rput(\x{z},\y{z}){\rnode{q}{$\Z$}}

		\psrelpoint{origin}{u}{90}{20}
		\rput(\x{u},\y{u}){\rnode{u}{$\t$}}

		\psrelpoint{q}{s}{35}{0}
		\rput(\x{s},\y{s}){\rnode{s}{
				\begin{tabular}{@{}c@{}}
				\psframebox[fillstyle=solid,
				fillcolor=white,framesep=3]{
					$\f$ \psframebox[fillstyle=solid,
					fillcolor=white,framesep=3]{\white $l$}
				} 
				\\
				$l_1$
				\end{tabular}}}
		{
		\psrelpoint{s}{r}{41}{0}
		\rput(\x{r},\y{r}){\rnode{r}{
				\begin{tabular}{@{}c@{}}
				\psframebox[fillstyle=solid,
				fillcolor=white,framesep=3]{
					$\g$ \psframebox[fillstyle=solid,
					fillcolor=white,framesep=3]{\white $l$}
				}
				\\
				$l_2$
				 \end{tabular}}}
		}{}
		{
		\psrelpoint{r}{t}{41}{0}
		\rput(\x{t},\y{t}){\rnode{t}{
				\begin{tabular}{@{}c@{}}
				\psframebox[fillstyle=solid,
				fillcolor=white,framesep=3]{
					$\ \ $ \psframebox[fillstyle=solid,
					fillcolor=white,framesep=3]{\white $l$}
				} 
				\\
				$l_3$
				\end{tabular}}}
		}{}
		\putnode{k}{q}{7}{0}{}
		\putnode{z}{s}{-17}{0}{}
		\ncline[nodesepA=0,nodesepB=-4]{->}{k}{z}
		\putnode{z}{s}{8}{-22}{}

		\putnode{k}{w}{7}{0}{}
		\putnode{z}{s}{-17}{13}{}
		\ncline[nodesepA=0,nodesepB=-4]{->}{k}{z}
		\putnode{z}{s}{8}{-22}{}

		\putnode{a}{s}{10}{8}{}
		\putnode{b}{r}{-13}{8}{}
		\ncline{->}{a}{b}

		\putnode{a}{r}{10}{8}{}
		\putnode{b}{t}{-13}{8}{}
		\ncline{->}{a}{b}
			
		\putnode{a}{z}{17}{55}{}
		\putnode{b}{s}{40}{20}{}
		\ncline{->}{a}{b}

		\putnode{a}{z}{47}{55}{}
		\putnode{b}{s}{40}{20}{}
		\ncline{->}{a}{b}

		\end{pspicture}
		}
}}

\newcommand{\Gmemfive}
{{
{\psset{unit=.28mm}
\psset{linewidth=.2mm,arrowsize=4pt,arrowinset=0}
			\begin{pspicture}(-15,-25)(140,17)

		\psrelpoint{origin}{q}{0}{-10}
		\rput(\x{q},\y{q}){\rnode{q}{$\X$}}
		\psrelpoint{origin}{w}{0}{10}
		\rput(\x{w},\y{w}){\rnode{q}{$\W$}}

		\psrelpoint{q}{s}{35}{8}
		\rput(\x{s},\y{s}){\rnode{s}{
				\begin{tabular}{@{}c@{}}
				$l_{9}$
				\\
				\psframebox[fillstyle=solid,
				fillcolor=white,framesep=3]{
					$\f$ \psframebox[fillstyle=solid,
					fillcolor=white,framesep=3]{\white $l$}
				} \end{tabular}}}
		{
		\psrelpoint{s}{r}{41}{0}
		\rput(\x{r},\y{r}){\rnode{r}{
				\begin{tabular}{@{}c@{}}
				$l_{10}$
				\\
				\psframebox[fillstyle=solid,
				fillcolor=white,framesep=3]{
					$\g$ \psframebox[fillstyle=solid,
					fillcolor=white,framesep=3]{\white $l$}
				} \end{tabular}}}
		}{}
		{
		\psrelpoint{r}{t}{41}{0}
		\rput(\x{t},\y{t}){\rnode{t}{
				\begin{tabular}{@{}c@{}}
				$l_{11}$
				\\
				\psframebox[fillstyle=solid,
				fillcolor=white,framesep=3]{
					$\ \ $ \psframebox[fillstyle=solid,
					fillcolor=white,framesep=3]{\white $l$}
				} \end{tabular}}}
		}{}
		\putnode{k}{q}{7}{0}{}
		\putnode{z}{s}{-17}{-8}{}
		\ncline[nodesepA=0,nodesepB=-4]{->}{k}{z}
		\putnode{z}{s}{8}{-22}{}

		\putnode{a}{s}{10}{-8}{}
		\putnode{b}{r}{-13}{-8}{}
		\ncline{->}{a}{b}

		\putnode{a}{r}{10}{-8}{}
		\putnode{b}{t}{-13}{-8}{}
		\ncline{->}{a}{b}

		\putnode{k}{w}{7}{0}{}
		\putnode{z}{s}{-17}{3}{}
		\ncline[nodesepA=0,nodesepB=-4]{->}{k}{z}
		\putnode{z}{s}{8}{-22}{}

		\end{pspicture}
		}
}}

\newcommand{\Gmemthree}
{{
{\psset{unit=.28mm}
\psset{linewidth=.2mm,arrowsize=4pt,arrowinset=0}
			\begin{pspicture}(-15,-45)(140,25)

		\psrelpoint{origin}{q}{0}{-20}
		\rput(\x{q},\y{q}){\rnode{q}{$\Y$}}
		\psrelpoint{origin}{z}{75}{20}
		\rput(\x{z},\y{z}){\rnode{q}{$\t$}}
		\psrelpoint{q}{s}{35}{0}
		\rput(\x{s},\y{s}){\rnode{s}{
				\begin{tabular}{@{}c@{}}
				\psframebox[fillstyle=solid,
				fillcolor=white,framesep=3]{
					$\f$ \psframebox[fillstyle=solid,
					fillcolor=white,framesep=3]{\white $l$}
				} 
				\\
				$l_4$
				\end{tabular}}}
		{
		\psrelpoint{s}{r}{41}{0}
		\rput(\x{r},\y{r}){\rnode{r}{
				\begin{tabular}{@{}c@{}}
				\psframebox[fillstyle=solid,
				fillcolor=white,framesep=3]{
					$\g$ \psframebox[fillstyle=solid,
					fillcolor=white,framesep=3]{\white $l$}
				} 
				\\
				$l_5$
				\end{tabular}}}
		}{}
		{
		\psrelpoint{r}{t}{41}{0}
		\rput(\x{t},\y{t}){\rnode{t}{
				\begin{tabular}{@{}c@{}}
				\psframebox[fillstyle=solid,
				fillcolor=white,framesep=3]{
					$\ \ $ \psframebox[fillstyle=solid,
					fillcolor=white,framesep=3]{\white $l$}
				} 
				\\
				$l_6$
				\end{tabular}}}
		}{}
		\putnode{k}{q}{7}{0}{}
		\putnode{z}{s}{-17}{0}{}
		\ncline[nodesepA=0,nodesepB=-4]{->}{k}{z}
		\putnode{z}{s}{8}{-22}{}

		\putnode{a}{s}{10}{8}{}
		\putnode{b}{r}{-13}{8}{}
		\ncline{->}{a}{b}

		\putnode{a}{r}{10}{8}{}
		\putnode{b}{t}{-13}{8}{}
		\ncline{->}{a}{b}

		\putnode{a}{z}{32}{55}{}
		\putnode{b}{s}{40}{20}{}
		\ncline{->}{a}{b}

		\end{pspicture}
		}
}}

\newcommand{\Gmemfour}
{{
{\psset{unit=.28mm}
\psset{linewidth=.2mm,arrowsize=4pt,arrowinset=0}
			\begin{pspicture}(-10,-25)(100,20)

		\psrelpoint{origin}{q}{0}{-10}
		\rput(\x{q},\y{q}){\rnode{q}{$\Z$}}
		\psrelpoint{q}{s}{35}{8}
		\rput(\x{s},\y{s}){\rnode{s}{
				\begin{tabular}{@{}c@{}}
				$l_7$
				\\
				\psframebox[fillstyle=solid,
				fillcolor=white,framesep=3]{
					$\g$ \psframebox[fillstyle=solid,
					fillcolor=white,framesep=3]{\white $l$}
				} \end{tabular}}}
		{
		\psrelpoint{s}{r}{41}{0}
		\rput(\x{r},\y{r}){\rnode{r}{
				\begin{tabular}{@{}c@{}}
				$l_8$
				\\
				\psframebox[fillstyle=solid,
				fillcolor=white,framesep=3]{
					$\ \ $ \psframebox[fillstyle=solid,
					fillcolor=white,framesep=3]{\white $l$}
				} \end{tabular}}}
		}{}
		\putnode{k}{q}{7}{0}{}
		\putnode{z}{s}{-17}{-8}{}
		\ncline[nodesepA=0,nodesepB=-4]{->}{k}{z}
		\putnode{z}{s}{8}{-22}{}

		\putnode{a}{s}{10}{-8}{}
		\putnode{b}{r}{-13}{-8}{}
		\ncline{->}{a}{b}

		\end{pspicture}
		}
}}

\newcommand{\Gmemtwo}
{{
{\psset{unit=.28mm}
\psset{linewidth=.2mm,arrowsize=4pt,arrowinset=0}
			\begin{pspicture}(-15,-40)(60,10)

		\psrelpoint{origin}{q}{0}{-30}
		\rput(\x{q},\y{q}){\rnode{q}{$\X$}}
		\psrelpoint{origin}{w}{0}{-10}
		\rput(\x{w},\y{w}){\rnode{q}{$\W$}}

		\psrelpoint{q}{s}{35}{10}
		\rput(\x{s},\y{s}){\rnode{s}{
				\begin{tabular}{@{}c@{}}
				$l_4$
				\\
				\psframebox[fillstyle=solid,
				fillcolor=white,framesep=3]{
					$\f$ \psframebox[fillstyle=solid,
					fillcolor=white,framesep=3]{\white $l$}
				} 
				\end{tabular}}}
		\putnode{k}{q}{7}{0}{}
		\putnode{z}{s}{-17}{-10}{}
		\ncline[nodesepA=0,nodesepB=-4]{->}{k}{z}
		\putnode{z}{s}{8}{-22}{}

		\putnode{k}{w}{7}{0}{}
		\putnode{z}{s}{-17}{3}{}
		\ncline[nodesepA=0,nodesepB=-4]{->}{k}{z}
		\putnode{z}{s}{8}{-22}{}

		\end{pspicture}
		}
}}

\newcommand{\Gmemzg}
{{
{\psset{unit=.28mm}
\psset{linewidth=.2mm,arrowsize=4pt,arrowinset=0}
			\begin{pspicture}(-10,-25)(100,20)

		\psrelpoint{origin}{q}{0}{-10}
		\rput(\x{q},\y{q}){\rnode{q}{$\Z$}}
		\psrelpoint{q}{s}{35}{8}
		\rput(\x{s},\y{s}){\rnode{s}{
				\begin{tabular}{@{}c@{}}
				$l_3$
				\\
				\psframebox[fillstyle=solid,
				fillcolor=white,framesep=3]{
					$\g$ \psframebox[fillstyle=solid,
					fillcolor=white,framesep=3]{\white $l$}
				} \end{tabular}}}
		{
		\psrelpoint{s}{r}{41}{0}
		\rput(\x{r},\y{r}){\rnode{r}{
				\begin{tabular}{@{}c@{}}
				$l_4$
				\\
				\psframebox[fillstyle=solid,
				fillcolor=white,framesep=3]{
					$\ \ $ \psframebox[fillstyle=solid,
					fillcolor=white,framesep=3]{\white $l$}
				} \end{tabular}}}
		}{}
		\putnode{k}{q}{7}{0}{}
		\putnode{z}{s}{-17}{-8}{}
		\ncline[nodesepA=0,nodesepB=-4]{->}{k}{z}
		\putnode{z}{s}{8}{-22}{}

		\putnode{a}{s}{10}{-8}{}
		\putnode{b}{r}{-13}{-8}{}
		\ncline{->}{a}{b}

		\end{pspicture}
		}
}}
\newcommand{\Gmemxy}
{{
{\psset{unit=.28mm}
\psset{linewidth=.2mm,arrowsize=4pt,arrowinset=0}
			\begin{pspicture}(-5,-40)(55,10)

		\psrelpoint{origin}{q}{0}{-30}
		\rput(\x{q},\y{q}){\rnode{q}{$\X$}}
		\psrelpoint{origin}{w}{0}{-10}
		\rput(\x{w},\y{w}){\rnode{q}{$\Y$}}

		\psrelpoint{q}{s}{35}{10}
		\rput(\x{s},\y{s}){\rnode{s}{
				\begin{tabular}{@{}c@{}}
				$l_1$
				\\
				\psframebox[fillstyle=solid,
				fillcolor=white,framesep=3]{
					$\ \ $ \psframebox[fillstyle=solid,
					fillcolor=white,framesep=3]{\white $l$}
				} 
				\end{tabular}}}
		\putnode{k}{q}{7}{0}{}
		\putnode{z}{s}{-17}{-10}{}
		\ncline[nodesepA=0,nodesepB=-4]{->}{k}{z}
		\putnode{z}{s}{8}{-22}{}

		\putnode{k}{w}{7}{0}{}
		\putnode{z}{s}{-17}{3}{}
		\ncline[nodesepA=0,nodesepB=-4]{->}{k}{z}
		\putnode{z}{s}{8}{-22}{}

		\end{pspicture}
		}
}}
\newcommand{\Gmemxyf}
{{
{\psset{unit=.28mm}
\psset{linewidth=.2mm,arrowsize=4pt,arrowinset=0}
			\begin{pspicture}(-5,-40)(100,25)

		\psrelpoint{origin}{q}{0}{-20}
		\rput(\x{q},\y{q}){\rnode{q}{$\X$}}
		\psrelpoint{origin}{w}{0}{0}
		\rput(\x{w},\y{w}){\rnode{q}{$\Y$}}
		\psrelpoint{origin}{z}{70}{20}
		\rput(\x{z},\y{z}){\rnode{q}{$\v$}}

		\psrelpoint{q}{s}{35}{0}
		\rput(\x{s},\y{s}){\rnode{s}{
				\begin{tabular}{@{}c@{}}
				\psframebox[fillstyle=solid,
				fillcolor=white,framesep=3]{
					$\f$ \psframebox[fillstyle=solid,
					fillcolor=white,framesep=3]{\white $l$}
				} 
				\\
				$l_1$
				\end{tabular}}}
		{
		\psrelpoint{s}{r}{41}{0}
		\rput(\x{r},\y{r}){\rnode{r}{
				\begin{tabular}{@{}c@{}}
				\psframebox[fillstyle=solid,
				fillcolor=white,framesep=3]{
					$\ \ $ \psframebox[fillstyle=solid,
					fillcolor=white,framesep=3]{\white $l$}
				}
				\\
				$l_2$
				 \end{tabular}}}
		}{}
		\putnode{k}{q}{7}{0}{}
		\putnode{z}{s}{-17}{0}{}
		\ncline[nodesepA=0,nodesepB=-4]{->}{k}{z}
		\putnode{z}{s}{8}{-22}{}

		\putnode{k}{w}{7}{0}{}
		\putnode{z}{s}{-17}{13}{}
		\ncline[nodesepA=0,nodesepB=-4]{->}{k}{z}
		\putnode{z}{s}{8}{-22}{}

		\putnode{a}{s}{10}{8}{}
		\putnode{b}{r}{-13}{8}{}
		\ncline{->}{a}{b}

		\putnode{a}{z}{32}{60}{}
		\putnode{b}{s}{40}{20}{}
		\ncline{->}{a}{b}

		\end{pspicture}
		}
}}
\newcommand{\Gmemv}
{{
{\psset{unit=.28mm}
\psset{linewidth=.2mm,arrowsize=4pt,arrowinset=0}
			\begin{pspicture}(0,-40)(55,10)

		\psrelpoint{origin}{q}{0}{-30}
		\rput(\x{q},\y{q}){\rnode{q}{$\v$}}

		\psrelpoint{q}{s}{35}{10}
		\rput(\x{s},\y{s}){\rnode{s}{
				\begin{tabular}{@{}c@{}}
				$l_2$
				\\
				\psframebox[fillstyle=solid,
				fillcolor=white,framesep=3]{
					$\ \ $ \psframebox[fillstyle=solid,
					fillcolor=white,framesep=3]{\white $l$}
				} 
				\end{tabular}}}
		\putnode{k}{q}{7}{0}{}
		\putnode{z}{s}{-17}{-10}{}
		\ncline[nodesepA=0,nodesepB=-4]{->}{k}{z}
		\putnode{z}{s}{8}{-22}{}

		\end{pspicture}
		}
}}
\newcommand{\Gmemzwg}
{{
{\psset{unit=.28mm}
\psset{linewidth=.2mm,arrowsize=4pt,arrowinset=0}
			\begin{pspicture}(-10,-25)(100,20)

		\psrelpoint{origin}{q}{0}{-10}
		\rput(\x{q},\y{q}){\rnode{q}{$\Z$}}
		\psrelpoint{origin}{w}{0}{10}
		\rput(\x{w},\y{w}){\rnode{q}{$\W$}}
		\psrelpoint{q}{s}{35}{8}
		\rput(\x{s},\y{s}){\rnode{s}{
				\begin{tabular}{@{}c@{}}
				$l_3$
				\\
				\psframebox[fillstyle=solid,
				fillcolor=white,framesep=3]{
					$\g$ \psframebox[fillstyle=solid,
					fillcolor=white,framesep=3]{\white $l$}
				} \end{tabular}}}
		{
		\psrelpoint{s}{r}{41}{0}
		\rput(\x{r},\y{r}){\rnode{r}{
				\begin{tabular}{@{}c@{}}
				$l_4$
				\\
				\psframebox[fillstyle=solid,
				fillcolor=white,framesep=3]{
					$\ \ $ \psframebox[fillstyle=solid,
					fillcolor=white,framesep=3]{\white $l$}
				} \end{tabular}}}
		}{}
		\putnode{k}{q}{7}{0}{}
		\putnode{z}{s}{-17}{-8}{}
		\ncline[nodesepA=0,nodesepB=-4]{->}{k}{z}
		\putnode{z}{s}{8}{-22}{}

		\putnode{a}{s}{10}{-8}{}
		\putnode{b}{r}{-13}{-8}{}
		\ncline{->}{a}{b}

		\putnode{k}{w}{7}{0}{}
		\putnode{z}{s}{-17}{3}{}
		\ncline[nodesepA=0,nodesepB=-4]{->}{k}{z}
		\putnode{z}{s}{8}{-22}{}

		\end{pspicture}
		}
}}

\newcommand{\Gmemleft}
{{
{\psset{unit=.28mm}
\psset{linewidth=.2mm,arrowsize=4pt,arrowinset=0}
			\begin{pspicture}(-15,-25)(140,17)

		\psrelpoint{origin}{q}{0}{-10}
		\rput(\x{q},\y{q}){\rnode{q}{$\X$}}
		\psrelpoint{q}{s}{35}{8}
		\rput(\x{s},\y{s}){\rnode{s}{
				\begin{tabular}{@{}c@{}}
				$l_1$
				\\
				\psframebox[fillstyle=solid,
				fillcolor=white,framesep=3]{
					$\f$ \psframebox[fillstyle=solid,
					fillcolor=white,framesep=3]{\white $l$}
				} \end{tabular}}}
		{
		\psrelpoint{s}{r}{41}{0}
		\rput(\x{r},\y{r}){\rnode{r}{
				\begin{tabular}{@{}c@{}}
				$l_2$
				\\
				\psframebox[fillstyle=solid,
				fillcolor=white,framesep=3]{
					$\g$ \psframebox[fillstyle=solid,
					fillcolor=white,framesep=3]{\white $l$}
				} \end{tabular}}}
		}{}
		{
		\psrelpoint{r}{t}{41}{0}
		\rput(\x{t},\y{t}){\rnode{t}{
				\begin{tabular}{@{}c@{}}
				$l_3$
				\\
				\psframebox[fillstyle=solid,
				fillcolor=white,framesep=3]{
					$\ \ $ \psframebox[fillstyle=solid,
					fillcolor=white,framesep=3]{\white $l$}
				} \end{tabular}}}
		}{}
		\putnode{k}{q}{7}{0}{}
		\putnode{z}{s}{-17}{-8}{}
		\ncline[nodesepA=0,nodesepB=-4]{->}{k}{z}
		\putnode{z}{s}{8}{-22}{}

		\putnode{a}{s}{10}{-8}{}
		\putnode{b}{r}{-13}{-8}{}
		\ncline{->}{a}{b}

		\putnode{a}{r}{10}{-8}{}
		\putnode{b}{t}{-13}{-8}{}
		\ncline{->}{a}{b}

		\end{pspicture}
		}
}}

\newcommand{\Gmemright}
{{
{\psset{unit=.28mm}
\psset{linewidth=.2mm,arrowsize=4pt,arrowinset=0}
			\begin{pspicture}(-15,-25)(140,17)

		\psrelpoint{origin}{q}{0}{-10}
		\rput(\x{q},\y{q}){\rnode{q}{$\X$}}
		\psrelpoint{q}{s}{35}{8}
		\rput(\x{s},\y{s}){\rnode{s}{
				\begin{tabular}{@{}c@{}}
				$l_4$
				\\
				\psframebox[fillstyle=solid,
				fillcolor=white,framesep=3]{
					$\h$ \psframebox[fillstyle=solid,
					fillcolor=white,framesep=3]{\white $l$}
				} \end{tabular}}}
		{
		\psrelpoint{s}{r}{41}{0}
		\rput(\x{r},\y{r}){\rnode{r}{
				\begin{tabular}{@{}c@{}}
				$l_5$
				\\
				\psframebox[fillstyle=solid,
				fillcolor=white,framesep=3]{
					$\f$ \psframebox[fillstyle=solid,
					fillcolor=white,framesep=3]{\white $l$}
				} \end{tabular}}}
		}{}
		{
		\psrelpoint{r}{t}{41}{0}
		\rput(\x{t},\y{t}){\rnode{t}{
				\begin{tabular}{@{}c@{}}
				$l_6$
				\\
				\psframebox[fillstyle=solid,
				fillcolor=white,framesep=3]{
					$\ \ $ \psframebox[fillstyle=solid,
					fillcolor=white,framesep=3]{\white $l$}
				} \end{tabular}}}
		}{}
		\putnode{k}{q}{7}{0}{}
		\putnode{z}{s}{-17}{-8}{}
		\ncline[nodesepA=0,nodesepB=-4]{->}{k}{z}
		\putnode{z}{s}{8}{-22}{}

		\putnode{a}{s}{10}{-8}{}
		\putnode{b}{r}{-13}{-8}{}
		\ncline{->}{a}{b}

		\putnode{a}{r}{10}{-8}{}
		\putnode{b}{t}{-13}{-8}{}
		\ncline{->}{a}{b}

		\end{pspicture}
		}
}}

\newcommand{\soundfig}
{
\renewcommand{\arraystretch}{.6}
  \scalebox{.9}
  {
    \psset{unit=.75mm}
    \psset{linewidth=.3mm}
    \psset{linewidth=.3mm,arrowsize=5pt,arrowinset=0}
    \begin{pspicture}(-20,-60)(100,70)

    \putnode{f}{origin}{40}{60}{\START \psframebox{\begin{tabular}{@{}l@{}}$\quad {\white f} \ \ \ {\white f} \quad$\end{tabular}}\white\ \START}
    \putnode{a}{f}{0}{-28}{{\white $\q_b$} $--$ $\q_b$}
    \putnode{c}{a}{25}{-28}{{\white $\q_i$} $--$ $\q_i$}
    \putnode{d}{a}{0}{-56}{{\white $\q_a$} $--$ $\q_a$}
    \putnode{e}{d}{0}{-28}{$--$}

    \nccurve[angleA=270,angleB=90,loopsize=20,linearc=2,offset=0,arm=5]{->}{a}{b}
    \nccurve[angleA=270,angleB=90,loopsize=-20,linearc=2,offset=0,arm=5]{->}{b}{d}
    \nccurve[angleA=270,angleB=90,loopsize=-20,linearc=2,offset=0,arm=5]{->}{a}{c}
    \nccurve[angleA=270,angleB=90,loopsize=20,linearc=2,offset=0,arm=5]{->}{c}{d}
        \ncline{->}{f}{a}
        \ncline{->}{d}{e}

    \putnode{p}{a}{-52}{33}{${\mathbf --}$}
    \putnode{q}{p}{0}{-117}{${\mathbf --}$}
        \ncline[linewidth=1mm,linestyle=dotted]{-}{p}{q}
        \ncput*{\large $\pi$}

    \putnode{p}{a}{-32}{0}{${\mathbf --}$}
    \putnode{q}{d}{-32}{0}{${\mathbf --}$}
        \ncline[linewidth=1mm,linestyle=dotted]{-}{p}{q}
        \ncput*{\large $\pi'$}

    \putnode{p}{a}{57}{0}{${\mathbf --}$}
    \putnode{q}{c}{32}{0}{${\mathbf --}$}
        \ncline[linewidth=1mm,linestyle=dotted]{-}{p}{q}
        \ncput*{\large $\pi_1$}

    \putnode{p}{c}{32}{0}{${\mathbf --}$}
    \putnode{q}{d}{57}{0}{${\mathbf --}$}
        \ncline[linewidth=1mm,linestyle=dotted]{-}{p}{q}
        \ncput*{\large $\pi_2$}

   \end{pspicture}
   }
}

\section{Introduction}
%
Heap memory management is non-trivial and error-prone. Poorly written programs
tend to run out of memory leading to failures. Heap liveness analysis
identifies usages of heap locations.  It can be used for identifying possible
memory leaks, suggesting better use of stack in place of heap, identifying heap
memory locations and links that can be reclaimed when they are not required any
longer in a program, improving cache behaviour of programs etc. 


\begin{figure}
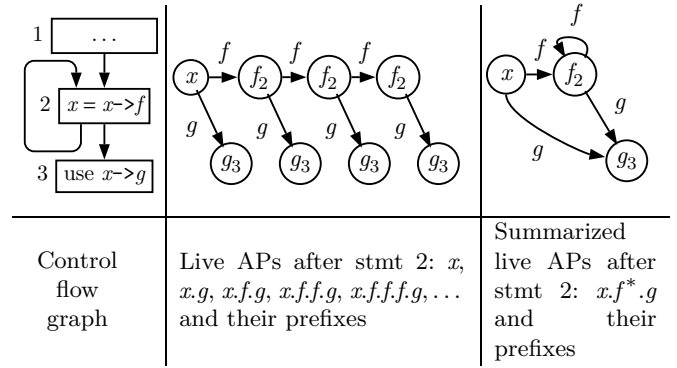

\begin{tabular}{@{}m{1.85cm}|@{}m{4cm}|@{}m{2.3cm}}
\introcfg & \Gfull & \Gcompact \\ \hline
\small
\begin{tabular}{@{\quad}c} \small Control \\ \small flow \\ \small graph \end{tabular}
&
\multicolumn{1}{m{3.8cm}|}{\small Live APs
after stmt 2: $\X,$ $\X.\g,$ $\X.\f.\g,$ $\X.\f.\f.\g,$
$\X.\f.\f.\f.\g, \dots$ and their prefixes}
&
\multicolumn{1}{m{2.15cm}}{\small 
Summarized live APs after stmt 2:
$\X.(\f\,)^*.\g$ and their prefixes} \\
\end{tabular}

\caption{\small Summarization of infinite number of live
access paths (APs) in the existing abstraction. Each node's field label
is same as its in-edge label.}
\label{fig:summ-simple}
\end{figure}

Heap liveness information can be represented using a store-based
model~\cite{kk16,kanvar2017s} of memory by constructing a flow-sensitive {\em live
points-to graph}~\cite{amss06}.  At a given program point $\q$, its nodes
denote concrete memory locations and edges denote links between the
locations such that the links are likely to be used before being re-defined
along the control flow paths starting at $\q$.  For this model, a finite
representation of live memory can be created by abstracting concrete heap
locations using allocation-sites~\cite{cwz90}.  However, this
is a very coarse representation because even large programs have a small number
of allocation-sites (see our measurements for SPEC CPU2006 benchmarks in
Figure~\ref{fig:bench}).  This makes very few distinctions
between heap locations making it difficult to identify allocated heap locations
that are not used.

A more precise representation of live heap can be constructed when it is viewed
under a storeless model~\cite{kk16,kanvar2017s} by naming locations in terms of {\em live
access paths}~\cite{ksk07} at each program point.  An access path is a variable
in $\V$ followed by a sequence of fields in $\F$. Thus, $\V \times \F^*$
represents the universal set of access paths. Variables and fields denote memory links. An
access path is marked as live at program point $\q$, if memory link(s)
corresponding to its last field are likely to  be used before being defined
along some control flow path starting at $\q$.  For example, access path $\Y$
in this abstraction indicates that variable $\Y$ is dereferenced in the program,
and access path $\Y.\f.\g$ suggests that the last field $\g$ is
dereferenced from the location pointed to by field $\f$ of the location pointed to by
variable $\Y$.

Due to loops and recursion, the set of access paths may be infinite
and lengths of access paths may be unbounded. This warrants 
summarization.  We ignore the adhoc summarization techniques such as
$k$-limiting~\cite{kk16} (Section~\ref{sec:rel}) and focus on \emph{use-site} based
summarization~\cite{ksk07}: every use of a field name $\f \in \F$ is labelled with a
statement number $\s \in \SS$; the labelled fields and their statement
numbers, represented as $\f_\s$, are treated as unique across the program across all access
paths.  This is illustrated below.

\begin{myexample}
Due to the presence of a loop in the program in Figure~\ref{fig:summ-simple},
field $\f$ appears an unbounded number of times in an access path rooted at
variable \X; there is one occurrence of \g after every occurrence of \f.  All
these are represented as a graph whose nodes are labeled with fields and their
use-sites, $\f_2$ and $\g_3$. Nodes with the same label are merged to obtain a
finite automaton, which precisely represents that the access paths denoted by
regular expression $\X.(\f\,)^*.\g$ are live after statement 2.
\hfill $\Box$
\end{myexample}

This abstraction has the following two limitations:
\begin{itemize}
\item In some cases, the resulting graph is imprecise in that it admits spurious access
      paths thereby over-approximating heap liveness.
\item In some cases, the resulting graph is larger than the graph required
      to admit the same set of access paths.
\end{itemize}

Section~\ref{sec:motiv} explains these limitations and our key ideas to
overcome them.  Sections~\ref{sec:abstract} and~\ref{sec:formulations}
formalize our ideas.  Section~\ref{sec:inter}  describes how we lift our
analysis to interprocedural level.  Section~\ref{sec:impl} describes our
implementations and measurements.  Section~\ref{sec:rel} presents the related
work.  Section~\ref{sec:concl} concludes the paper.
Appendix~\ref{sec:sound-live} proves soundness of the proposed heap liveness
analysis. Appendix~\ref{sec:repr} identifies features in programs that show
precision benefit (in terms of garbage collection of locations).

\section{Our Key Ideas and Contributions}
\label{sec:motiv}

This section explains the limitations of the
existing technique
and how it can be improved.  Let $\SS$ denote the set of program statements (or
sites). $\Q  = \{\Inn, \Outn\} \times \SS$ is a set of program points. For 
statement $\s \in \SS$, $\Inn_\s$
and $\Outn_\s$ denote entry of $\s$ (i.e. program point immediately before
$\s$) and exit of $\s$ (i.e. program point immediately after
$\s$). For simplicity, we will assume that only heap
locations have fields and the variables are not addressable. These assumptions
are relaxed in Section~\ref{sec:addr}.

\begin{figure}
\begin{tabular}{@{}m{2.2cm}m{6.4cm}}
\Gsumm &
\begin{tabular}{@{}c@{}}
	\begin{tabular}{@{}c@{}}\Gexec \\ \small Possible execution snapshot at $\Outn_1$ or $\Inn_2$ \end{tabular} \\
	\begin{tabular}{@{}c@{}}\Gexecone \\ \small Possible execution snapshot at $\Outn_3$ or $\Inn_4$ \end{tabular} \\
	\begin{tabular}{@{}c@{}}\Gexectwo \\ \small Possible execution snapshot at $\Outn_4$ \end{tabular} \\
\end{tabular}
\end{tabular}
\small
\begin{tabular}{|c|c|c|c|}
\hline
\cellcolor{lightgray}
Point &\cellcolor{lightgray} Existing Non- & \cellcolor{lightgray}Proposed \\
	\cellcolor{lightgray}& \cellcolor{lightgray}Deterministic Graph & \cellcolor{lightgray}Deterministic Graph \\ \hline \hline
$\Inn_4$ 
	& \begin{tabular}{@{}c@{}}\Goldthree \\ $\equiv$ \\ \Goldthreefinal \end{tabular} & \Gnewthree \\ \hline
$\Inn_2$
	& \Goldthreefinal & \Gnewone \\ \hline
\end{tabular}
\caption{\small Example. Execution snapshots and fixpoint computations of
backward existing and proposed liveness analyses are shown. Each live node
represents live access path(s) reaching it. Each node's field label
is same as its in-edge label. For simplicity, aliases of the
access paths are not included. Step by step working is shown
in Figure~\ref{fig:summ-steps}}
\label{fig:summ}
\end{figure}
\begin{figure}
\small
\begin{tabular}{@{}m{3.2cm}@{}m{5.5cm}}
\Galias & 
\begin{tabular}{@{}c@{}}
	\Gmemone \\
	\small Possible execution snapshot at $\Outn_7$ \\
	\small for the left control flow path \\ \hline
	\Gmemthree \\[-2ex]
	\Gmemfour \\
	\Gmemfive \\
	\small Possible execution snapshot at $\Outn_7$ \\
	\small for the right control flow path\\
\end{tabular}
\end{tabular}
\small
\begin{tabular}{|l|l|l|l|}
\hline 
\cellcolor{lightgray}
Point & \cellcolor{lightgray}Existing Greedy & \multicolumn{2}{l|}{\cellcolor{lightgray}Minimal Liveness Analysis} \\
\cline{3-4}
 \cellcolor{lightgray}& \cellcolor{lightgray}Liveness Analysis & \cellcolor{lightgray}Phase 1 & \cellcolor{lightgray}Phase 2 \\ \hline \hline
$\Outn_7$ & $\emptyset$
	& $\emptyset$
	& $\emptyset$ \\ \hline
$\Inn_6$ & 
	\begin{tabular}{@{}l@{}} $\{\Y.\f.\g,\Y.\f,\Y,$ \\ $\X.\f.\g,\X.\f,$ \\ $\W.\f.\g,\W.\f,$ \\ $\Z.\g\}$ \end{tabular}
	& \begin{tabular}{@{}l@{}} $\{\Y.\f.\g,\Y.\f,\Y\}$ \end{tabular}
	& \begin{tabular}{@{}l@{}} $\{\Y.\f.\g,\Y.\f,\Y,$ \\ $\X.\f.\g,\X.\f,$ \\ $\Z.\g\}$ \end{tabular}
	\\ \hline
$\Outn_4$ & 
	\begin{tabular}{@{}l@{}} $\{\Y.\f.\g,\Y.\f,\Y,$ \\ $\X.\f.\g,\X.\f,$ \\ $\W.\f.\g,\W.\f,$ \\ $\Z.\g\}$ \end{tabular}
	& \begin{tabular}{@{}l@{}} $\{\Y.\f.\g,\Y.\f,\Y\}$ \end{tabular}
	&\begin{tabular}{@{}l@{}} $\{\Y.\f.\g,\Y.\f,\Y\}$ \end{tabular}
	\\ \hline
$\Outn_3$ & 
	\begin{tabular}{@{}l@{}} $\{\Y.\f.\g,\Y.\f,\Y,$ \\ $\X.\f.\g,\X.\f,$ \\ $\W.\f.\g,\W.\f,$ \\ $\Z.\g\}$ \end{tabular}
	& \begin{tabular}{@{}l@{}} $\{\Y.\f.\g,\Y.\f,\Y\}$ \end{tabular}
	& \begin{tabular}{@{}l@{}} $\{\Y.\f.\g,\Y.\f,\Y,$ \\ $\X.\f.\g,\X.\f,$ \\ $\Z.\g\}$ \end{tabular}
	\\ \hline
$\Inn_3$ & 
	\begin{tabular}{@{}l@{}} $\{\Y,\X,\Z.\g,\Z,$  \\ $\W.\f.\g,\W.\f\}$ \end{tabular}
	& \begin{tabular}{@{}l@{}} $\{\Y,\X,\Z.\g,\Z\}$ \end{tabular}
	& \begin{tabular}{@{}l@{}} $\{\Y,\X,\Z.\g,\Z\}$ \end{tabular}
	\\ \hline
\end{tabular}
\caption{\small Example. A live access path represents that the memory link
corresponding to its last field is live.}
\label{fig:link}
\end{figure}

The existing representation~\cite{ksk07} of live access paths uses a
non-deterministic automaton, which admits multiple transition sequences for a
given access path.  To represent each live access path by only one sequence of
transitions, we propose to use a deterministic graph.  Note that we do not
merely create a deterministic version of the final non-deterministic graph but
construct a deterministic graph in each step.  This requires us to store {\em
sets of use-sites} rather than a single use-site per field name. The example
below illustrates that using sets of use-sites brings precision. 

\begin{myexample}
Assume that the program in Figure~\ref{fig:summ} allocates heap locations in
statement 1.  At $\Inn_2$, field $\f$ of location $l_1$ is live because it will
get used in statement 2.  The rest of the links are dead because they are made
unreachable in statement 3.  The existing non-deterministic liveness graph in the row for
$\Inn_2$, denotes that the access paths $\X.(\f\,)^+.\f$
 and its prefixes are live.  This is spurious. Our proposed deterministic
liveness graph in the row for $\Inn_2$ precisely shows that only $\X.\f$ and its
prefixes are live. 
\hfill $\Box$
\end{myexample}

For soundness, the existing technique includes aliases\footnote{The
existing technique~\cite{ksk07} includes alias closure. This introduces further
over-approximation because may aliases are not transitive but the closure
computation treats them as transitive.} of live access paths during backward
liveness analysis. We call this approach {\em greedy liveness analysis} because
it greedily includes all the aliases 
and propagates these in the backward analysis, irrespective of
their need in the analysis of a statement.  This
backward propagation of live aliased access paths 
leads to imprecision. We propose to perform a two-phase analysis called {\em
minimal liveness analysis}. Phase 1 uses aliases of live access paths minimally
and phase 2 generates all aliases of live access paths without propagating them
backwards in the control flow path. Below we illustrate how our
minimal liveness analysis is more precise than greedy liveness
analysis.

\begin{myexample}
In the program in Figure~\ref{fig:link}, statements 6 and 7 are equivalent to
``$\text{use }\Y \texttt{->} \f \texttt{->} \g$''. Therefore, using the snapshot of the
right control flow path, we see that link $\f$ of location $l_{9}$ and link
$\g$ of location $l_{10}$ are dead at $\Outn_4$ as they are not used in
statements 5, 6, 7. 

Existing analysis marks access path $\Y.\f.\g$ and its prefixes as live at
$\Inn_6$.  Since $\X$ and $\Y$ are aliased, the existing greedy liveness
analysis marks the aliases i.e. $\X.\f.\g$ and $\X.\f$ also as live at $\Inn_6$.
These access paths are propagated backwards along the two control flow paths
thereby imprecisely marking them as live at $\Outn_4$ also.

In phase 1 of our proposed analysis, access paths $\Y.\f.\g$ and its prefixes
are marked as live at $\Inn_6$. However, their aliases, viz. $\X.\f.\g$ and
$\X.\f$ are not marked as live. Thus, when the live access paths
are propagated backwards along the two control flow paths in phase 1, no
spurious live access path gets propagated to $\Outn_4$.  Note that aliases
are minimally used when statement 3 is analysed in order to treat LHS $\X \texttt{->}
\f$ of the statement as $\Y \texttt{->} \f$.  Phase 2 of our proposed analysis generates 
aliases $\X.\f.\g$ and $\X.\f$ of live access paths at $\Inn_6$. These are not
propagated backwards in phase 2. At $\Outn_4$, since $\Y.\f.\g$ and its prefixes do not
have any alias, no spurious liveness is generated at $\Outn_4$.
\hfill $\Box$
\end{myexample}

Our main contributions are:
\begin{inparaenum}[\em (i)]
\item showing how heap liveness information summarized with use-sites can be
represented as a deterministic automaton 
(Section~\ref{sec:abstract}), \item showing how inclusion of aliases of live
access paths can be minimized (Section~\ref{sec:formulations}),
\item showing which heap liveness information can be bypassed from the callees
and which can be memoized in the callee (Section~\ref{sec:inter}).
\end{inparaenum}
These improve scalability and efficiency of the analysis, and reduce the amount
of computed liveness information of SPEC
CPU2006.

\section{Representations for Heap Liveness}
\label{sec:abstract}
A set of live access paths in the memory graph
forms a regular language, represented as a finite automaton. 
Here, we describe two representations for storing heap
live access paths and compare their precision using examples explained
intuitively. Section~\ref{sec:formulations} presents the formal
details of computation of the representations.

To avoid handling of statements with multiple occurrences of the same field
name, we normalize our statements such that there is at most one field
name in each statement. For example, statement ``$\text{use } \Y \texttt{->}\f
\texttt{->}\f$'' is normalized to $\t=\Y \texttt{->} \f; \text{use } \t
\texttt{->} \f$.

\subsection{Existing Non-Deterministic Graph}
\label{sec:conv-abstr}
Conventionally~\cite{ksk07}, the set of live access paths at program point $\q
\in \Q$
is represented by a non-deterministic graph $\widehat{\G}_\q = \langle
\hat{\V}_\q \cup \hat{\N}_\q, \hat{\E}_\q \rangle$. Here node set $\hat{\V}_\q
\cup \hat{\N}_\q$ is a discriminated union of $\hat{\V}_\q \subseteq \V$, which
is a set of variables, and $\hat{\N}_\q$, which is a set of field nodes,
uniquely identified by a field and its use-site i.e., $\hat{\N}_\q
\subseteq \F \times \SS$.  In other words, we identify each node with $\f_\s$
where $\f \in \F$ and $\s \in \SS$.

Field nodes represented by the same $\f_\s$ across access paths, or
within an access path are merged at a
program point.  
The intuition behind this merging is that fields with a
common use-site are likely to be followed by the same sequences of fields.
Hence their different occurrences across access paths, or within an access path, can
be treated alike.
In other words, access path construction can be assumed to reach the same state
whenever a field node represented by a use-site is included in an access path. Hence,
${\hat{\V}_\q}$ and $\hat{\N}_\q$ represent states of the liveness graph (or automaton),
which determines if access paths are accepted as live. 
For simplicity, we 
assume that all nodes
in $\hat{\V}_\q \cup \hat{\N}_\q$ are accepting states. Set $\hat{\V}_\q$ is a set of initial
states of the automaton reachable by the variables in $\hat{\V}_\q$.

Edge set $\hat{\E}_\q = \hat{\E\V}_\q \cup
\hat{\E\F}_\q$ is a discriminated union of
\begin{itemize}
\item $\hat{\E\V}_\q\subseteq\hat{\V}_\q \times \hat{\N}_\q$ which is a set of edges
from variables in $\hat{\V}_\q$
\item $\hat{\E\F}_\q \subseteq \hat{\N}_\q \times \F \times \hat{\N}_\q$ which is a
set of edges from nodes in $\hat{\N}_\q$ labeled with fields in $\F$.
\end{itemize}


Below we illustrate how the existing abstraction may generate spurious
cycles in the liveness graph.
\begin{myexample}
\label{ex:summ-conv-2}
At $\Inn_4$ of the program in Figure~\ref{fig:summ}, links $\f$ of only
locations $l_1$ and $l_2$ are live. Here the first link will be used by
statement 4, and the second link will possibly be used by statements 2 and 4 of
next loop iteration.  Thus, only access path $\X.\f.\f$ and its
prefixes are live at $\Inn_4$. 

Observe that, even though both $\f$ of $l_1$ and $\f$ of $l_2$ will
used by statement 4, $\f$ of $l_2$ will additionally be used by statement 2
followed by killing of the field after $\f$ of $l_2$ via variable $\t$ in
statement 3.  Basically, the two links corresponding to field $\f$ go through different
changes and are followed by different sequences of fields because they are used
by different statements. In spite of this, the existing technique merges the
two corresponding field nodes identified as $\f_4$ in its non-deterministic
representation thereby creating a spurious cycle around the field node
$\f_4$ (as
shown in the row corresponding to $\Inn_4$).  This spuriously indicates that
infinite number of fields are live at $\Inn_4$.
\hfill $\Box$
\end{myexample}

The example below illustrates how the existing abstraction causes imprecise
merging of access paths in a program even without loops.

\begin{myexample}
\label{ex:merge}
In the program in Figure~\ref{fig:merge}, statements 5 and 6 are equivalent to
``$\text{use }\X \texttt{->} \f \texttt{->} \g$''. Thus, at $\Inn_5$, access
path $\X.\f.\g$ and its prefixes are live. Its existing abstraction is shown
in the row corresponding to $\Inn_5$ where nodes pointed to by $\f$ and $\g$ are
identified as $\f_5$ and $\g_6$. Statements 3 and 4 are
equivalent to ``$\X \texttt{->} \f \texttt{->} \g = \texttt{new}$''. Thus, at
$\Inn_3$, $\X.\f.\g$ is dead. Since $\X.\f$ and its prefixes are live
at $\Outn_2$, access path $\X.\h.\f$ and its prefixes are live at $\Inn_2$.
Overall, at $\Outn_1$, $\X.\f.\g$ from the left control flow path and $\X.\h.\f$
from the right control flow path are live.

Its existing abstraction is shown in the row corresponding to $\Outn_1$.
Here node pointed to by $\h$ is identified as $\h_2$, node
pointed to by $\f$ along $\X.\f.\g$ is identified as $\f_3$, and nodes
pointed to by $\f$ along $\X.\h.\f$ are identified as $\f_3$ and $\f_5$ in
the non-deterministic graph because at $\Outn_1$, link $\f$ of $l_1$ (i.e. along
$\X.\f.\g$) will be used by statement 5 and link $\f$ of $l_5$ (i.e. along
$\X.\h.\f$) will be used by both statements 3 and 5. Here statement 3 leads to
killing of field after $\f$ of $l_5$. Basically, $\f$ of $l_1$ and $l_2$
go through different changes and are followed by different sequences of fields
since they are used by different statements.  In spite of this, the
existing technique merges the corresponding two field nodes $\f_5$ (along
$\X.\h.\f$ and $\X.\f.\g$) thereby creating a spurious
access path $\X.\h.\f.\g$ at $\Outn_1$.
\hfill $\Box$
\end{myexample}

Examples~\ref{ex:summ-conv-2} and~\ref{ex:merge} show that applying use-sites
may lead to spurious cycles and spurious access paths. This happens because the
merged nodes do not necessarily have the same sequences of fields following them
even though their use-site is identical. The non-deterministic
representation does not account for all possible set of use-sites of a field
while merging field nodes. This brings us to our proposed abstraction which
labels field nodes with the set of all use-sites rather than a single use-site.


\begin{figure}
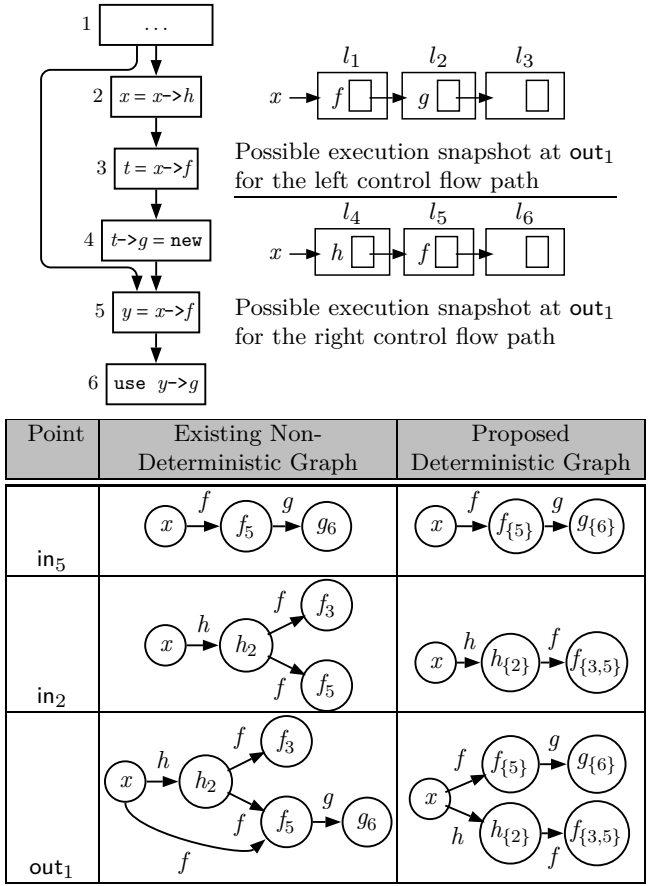

\small
\begin{tabular}{@{}m{2.7cm}m{5.5cm}}
\Gmerge &
\begin{tabular}{@{}l@{}}
\Gmemleft \\ 
\small Possible execution snapshot at $\Outn_1$ \\
\small for the left control flow path \\
\hline
\Gmemright \\
\small Possible execution snapshot at $\Outn_1$ \\
\small for the right control flow path \\
\end{tabular}
\end{tabular}
\small
\begin{tabular}{|c|c|c|}
\hline
\cellcolor{lightgray}
Point & \cellcolor{lightgray}Existing Non- & \cellcolor{lightgray}Proposed  \\
	\cellcolor{lightgray}& \cellcolor{lightgray}Deterministic Graph & \cellcolor{lightgray}Deterministic Graph \\ \hline \hline
$\Inn_5$ & \Gone & \Gnewmergeone \\ \hline
$\Inn_2$ & \Gthree & \Gnewmergethree \\ \hline
$\Outn_1$ &\Gfour &\Gnewmergefour \\ \hline
\end{tabular}
\caption{\small Example. Each node represents live access path(s) reaching it.
For simplicity, aliases of the access paths are not included.
Each node's field label is same as its in-edge label.}
\label{fig:merge}
\end{figure}


\subsection{Proposed Deterministic Graph}
\label{sec:prop-abstr}

We represent a set of live access
paths at program point $\q \in \Q$ as a deterministic graph $\G_\q = \langle \V_\q \cup \N_\q,
\E_\q \rangle$. Here node set $\V_\q \cup \N_\q$ is a discriminated union of
$\V_\q \subseteq \V$, which is a set of variables, and $\N_\q$, which is a set of field nodes,
uniquely identified by a field and its set of use-sites i.e., $\N_\q \subseteq
\F \times 2^{\SS}$. 
In other words, we identify each node with $\f_\tau$ where $\f \in
\F$ and $\tau \subseteq \SS$.

Field nodes are merged if and only if they are represented by the same
$\f_\tau$ across access paths, or within an access path at a program point.
Like the existing representation, all nodes in {${\V_\q \cup \N_\q}$} are
assumed to be accepting states. Set $\V_\q$ is a set of initial
states of the automaton reachable by the variables in $\V_\q$.

Edge set $\E_\q = \E\V_\q \cup \E\F_\q$ is a discriminated
union of
\begin{itemize}
\item $\E\V_\q: \V_\q \rightarrow \N_\q$ which is a set of edges from variables
in $\V_\q$
\item $\E\F_\q: \N_\q \times \F \rightarrow \N_\q$ which is a set of edges from
nodes in $\N_\q$ labeled with fields in $\F$.
\end{itemize}


The example below illustrates how the proposed abstraction improves
summarization.
\begin{myexample}
For the program in Figure~\ref{fig:summ}, at $\Inn_4$, access path $\X.\f.\f$
and its prefixes are live. $\f$ of $l_1$ will be used by statement 4 and $\f$ of
$l_2$ will be used by statements 2 and 4. The proposed abstraction of $\X.\f.\f$ is
shown in the row corresponding to $\Inn_4$. The two nodes $\f_{\{4\}}$ and
$\f_{\{2,4\}}$ are not merged
because their sets of use-sites are different thereby improving the precision of
liveness information as compared to the existing abstraction.
\hfill $\Box$
\end{myexample}

The example below illustrates how the proposed abstraction improves merging of
access paths.
\begin{myexample}
In the program in Figure~\ref{fig:merge}, statements 5 and 6 are equivalent to
``$\text{use }\X \texttt{->} \f \texttt{->} \g$''. Thus, at $\Inn_5$, access path
$\X.\f.\g$ and its prefixes are live. $\f$ of $l_1$ and $\g$ of $l_2$ are used
by the sets of statements $\{5\}$ and $\{6\}$, respectively. The proposed
abstraction is shown in the row corresponding to $\Inn_5$. At $\Inn_2$, access
path $\X.\h.\f$ and its prefixes are live. $\h$ of $l_4$ and $\f$ of $l_5$ will
be used by the sets of statements $\{2\}$ and $\{3,5\}$. The proposed
abstraction is shown in the row corresponding to $\Inn_2$. At the join of the
control flow paths, the nodes $\f_{\{5\}}$ and $\f_{\{3,5\}}$ are not merged because they are
labeled with different sets of use-sites. The spurious access path
$\X.\h.\f.\g$ generated in the existing abstraction is not generated here. 
\hfill $\Box$
\end{myexample}


\section{Analyses for Heap Liveness}
\label{sec:formulations}
This section formalizes existing greedy liveness~\cite{ksk07}
(Section~\ref{sec:greedy}) and our proposed minimal liveness analyses
(Section~\ref{sec:minimal}).  Section~\ref{sec:alias} explains the role of
aliasing in computing the live memory.  It forms the basis of the key difference
between the two analyses.

For simplicity of exposition, we have formulated the data flow values in
Equations~\ref{eq:conv-in-temp} to~\ref{eq:prop-out} as sets of
access paths of unbounded lengths. In the implementation, each data flow value
is represented as a non-deterministic finite automaton in the existing
technique (Section~\ref{sec:conv-abstr}) and as a deterministic finite
automaton in the proposed technique (Section~\ref{sec:prop-abstr}). 

\subsection{The Role of Aliasing}
\label{sec:alias}

A memory link can be {\em reached} by more than one access path i.e., a link may
correspond to the last field of more than one access paths.
Two access paths are said to be {\em link-aliased}~\cite{ksk07} if they share a
memory link. For a sound liveness analysis, all access paths that reach a live
memory link are recorded. Thus, link-aliases of live access paths are
included in the liveness information at each program point so that the full set
of live access paths is recorded.

\subsection{Local Analysis for Computing \Gen and \Kill Sets}

The local effects of generation and killing of live access paths is represented by
$\Kill_\s(\cdot)$ and $\Gen_\s(\cdot)$. As defined in Figure~\ref{table:dfe},
they are parameterized by a set of live access paths denoted \Liv.
The figure lists the basic pointer statements in C, excluding those using addressof
operator $\&$ and non-pointer fields, which are handled in
Section~\ref{sec:addr}.
$\Kill_\s(\cdot)$ is the set of access paths that are killed because its
prefix is defined by the statement. $\Gen_\s(\cdot)$ is the set of access
paths that are generated because either they are used as ``$\text{use } \X$'', or
they are used in the RHS to define a live memory link in the LHS.
$\MayLnA\Inn_\s$, $\MayLnA\Outn_\s$ and $\MustLnA\Inn_\s$, $\MustLnA\Outn_\s$
represent pre-computed may- and must-link-alias~\footnote{Must points-to/alias
information can be computed from may points-to/alias information~\cite{kmr12}
or independently~\cite{scdfy13}.} information at $\Inn_\s$ and
$\Outn_\s$.

\subsection{Existing Greedy Liveness Analysis}
\label{sec:greedy}
The data flow equations for the existing analysis compute sets of live
access paths in $\widehat{\LL}$ or more specifically in $\widehat{\LL}\Inn_{\s}$
and $\widehat{\LL}\Outn_{\s}$ at $\Inn_{\s}$ and $\Outn_{\s}$.

Intermediate values ${\widehat{\LL}\,}'\Inn_{\s}$ and
${\widehat{\LL}\,}'\Outn_{\s}$ are computed using $\Gen_\s(\cdot)$ and
$\Kill_\s(\cdot)$ in Equations~\ref{eq:conv-in-temp} and~\ref{eq:conv-out-temp}. The final $\widehat{\LL}\Inn_{\s}$ and
$\widehat{\LL}\Outn_{\s}$ are computed in Equations~\ref{eq:conv-in}
and~\ref{eq:conv-out} by including link-aliases of the
respective intermediate values. Observe that the link-aliases of access
paths are included during the liveness analysis i.e., greedily the full set of live access
paths is propagated from $\Inn_\k$ to $\Outn_\s$ for all $\k \in \Succ(\s)$ in
Equation~\ref{eq:conv-out-temp}.
Below $\END$ denotes the ending statement of the program, and $\Succ(\s)$
returns the set of control flow successors of statement $\s$ in the control flow
graph. 
\\[-2.3ex]
\begin{align}
\label{eq:conv-in-temp}
{\widehat{\LL}\,}'\Inn_{\s} &= (\widehat{\LL}\Outn_{\s} - \Kill_{\s}(\widehat{\LL})) \cup \Gen_{\s}(\widehat{\LL}) \\
\label{eq:conv-in}
\displaystyle \widehat{\LL}\Inn_{\s} &= \bigcup_{\text{\scalebox{1.2}{$\rho \in {\widehat{\LL}\,}'\Inn_{\s}$}}} \MayLnA\Inn_{\s}(\rho) 
\end{align}

{\white .]}\\[-5.7ex]
\begin{align}
\label{eq:conv-out-temp}
{\widehat{\LL}\,}'\Outn_{\s} &= 
	\begin{cases}
	\begin{array}{lr}
	\emptyset & \quad \s = \END \\
	\displaystyle \bigcup_{\text{\scalebox{1.2}{$\k \in \Succ(\s)$}}} \widehat{\LL}\Inn_{\k} & \quad \text{otherwise}
	\end{array}
	\end{cases}\\
\label{eq:conv-out}
\displaystyle \widehat{\LL}\Outn_{\s} &= \bigcup_{\text{\scalebox{1.2}{$\rho \in {\widehat{\LL}\,}'\Outn_{\s}$}}} \MayLnA\Outn_{\s}(\rho) 
\end{align}

The example below illustrates intermediate and final sets of live access paths.

\begin{myexample}
For the program in Figure~\ref{fig:link}, the intermediate live value
${\widehat{\LL}\,}'\Inn_6$ is $\Y.\f.\g$ and its prefixes. Link
aliases are included to obtain the final value  ${\widehat{\LL}}\Inn_6$ which
contains $\X.\f.\g,\X.\f,\W.\f.\g,\W.\f,\Z.\g$ also. This is obtained using $\MayLnA\Inn_6$
which contains that $\Y.\f.\g$ is link-aliased to $\X.\f.\g, \X.\f,\Z.\g$, and
$\X.\f.\g$ is link-aliased to $\W.\f.\g,\W.\f$.  These are propagated backwards
to $\Outn_3$ and $\Outn_4$. 
$\widehat{\LL}\Outn_4$ is imprecise as links of $l_7-l_{11}$ are
dead.  
\hfill $\Box$
\end{myexample}

The backward propagation of live link-aliased access paths in the existing
technique is the source of imprecision which is handled in the proposed
technique in Section~\ref{sec:minimal}.

\subsection*{Representing Live Access Paths and their Link Aliases}
Each data flow value $\widehat{\LL}\Inn_{\s}$ and $\widehat{\LL}\Outn_{\s}$ is
represented as a non-deterministic graph by labeling each field with its
use-sites. However, some access paths that do not have an explicit use in the program 
are still discovered by the analysis due to link-alias computation. Hence
some fields of link-aliases may not have a use-site. For
example, $\X.\f$ may be live at a program point due to the use of field $\f$ in a subsequent
statement, and $\Y.\g.\f$
may be link-aliased to $\X.\f$.  Field $\g$ may not have a use-site because there may
be no subsequent statement that uses $\g$ in some access path link aliased to $\Y.\g$.
We handle this problem by labelling such field nodes with 
allocation-sites of locations pointed to by the fields. Then field nodes are
merged if and only if either their allocation- or use-site is same.

\subsection{Proposed Minimal Liveness Analysis}
\label{sec:minimal}
To prevent the backward propagation of live link-aliased access paths, we
perform a two-phase analysis called {\em minimal liveness analysis}. Phase 1
uses link-alias information only when required i.e., minimally to generate live access paths
used in the RHS of an assignment that define a live link in the LHS.  This
phase performs fixpoint computation by propagating the access paths backwards in the
control flow path.  Phase 2 includes link-aliases of live access paths computed
in phase 1 to generate the full set of live access paths without backward
propagation.

Our proposed data flow equations compute sets of live access paths in ${\LL}$ or
more specifically in ${\LL}\Inn_{\s}$ and ${\LL}\Outn_{\s}$ at $\Inn_{\s}$ and
$\Outn_{\s}$, respectively. These are computed in two phases.  Intermediate values
${{\LL}}'\Inn_{\s}$ and ${{\LL}}'\Outn_{\s}$ are obtained using a fixpoint
computation in Phase 1, and final values are computed in Phase 2.

\begin{figure}
\setlength{\tabcolsep}{3pt}
\begin{tabular}{|l|l|l|}
\hline
\cellcolor{lightgray}Stmt $\s$ & \cellcolor{lightgray}$\Kill_{\s}(\Liv)$ & \cellcolor{lightgray}$\Gen_{\s}(\Liv)$ \\ \hline \hline
$\text{use } \X$ & $\emptyset$ & $\{\X\}$ \\ \hline
$\X = \new$ & $\{\X.*\}$ & $\emptyset$ \\ \hline
$\X = \Null$ & $\{\X.*\}$ & $\emptyset$ \\ \hline
$\X = \Y$ & $\{\X.*\}$ & $\{\Y.\sigma \mid \X.\sigma \in \Liv\Outn_{\s}\}$ \\ \hline
$\X = \Y \texttt{->} \f$ & $\{\X.*\}$ & $\{\Y.\f.\sigma \mid \X.\sigma \in \Liv\Outn_{\s}\}$ \\ 
		& 	& $ \ \cup \ \{\Y \mid \X.\sigma \in \Liv\Outn_{\s}\}$ \\ \hline
$\X \texttt{->} \f = \Y$ & $\{\Z.\f.* \mid \Z.\f$ & $\{\Y.\sigma \mid \rho.\sigma \in \Liv\Outn_{\s},$ \\ 
	& $\in \MustLnA\Inn_{\s}(\X.\f)\}$ & $\rho \in \MayLnA\Inn_{\s}(\X.\f)\} \cup \{\X\}$ \\ \hline
\end{tabular}
\caption{\small Gen and Kill to compute liveness information at entry and exit of
statement $\s$. Here $\rho \in \V \times \F^*$ and $\sigma \in \F^*$.}
\label{table:dfe}
\end{figure}

\begin{figure}
\setlength{\tabcolsep}{.5pt}
\begin{tabular}{|l|l|l|}
\hline
\cellcolor{lightgray}Stmt $\s$ & \cellcolor{lightgray}$\Kill_{\s}(\Liv)$ & \cellcolor{lightgray}$\Gen_{\s}(\Liv)$ \\ \hline \hline
$\X\!=\!\&\Y$ & $\{\X.*\}$ & $\{\Y.\&.\sigma \mid \X.\sigma \in \Liv\Outn_{\s}\}$ \\ \hline
$\X\!=\!\&(\Y \texttt{->} \f)$ & $\{\X.*\}$ & $\{\Y.\f.\&.\sigma \mid \X.\sigma \in \Liv\Outn_{\s}\}$ \\ 
		& 	& $ \ \cup \ \{\Y \mid \X.\sigma \in \Liv\Outn_{\s}\}$ \\ \hline
$\X \texttt{->} \f\!=\!\&\Y$ & $\{\Z.\f.* \mid \Z.\f$ & $\{\Y.\&.\sigma \mid \rho.\sigma \in \Liv\Outn_{\s},$ \\ 
	& $\in \MustLnA\Inn_{\s}(\X.\f)\}$ & $\rho \in \MayLnA\Inn_{\s}(\X.\f)\} \cup \{\X\}$ \\ \hline
\end{tabular}
\caption{\small Gen and Kill to compute liveness information at entry and exit of
statement $\s$. Here $\rho \in \V \times \F^*$ and $\sigma \in \F^*$.}
\label{table:dfe-2}
\end{figure}

{\bf Phase 1.}
Intermediate values ${{\LL}}'\Inn_{\s}$ and ${{\LL}}'\Outn_{\s}$ are computed using
$\Gen_\s(\cdot)$ and $\Kill_\s(\cdot)$. This phase does not include link-aliases
of live access paths.  It uses link-aliases
to compute $\Gen_\s(\cdot)$ and $\Kill_\s(\cdot)$ (Figure~\ref{table:dfe}).

{\white .}\\[-6ex]
\begin{align}
\label{eq:prop-in-temp}
\LL'\Inn_{\s} &= (\LL'\Outn_{\s} - \Kill_{\s}(\LL')) \cup \Gen_{\s}(\LL') \\
\label{eq:prop-out-temp}
\LL'\Outn_{\s} &= 
	\begin{cases}
	\begin{array}{lr}
	\emptyset & \s = \END \\
	\displaystyle \bigcup_{\text{\scalebox{1.2}{$\k \in \Succ(\s)$}}} \LL'\Inn_{\k} & \text{otherwise}
	\end{array}
	\end{cases}
\end{align}

{\bf Phase 2.}
The final ${\LL}\Inn_{\s}$ and ${\LL}\Outn_{\s}$ are computed by including
link-aliases of the respective intermediate values computed in Phase 1. This
full set of live access paths is not propagated across program points.

{\white .}\\[-6ex]
\begin{align}
\label{eq:prop-in}
\displaystyle \LL\Inn_{\s} &= \bigcup_{\text{\scalebox{1.2}{$\rho \in \LL'\Inn_{\s}$}}} \MayLnA\Inn_{\s}(\rho) \\
\label{eq:prop-out}
\displaystyle \LL\Outn_{\s} &= \bigcup_{\text{\scalebox{1.2}{$\rho \in \LL'\Outn_{\s}$}}} \MayLnA\Outn_{\s}(\rho) 
\end{align}

The example below illustrates intermediate and final sets of live access paths.

\begin{myexample}
For the program in Figure~\ref{fig:link}, in phase 1, the intermediate live
value $\LL'\Inn_6$ is $\Y.\f.\g$ and its prefixes. This is
propagated backwards to $\LL'\Outn_3$ and $\LL'\Outn_4$ without including link-aliases
in {phase~1}. $\Gen(\cdot)$ of statement 3 is $\Z$ and $\Z.\g$ since
$\MayLnA\Inn_6$ contains that $\Y.\f.\g$ is link-aliased to $\X.\f.\g$. 

In phase 2, link-aliases are included at $\Inn_6$ to obtain the final
value $\LL\Inn_6$ which contains $\X.\f.\g,\X.\f,\Z.\g$ also.  Similarly at
$\Outn_3$. At $\Outn_4$, since $\Y.\f.\g$ and its prefixes are not link-aliased
to any other access path, $\LL\Outn_4$ does not generate any spurious access
path.
\hfill $\Box$
\end{myexample}

\subsection*{Representing Live Access Paths and their Link Aliases}
Each data flow value ${\LL}\Inn_{\s}$ and ${\LL}\Outn_{\s}$ is represented as a
deterministic graph by labeling each field with its set of use-sites.  Set of
allocation-sites are labeled on link-aliased fields that do not have a use-site
(Section~\ref{sec:greedy}). 

\subsection{Addressof Operator and Non-Pointer Fields}
\label{sec:addr}
Until now, we have assumed that all variables and fields are of pointer types. For
example, live access path $\X.\f$ represents that pointer field $\f$ is
dereferenced from the location pointed to by pointer variable $\X$. We now extend
access paths to represent non-pointers also. A non-pointer variable $\X$ and a
non-pointer field $\f$ are represented by $\X.\&$ and $\f.\&$, in
the live access paths. This correctly models C style address expressions
`$\&\X$' and `$\&(\X.f)$' as pointers.
For example, live access path $\X.\&.\f$ represents that
pointer field $\f$ is dereferenced from a non-pointer variable $\X$. Live
access path $\X.\f.\&.\g$ represents that non-pointer field $\f$ of the location
pointed to by variable $\X$ dereferences pointer field $\g$.

Liveness analysis of C pointer statements containing $*$, $\&$,
$\X\texttt{->}\f$, and $\X.\f$ is explained below.
\begin{itemize}
\item We model $*$ as field $\deref$; the
statement $\X=*\Y$ is modeled as $\X=\Y\texttt{->}\deref$. Therefore, before
the statements, $\X=*\Y; \text{ use } \X;$, access path $\Y.\deref$ is marked
as live.  
\item C pointer statements containing the
addressof operator `\&' are handled in Figure~\ref{table:dfe-2}.
\item 
C pointer statements containing $\X\texttt{->}\f$ (i.e. pointer variable $\X$)
in LHS/RHS are handled in Figure~\ref{table:dfe}.
\item C pointer statements containing $\X.\f$ (i.e. non-pointer variable $\X$)
in LHS/RHS
are modeled with $\t\texttt{->}\f$ where $\t=\&\X$. For
example, statement $\X.\f=\Y;$ is modeled as $\t=\&\X; t\texttt{->}\f=\Y$;
Expression $\X.\f.\g$ in a C statement can also be modeled similarly using only
$\&$ and $\texttt{->}$.
\end{itemize}
Thus, Figures~\ref{table:dfe} and~\ref{table:dfe-2} handle all C pointer statements.




\section{Interprocedural Analysis}
\label{sec:inter}
We perform context-sensitive interprocedural analysis using
value-contexts~\cite{kk08,pk13}. We traverse
the call graph top-down and represent each
context of a call by an input-output pair of data flow values. 
Thus a new context is created for a procedure only when a new data flow reaches it
through some call.  The
number of contexts is finite because the number of data flow values is finite. 
This follows from the fact that the total number of field nodes that could be
created is bounded by the number of combinations that can be created from the variables, field names, and
the use-sites in a program and there are no parallel edges.
This enables performing a fully context-sensitive analysis even for recursive calls.

For scalability, we perform access-based  localization~\cite{oy11} to
eliminate irrelevant data.
We observe that not all live access paths of the caller need to be passed to the callees. 
Live access paths that are not used in the (direct/indirect) callees, can
be bypassed from the callees (Section~\ref{sec:bypass-unused}). Live
access paths that are not used in the direct callees but used in the indirect
callees, can be memoized in the direct callees
(Section~\ref{sec:bypass-unaff}).


\subsection{Bypassing Unused Information}
\label{sec:bypass-unused}
If a memory link in a caller's live access path is defined
in any of its (direct/indirect) callees, then the live access path needs to be
passed from the caller to the callees. 
The rest can be bypassed from the callees.



\begin{figure}
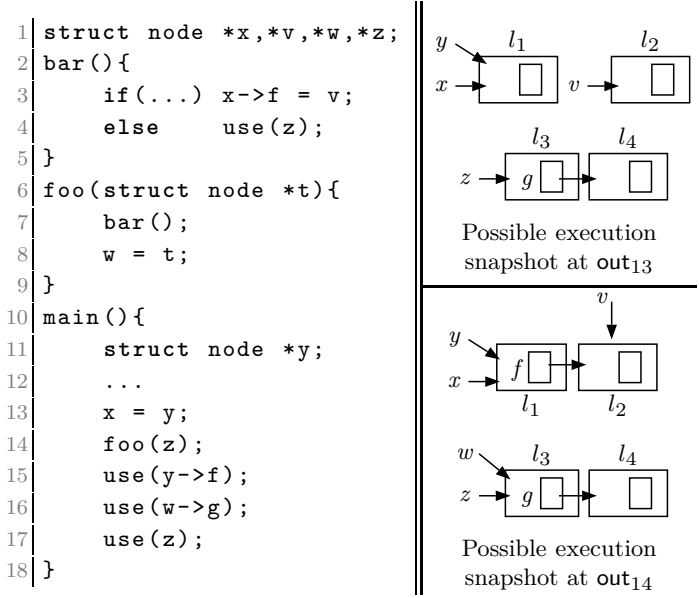

\small
\begin{tabular}{l||@{}l@{}}
\begin{lstlisting}
struct node *x,*v,*w,*z;
bar(){
    if(...) x->f = v;
    else    use(z);
}
foo(struct node *t){
    bar();
    w = t;
}
main(){
    struct node *y;
    ...
    x = y;
    foo(z);
    use(y->f);
    use(w->g);
    use(z);
}
\end{lstlisting}
&
\begin{tabular}{@{}l@{}l@{}}
\Gmemxy & \Gmemv \\ 
\multicolumn{2}{l}{\Gmemzg} \\
\multicolumn{2}{c}{\small Possible execution} \\
\multicolumn{2}{c}{\small snapshot at $\Outn_{13}$} \\
& \\[-2ex]
\hline
 \multicolumn{2}{l}{\Gmemxyf} \\
 \multicolumn{2}{l}{\Gmemzwg} \\
\multicolumn{2}{c}{\small Possible execution} \\
\multicolumn{2}{c}{\small snapshot at $\Outn_{14}$} \\
\end{tabular}
\end{tabular}
\caption{\small Example to illustrate bypassing and memoization used for an efficient
interprocedural liveness analysis.}
\label{fig:inter}
\end{figure}

\begin{myexample} 
There are two function calls, foo() and bar(), in Figure~\ref{fig:inter}.  Live
access path $\Z$ of main() is not defined in callees, foo() and bar().  So, we
bypass $\Z$ from foo() and bar(). Live access path $\W.\g$ of main() is killed
in foo() when $\W$ is defined. So we pass $\W.\g$ from main() to foo(). Live
access path $\Y.\f$ of main() is defined by $\X\texttt{->}\f$ in bar() since
$\X.\f$ and $\Y.\f$ are link-aliased.  So, we pass $\Y.\f$ from main() to
foo() and also from foo() to bar().  Access path $\t.\g$ that is generated in
foo(), is not defined in bar(); therefore, we bypass $\t.\g$ from bar().
\hfill $\Box$
\end{myexample}

\subsection{Memoizing Unaffected Information}
\label{sec:bypass-unaff}
If a memory link of a live access path of a caller is defined in an indirect
callee and not in the direct callee, then the access path passed to the direct
callee will remain unaffected at all program points of the direct callee. Such
access paths need not be passed flow-sensitively, and can be be memoized in
the context of the direct callee flow-insensitively. 


\begin{myexample} 
There are two function calls, foo() and bar(), in Figure~\ref{fig:inter}.  Live
access path $\W.\g$ of main() is killed in foo() when $\W$ is defined 
i.e., $\W.\g$ is live at $\Outn_{8}$ but dead at $\Inn_{8}$ and $\Inn_{7}$; thus,
$\W.\g$ is passed flow-sensitively in foo() context. Live access path $\Y.\f$ of
main() is passed to foo() because $\Y.\f$ is defined by $\X \texttt{->} \f$ in
bar(). However, $\Y.\f$ is not defined in foo() i.e., $\Y.\f$ remains unaffected
at all program points in foo()---$\Y.\f$ is live before and after all
statements of foo(). Therefore, $\Y.\f$ can be memoized and need not be passed
flow-sensitively in foo(). However, $\Y.\f$ should be passed flow-sensitively in
bar() because $\Y.\f$ is live at $\Outn_4,\Inn_4,\Outn_3$ but dead at $\Inn_3$.
\hfill $\Box$
\end{myexample}

\section{Implementation and Measurements}
\label{sec:impl}
We have implemented flow-, context-, and field-sensitive heap liveness
and allocation sites-based points-to analyses~\cite{cwz90} in
GCC 4.7.2 with LTO on 64bit Ubuntu 14.04 using single Intel core
i7-3770 at 3.40 GHz and 8 GiB RAM.  We have measured all C
heap programs in SPEC CPU2006 up to 36 kLoC. We have manually validated the results 
on 7 SPEC CPU2006 and 14 SVComp 2016 benchmarks. For validation,
we have compared these results with brute
force variants that do not perform efficiency related optimizations.


\begin{figure}[t]
\setlength{\tabcolsep}{1.7pt}
\small
\centering
\begin{tabular}{|l|r|r|r|r|r|r|r|}
\hline
\cellcolor{lightgray}SPEC
&\cellcolor{lightgray}kLoC
&\cellcolor{lightgray}Funcs
&\cellcolor{lightgray}Blocks
&\cellcolor{lightgray}Stmts
&\cellcolor{lightgray}Allocs
&\cellcolor{lightgray}Vars
&\cellcolor{lightgray}PT time(s) \\ \hline \hline
lbm     &       0.9     &       17      &       361     &       326     &       1       &       287     &       0.002 \\ \hline
mcf     &       1.6     &       22      &       560     &       449     &       3       &       205     &       0.018 \\ \hline
libq    &       2.6     &       55      &       1236    &       193     &       7       &       322     &       0.013 \\ \hline
milc    &       5.7     &       114     &       3375    &       1379    &       5       &       574     &       0.176 \\ \hline
bzip2   &       9.5     &       56      &       4297    &       1330    &       39      &       1197    &       0.047 \\ \hline
sjeng   &       10.5    &       81      &       5986    &       373     &       10      &       523     &       0.272 \\ \hline
hmmer   &       20.6    &       219     &       8880    &       4879    &       11      &       3217    &       3.521 \\ \hline
h264ref &       36.0    &       497     &       23985   &       16672   &       163     &       8489    &       105.556 \\ \hline
\end{tabular}
\caption{\small Number of lines of code, functions, blocks,
pointer assignments, allocs, variables, and allocation-sites based
points-to analysis time (seconds) in SPEC CPU 2006.}
\label{fig:bench}
\end{figure}

\begin{figure}
{\white .} \\[-3ex]
\includegraphics[width=85mm]{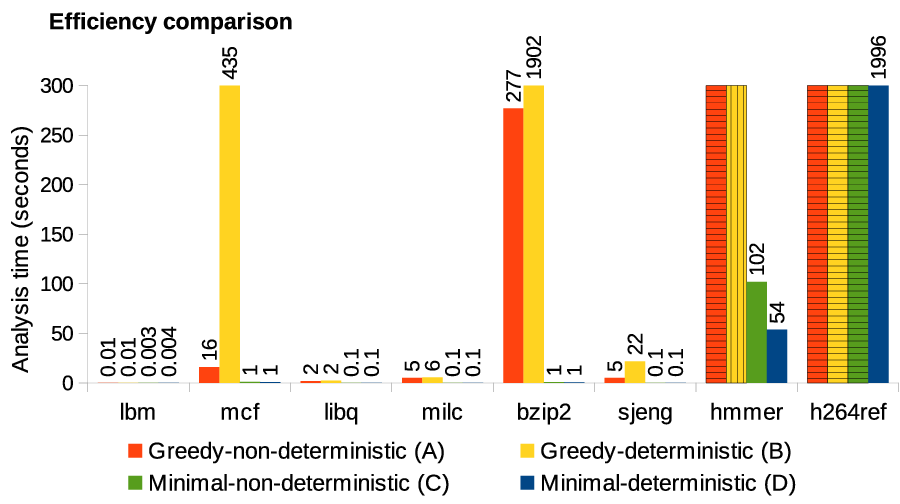} 
\caption{\small Analysis time in seconds of four variants of heap liveness
analysis. Bars with horizontal lines denote out-of-memory and bars with vertical
lines denote analysis time more than 6 hours. 
}
\label{fig:eff-scale}
\end{figure}




\begin{figure*}[t]
\setlength{\tabcolsep}{1pt}
\renewcommand{\arraystretch}{1}
\centering
\small
\begin{tabular}{@{}l@{}l@{}l@{}l@{}}
\includegraphics[height=21.8mm]{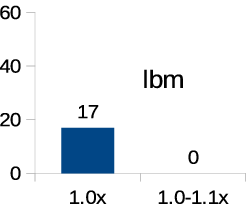} &
\includegraphics[height=21.8mm]{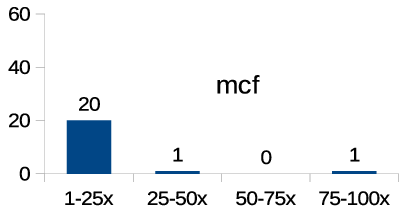} &
\includegraphics[height=21.8mm]{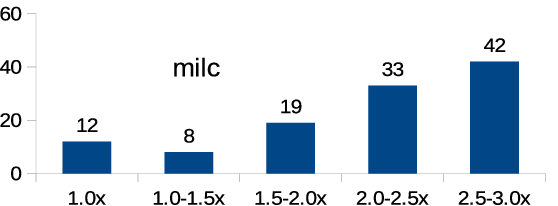} &
\multirow{-6}{*}{
	\begin{tabular}{p{3.3cm}}
	{\bf x-axis}: buckets of the number of times reduction in liveness info using proposed vs.
existing technique. \\
	{\bf y-axis}: number of functions.\\
	{\bf For example}, sjeng has 31 functions on which the proposed technique computes
30-35 times smaller liveness info.
	\end{tabular}}
\\
\includegraphics[height=21.8mm]{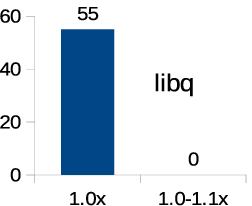} &
\includegraphics[height=21.8mm]{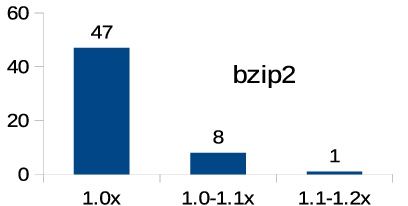} &
\includegraphics[height=21.8mm]{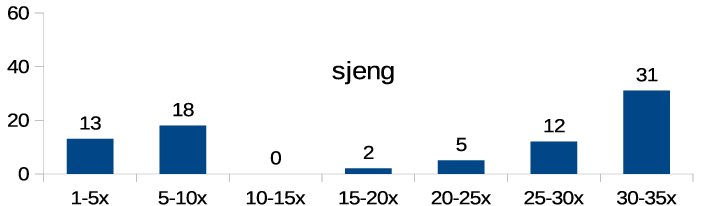} \\
\\[-4ex]
\end{tabular}
\caption{\small Number of times reduction in liveness information computed
using proposed technique vs. existing technique in terms of the
live access paths in each function. Existing
technique does not scale to hmmer and h264ref.}
\label{fig:prec-1}
\end{figure*}

\begin{figure*}[t]
\setlength{\tabcolsep}{1pt}
\renewcommand{\arraystretch}{1}
\small
\begin{tabular}{@{}l@{}l@{}l@{}l@{}}
\includegraphics[height=22mm]{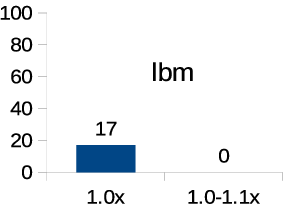} &
\includegraphics[height=22mm]{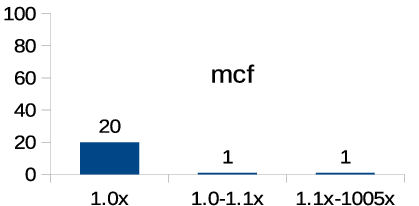} &
\includegraphics[height=22mm]{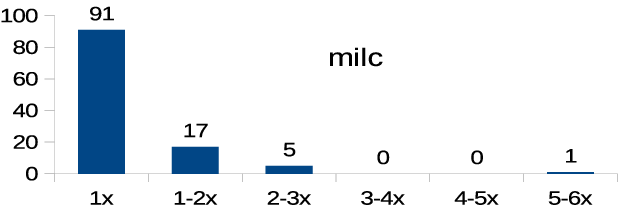} &
\multirow{-6}{*}{
	\begin{tabular}{p{3.4cm}}
	{\bf x-axis}: buckets of the number of times reduction in liveness info using 
proposed vs. existing technique. \\
	{\bf y-axis}: number of functions.\\
	{\bf For example}, sjeng has 3 functions on which the proposed technique 
computes 55-65 times smaller liveness info.
	\end{tabular}}
\\
\includegraphics[height=22mm]{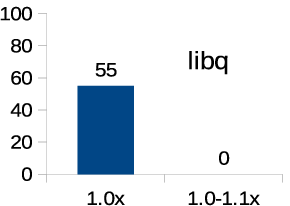} &
\includegraphics[height=22mm]{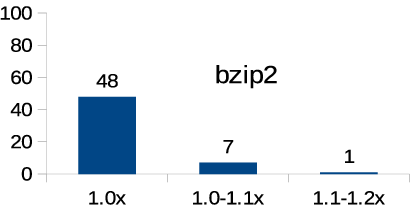} &
\includegraphics[height=22mm]{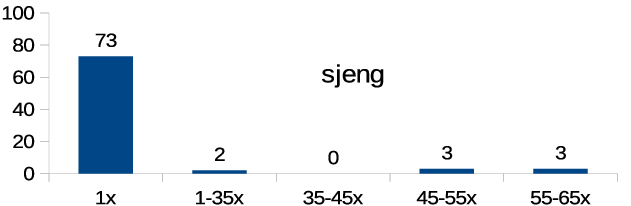} & \\
\\[-4ex]
\end{tabular}
\caption{\small Number of times reduction in liveness information computed
using proposed technique vs. existing technique in terms of the
live access paths within each function's lifetime.
Existing technique does not scale to hmmer and h264ref.}
\label{fig:prec-2}
\end{figure*}


\subsection{Language Features}
All variants of liveness and points-to analyses handle advanced
features of C except pointer arithmetic on structures. Targets of
function pointers are discovered with points-to analysis. Arrays are
abstracted as a single variable without distinguishing between the indices
thereby allowing us to ignore pointer arithmetic on arrays. A function's
address escaped variables are not removed because they can be passed
to the callers.


\subsection{Variants of Heap Liveness Analysis}
In order to study the impact of greedy vs. minimal liveness analysis and the
impact of non-deterministic vs. deterministic automata, we have implemented four
combinations of these techniques. The combination viz. greedy liveness with
non-deterministic automata (labelled as A in the charts) is the existing
technique~\cite{ksk07}. The other three combinations (labelled as B, C, and D
in the charts) are our proposed improvements.  For consistency, we perform
efficiency related optimizations (Section~\ref{sec:inter}) on all variants, and
we perform strong updates in none of the variants.

\subsection{Use of Heap Liveness Analysis}
\label{sec:use}
It is important to identify the following parts of heap using heap liveness
analysis so that one could manage cache misses properly and reclaim unused heap
memory.
\begin{itemize}
\item Part of the heap that is {\em live within the lifetime of a function}. This 
 part can be prefetched for efficiency in the cache on call to the function.
\item Part of the heap that is {\em live in a function}. This part is
required in the program. The rest can be reclaimed in the function since it
remains unused.
\end{itemize}

\begin{myexample}
For the program in Figure~\ref{fig:inter}, the following links are live when
foo() is called: $\X,\Y,\v,\Z$, and $\g$ of location $l_3$. Out of these the
following links are live only within the lifetime of foo(): $\X,\v$, and $\Z$. 
\hfill $\Box$
\end{myexample}

\begin{figure*}[t]
\small
\setlength{\tabcolsep}{3pt}
\begin{tabular}{|l|r|r|r|r|r|r|r|r|r|r|r|}
\hline
\multirow{3}{*}{\cellcolor{lightgray}} & \multicolumn{3}{c|}{\cellcolor{lightgray}}  &
\multicolumn{4}{c|}{\cellcolor{lightgray}Live \% of Heap in a Function} &
\multicolumn{4}{c|}{\cellcolor{lightgray}Live \% of Heap in a Function's Lifetime} \\ \hhline{~~~~--------} 
 \cellcolor{lightgray}& \multicolumn{3}{c|}{\multirow{-2}{*}{\cellcolor{lightgray}Allocated Heap}} & \multicolumn{2}{c|}{\cellcolor{lightgray}Variant A (conv.)} &
\multicolumn{2}{c|}{\cellcolor{lightgray}Variant D (prop.)} & \multicolumn{2}{c|}{\cellcolor{lightgray}\ \ Variant A
(conv.)\ \ } & \multicolumn{2}{c|}{\cellcolor{lightgray}Variant D (prop.)} \\ \hhline{~-----------} 
 \multirow{-3}{*}{\cellcolor{lightgray}SPEC} & \cellcolor{lightgray}Median & \cellcolor{lightgray}Avg & \cellcolor{lightgray}Max & \cellcolor{lightgray}Median & \cellcolor{lightgray}Min & \cellcolor{lightgray}Median & \cellcolor{lightgray}Min & \cellcolor{lightgray}\ Median \ & \cellcolor{lightgray}Min &
 \cellcolor{lightgray}\ Median \ & \cellcolor{lightgray}Min \\ \hline \hline
lbm	&	3	&	3	&	4	&	100\%	&	100\%	&	100\%	&	100\%	&	50\%	&	33\%	&	50\%	&	33\%		\\ \hline
mcf	&	13,058	&	13,108	&	23,186	&	100\%	&	100\%	&	100\%	&	1\%	&	100\%	&	0\%	&	100\%	&	0\% 		\\ \hline
libq	&	5	&	4	&	5	&	100\%	&	40\%	&	100\%	&	40\%	&	100\%	&	0\%	&	100\%	&	0\% 		\\ \hline
milc	&	17	&	16	&	22	&	81\%	&	43\%	&	37\%	&	23\%	&	20\%	&	0\%	&	13\%	&	0\% 		\\ \hline
bzip2	&	2	&	3,435	&	102,541	&	100\%	&	0\%	&	100\%	&	0\%	&	67\%	&	0\%	&	67\%	&	0\% 		\\ \hline
sjeng	&	128	&	126	&	249	&	98\%	&	0\%	&	4\%	&	0\%	&	1\%	&	0\%	&	1\%	&	0\% 		\\ \hline
hmmer	&	225,890	&	205,029	&	691,792	&	?	&	?	&	100\%	&	0\%	&	?	&	?	&	100\%	&	0\% 		\\ \hline
h264ref	&	7,890	&	6,664	&	11,029	&	?	&	?	&	51\%	&	21\%	&	?	&	?	&	5\%	&	0\% 		\\ \hline

\end{tabular}

{E.g., on sjeng, existing analysis reports 98\% memory as live; proposed analysis more precisely reports 4\% memory as live.}
\caption{\small Variant D (proposed) more precisely identifies less percentage of
heap as live than Variant A (existing). 
100\% live memory means that there is no scope of memory reclamation, and 0\%
live memory means that all the received memory can be reclaimed. Variant A
does not scale to hmmer and h264ref. Allocated heap is in terms of number of
access paths.}
\label{fig:prec-3}
\end{figure*}


\subsection{Measuring Amount of Heap Liveness Information}
We compare variants A, B, C, and D of liveness analysis in
terms of the amount of liveness 
within the lifetime of each function and also in each function.
%
For this, we count the number of live
access paths; using allocation-sites to name live objects is too coarse
an abstraction.
Since different variants summarize access paths
differently, a fair comparison is made by reporting access paths of the same
length after fixpoint computation. We have chosen to
retrieve access paths up to lengths 4 (i.e. a variable followed by 3 fields)
for hmmer and up to lengths 5 for rest of the benchmarks. 
Extraction of longer live access paths from the computed existing liveness
graphs, e.g., in hmmer goes out of memory on 8 GB RAM and does not
terminate within 24 hours on 165 GB RAM. This length does not limit
the precision of the techniques since the limiting is applied only 
after fixpoint computation. We extract these counts at the beginning of
each function. 

\subsection{Empirical Observations}
Below we compare the efficiency, scalability, and amount of liveness information
of the four variants
of liveness analysis on SPEC CPU2006 benchmarks (Figure~\ref{fig:bench}).

\subsubsection*{Scalability and Efficiency}
Figure~\ref{fig:eff-scale} shows that variants A and B scale to only 10.5 kLoC
(sjeng), variant C, due to the use of minimal liveness, scales to 20.6 kLoC
(hmmer), and variant D, due to the use of deterministic automaton, scales to 36
kLoC (h264ref). Comparing A with B and C with D, we see that the use of
deterministic automata is more scalable but less efficient than
non-deterministic automata since some time is spent in making the graph
deterministic. Comparing A with C, we see that minimal liveness is more
scalable and significantly more efficient than greedy liveness since it
includes link-aliases minimally.  Variant D is able to analyse
h264ref in 1996 seconds and mcf in 1 second. The average number of edges is
17121 in the existing liveness graph (variant A) and only 5858 in the proposed
liveness graph (variant D) of mcf. 
Overall, variant D (proposed) is more scalable and more efficient
than variant A (existing).

\subsubsection*{Amount of Liveness Information}
Figures~\ref{fig:prec-1} and~\ref{fig:prec-2} present the number of times
reduction in liveness information computed by variant D versus variant
A without compromising on soundness (Appendix~\ref{sec:sound-live}). Amount of
liveness information computed by variants B and C is similar to variants A and
D, respectively.  

As an example, we explain the reduction in liveness information for 81
functions in sjeng (Figure~\ref{fig:prec-1}).  The x-axis of this bar chart has
seven buckets: 1-5x, 5-10x, 10-15x, $\dots$, 30-35x.  On
 31 functions, the proposed technique computes 30-35 times smaller
liveness information. E.g., a function in sjeng receives an allocated
heap of 128 access paths; existing technique reports 125 live access paths and
proposed technique reports only 4 live access paths. Thus, 125/4=31.25 (i.e.,
30-35x bucket). Observe that the number of functions shown on all the bars in
this chart, add up to 81.  Overall, on benchmarks (mcf, milc, sjeng), our
technique shows multifold reduction in the amount of computed liveness, ranging
from 2 to 100 times (in terms of the number of live access paths), without
compromising on soundness.  Figure~\ref{fig:prec-2} shows that the proposed
technique computes 55-65 times smaller liveness in the lifetime of 3 functions
in sjeng. In all the charts, the amount of liveness computed using the proposed
technique is smaller.

We compute both the allocated memory and the percentage of live access paths to
identify benchmarks where memory needs to be reclaimed and cache misses need to
be managed (Figure~\ref{fig:prec-3}). 
The median of the live percentage of access paths for 81
functions in sjeng is 4\% using the proposed technique.  This median in bzip2
and hmmer is poor (100\% access paths is live). However, the minimum on bzip2 and hmmer
is good (0\% access paths is live) and the standard deviation (not shown in the table)
is 17\% and 48\% respectively.  Median 37\% indicates that half the number of
functions have a liveness smaller than 37\%. Median 100\%
indicates that a majority of the functions have a liveness of 100\%.  Overall,
we see that in SPEC CPU2006, very small percentages of the allocated heap are
used both in a function and within its lifetime.

Note that the count of live access paths does not say anything conclusively about
the relative precision of our proposed technique. This is because the numbers
include aliased access paths too and a coarse aliasing may inflate the number
without actually over-approximating the set of live locations. 
Precision is studied in Appendix~\ref{sec:repr}.

\section{Related Work}
\label{sec:rel}
We have already presented a detailed comparison with the existing
technique~\cite{ksk07}, which comes closest to our work. This section
highlights other work related to backward analyses. Details on forward heap
analyses can be found elsewhere~\cite{kk16}. 

Heap liveness information has been modeled as a
store-based~\cite{syks03,amss06,kmr12,kk14} and 
storeless~\cite{cr06,hdh02,gmf06,ksk07}.  The former uses allocation-sites,
types, variables that point to the same locations, and the latter uses access
paths.

Shape analysis techniques~\cite{syks03,amss06} first build a heap points-to
graph based on three-valued logic in a forward analysis. Live objects are then
marked on the points-to graph in a backward analysis using heap safety
automaton~\cite{syks03} and reference counting~\cite{cr06}. These techniques
scale to small benchmarks of SPECjvm98~\cite{syks03} and handwritten
programs~\cite{amss06}.

Allocation-sites~\cite{kmr12,kk14} based summarization of heap
liveness is coarse since even large
benchmarks contain few allocation-sites (Figure~\ref{fig:bench}). For each
allocation-site, the program is sliced using def-use chains with its value-flow
dependences computed using Anderson's points-to analysis to check memory
leaks~\cite{syx12,syx14,yscx16}. 

The number of distinctions made by context sensitive heap cloning~\cite{lla07,slx13,wl04,xr08} 
is a much bigger exponential than that made by sets of access paths.
However, the distinctions need not be made for all
interprocedural paths but only for the paths that use a heap location
differently.  If an allocation site has a million call contexts reaching it but
the location is always referred to by the access path $\X.\f$, there is no
precision gain by distinguishing between all these million heap clones as they
are used in the same way in the program.  
Thus, literature~\cite{slx13,xr08} merges contexts with equivalent points-to
information.  However, they scale at the cost of precision
(they are flow-insensitive, 1-object-sensitive, perform $k$-limiting
equivalence of points-to graphs) and are demand based.

Liveness analysis of only stack variables~\cite{gmf06,hdh02}
has also been performed by considering all heap objects as live.

{ Techniques that use access paths~\cite{sdab16} are incomparable to our
work since they compute points-to information (and not liveness information
like us), and perform demand driven analysis (and not exhaustive analysis like
us). Techniques make points-to graph
deterministic~\cite{tlx17} to only check equivalence of graphs
using HK algorithm; unlike us they do not store it in
a deterministic graph.}


\section{Conclusions}
\label{sec:concl}
Heap liveness analysis can be used to reclaim unused heap memory and improve
cache behaviour. Existing heap liveness analysis~\cite{ksk07} performs greedy
analysis, which includes aliases of live access paths during backward liveness
analysis for soundness. We show how aliases can be considered minimally in a
sound analysis. This reduces the amount of liveness information, and improves
efficiency and scalability.

Further, the existing technique assumes that repeating $\f_\s$ in an access
path (where $\f$ is a field and $\s$ is its use-site) indicates that sequences of
fields after repeating $\f_\s$ would be same. We observe that this need not always
be the case and it is possible that the sequences of fields after repeating $\f_\s$
may be different in access paths if all use-sites of the repeating $\f$ are not
the same along all execution instances. We show that it is more precise to
assume that the sequences of fields are same if repeating $\f$ are dereferenced
in the same set of statements along all execution instances. This 
improves scalability.

Our empirical measurements on SPEC CPU2006 benchmarks show that improved 
liveness analysis scales to 36 kLoC (existing analysis scales to 10.5 kLoC) 
and  improves efficiency even up to 99\%.  For some of the
benchmarks (mcf, milc, sjeng, hmmer, h264ref) { our technique shows multifold
reduction in the amount of liveness information}, ranging from 2 to 100 times
(in terms of the number of live access paths), without compromising on
soundness.


\appendix


\begin{figure*}[t]
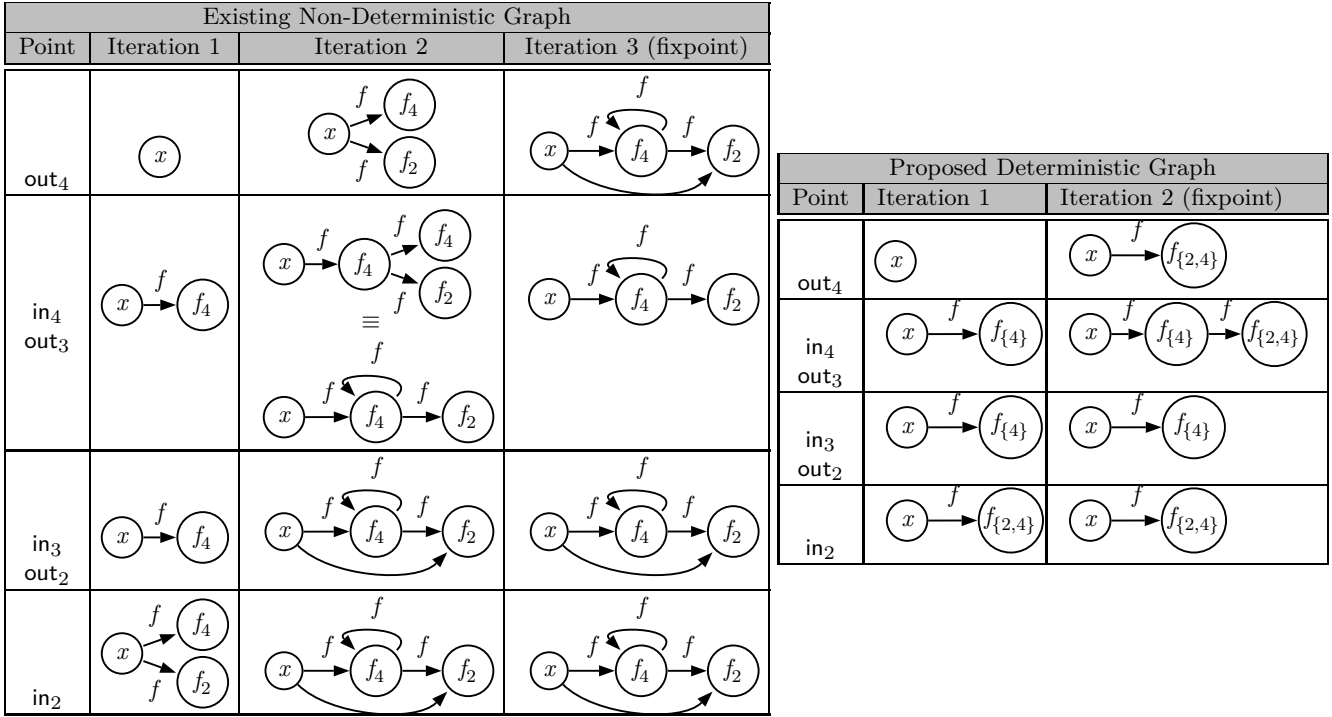

\small
\begin{tabular}{cr}
\begin{tabular}{|c|c|c|c|}
\hline
\multicolumn{4}{|c|}{\cellcolor{lightgray}Existing Non-Deterministic Graph} \\ \hline
 \cellcolor{lightgray}Point & \cellcolor{lightgray}Iteration 1 & \cellcolor{lightgray}Iteration 2 & \cellcolor{lightgray}Iteration 3 (fixpoint) \\ \hline \hline
$\Outn_4$ & \Goldx & \Goldone &  \Goldfour \\ \hline
\begin{tabular}{@{}c@{}} $\Inn_4$ \\ $\Outn_3$ \end{tabular}
	& \Goldtwo &\begin{tabular}{@{}c@{}}\Goldthree \\ $\equiv$ 
		\\ \Goldthreefinal \end{tabular} & \Goldthreefinal \\[-1ex] \hline
\begin{tabular}{@{}c@{}} $\Inn_3$ \\ $\Outn_2$ \end{tabular}
	& \Goldtwo & \Goldfour & \Goldfour \\ \hline
$\Inn_2$ & \Goldone & \Goldfour & \Goldfour \\ \hline
\end{tabular}
\small
\begin{tabular}{|c|l|l|}
\hline
\multicolumn{3}{|c|}{\cellcolor{lightgray}Proposed Deterministic Graph} \\ \hline
 \cellcolor{lightgray}Point & \cellcolor{lightgray}Iteration 1 & \cellcolor{lightgray}Iteration 2 (fixpoint) \\ \hline \hline
$\Outn_4$ & \Goldx & \Gnewone \\ \hline
\begin{tabular}{@{}c@{}} $\Inn_4$ \\ $\Outn_3$ \end{tabular}
	& \Gnewtwo &\Gnewthree \\ \hline
\begin{tabular}{@{}c@{}} $\Inn_3$ \\ $\Outn_2$ \end{tabular}
	& \Gnewtwo & \Gnewtwo \\ \hline
$\Inn_2$ & \Gnewone & \Gnewone \\ \hline
\end{tabular}

\end{tabular}
\caption{\small Step by step working of existing and proposed liveness
analyses on the program in Figure~\ref{fig:summ}. Each node represents
live access path(s) reaching it. Each node's field label
is same as its in-edge label. For simplicity, aliases of the access paths
are not included, and live access paths rooted at only $\X$
are shown.}
\label{fig:summ-steps}
\end{figure*}

\section{Abstracting Live Heap: Example}
Figure~\ref{fig:summ-steps} gives a step by step working of the existing and
the proposed liveness analyses on the program in Figure~\ref{fig:summ}.  The
proposed abstraction is not only more precise than the convention abstraction
but takes less number of iterations to reach fixpoint.

\section{Soundness}
\label{sec:sound-live}
In this section, we define liveness for an execution and then use it to prove soundness 
of our proposed heap liveness analysis.

\subsection{Definition of Liveness for an Execution}
\label{sec:exec-live}

Let $\pi$ denote an execution path represented by a sequence of possibly repeating statements. 
Let $\LnA{\white}_\q(\pi,\rho)$ denote link alias information
of access path $\rho$ in the concrete memory at $\q$
for the execution path $\pi$.  $\LnA{\white}\Inn_\s(\pi,\rho)$ and
$\LnA{\white}\Outn_\s(\pi,\rho)$ denote link alias information at entry and exit of
statement $\s$.  

Let $\q_b$ and $\q_a$ denote program points before and after a sequence of one
or more statements in an execution path.  We relate the access paths at $\q_a$
to the access paths at $\q_b$ by incorporating the effect of both the
intervening statements and of the statements executed before $\q_b$. 

Given an access path $\rho$ at $\q_a$, we define a transfer function
$\TT(\pi',\pi,\rho)$ representing a set of access paths at $\q_b$ that can be
used to reach the link represented by the last field of $\rho$.  The parameters
of $\TT(\pi',\pi,\rho)$ are explained below.
\begin{itemize}
\item $\pi$ is an execution path from $\START$ of the program\footnote{Execution path 
$\pi$ is required in the transfer function $\TT(\cdot)$ to find link aliases along this 
execution path.}.
\item $\pi'$ is a sequence of statements in $\pi$ that appear between
program points $\q_b$ and $\q_a$.
\item $\rho$ is an access path of interest at $\q_a$. 
\end{itemize}

$\TT(\pi',\pi,\rho)$ is defined using $\RR(\pi',\pi,\rho)$ as follows.
$$\TT(\pi',\pi,\rho) = \LnA{\white}_{\q_b}(\pi,\RR(\pi',\pi,\rho))$$


Function $\RR(\cdot)$, defined in Figure~\ref{fig:sound-def}\footnote{The
definition has been adapted from existing heap liveness analysis~\cite{ksk07}.
The existing technique does not use link aliases in its soundness proof
although it uses in the data flow equations.}, captures the transitive effect
of backward transfers of $\rho$ through $\pi'$, and also considers alias
information from $\START$ of the program to end of $\pi'$ i.e., until $\q_a$. 

\begin{figure*}[t]
{
$
\RR(\pi',\pi,\rho) = 
        \begin{cases}
        \begin{array}{l@{\quad \quad}l@{\!\!\!\!\!\!\!\!}l}
	\displaystyle \bigcup_{\rho' \in \RR(\pi_2,\pi,\rho)} \RR(\pi_1,\pi,\rho')	 &   \text{if } \pi' \text{ is a sequence } \pi_1;\pi_2\\
        \{\Y.\sigma\}       & \text{if } \pi' \text{ is } \X=\Y, 			& \rho = \X.\sigma \\
        \{\Y.\f.\sigma\}     & \text{if } \pi' \text{ is } \X=\Y \texttt{->}\f, 	& \rho = \X.\sigma \\
        \{\Y.\sigma\}	& \text{if } \pi' \text{ is } \X\texttt{->}\f=\Y,   	& \forall \sigma, \rho = \rho'.\sigma, \\
			&							& \rho' \in \LnA{\white}\Inn_{\pi'}(\pi,\X.\f)  \\
        \{\rho\}		    & \text{if } \pi' \text{ is a use stmt} 		& \\
        \{\rho\}           	& \text{if } \pi' \text{ is } \X=\dots, 		& \X \text{ is not a prefix of } \rho \\
        \{\rho\}           	& \text{if } \pi' \text{ is } \X\texttt{->}\f=\dots, 	& \nexists \rho', \rho' \in \LnA{\white}\Inn_{\pi'}(\pi,\X.\f) \\
			&							&  \text{ where } \rho' \text{ is a prefix of } \rho \\
        \emptyset      	    & \text{if } \pi' \text{ is } \X=\new, 		& \rho = \X.\sigma    \\
        \emptyset           & \text{if } \pi' \text{ is } \X=\Null, 		& \rho = \X.\sigma \\
	\{\rho\}		&  \text{otherwise}
        \end{array}
        \end{cases} 
$
}
\caption{Defining liveness for an execution.}
\label{fig:sound-def}
\end{figure*}

%


\subsection{Proving Soundness}
This section proves the soundness of 
Equations~\ref{eq:prop-in-temp},~\ref{eq:prop-out-temp},~\ref{eq:prop-in},
and~\ref{eq:prop-out} (this proof has been adapted from 
existing heap liveness analysis~\cite{ksk07}).

\begin{thm}
Let $\rho^a$ be computed in some proposed liveness graph at $\q_a$. Let the
sequence of statements between $\q_b$ to $\q_a$ be $\pi'$ along an execution
path $\pi$.  Then, for $\TT(\pi',\pi, \rho^a)=\Sigma^b$, all access paths in
$\Sigma^b$ will be computed in some proposed liveness graph at $\q_b$.
\end{thm}
\begin{proof}
The proof is by structural induction on $\pi'$. 
The inductive step corresponds to the last case in the definition of $\RR(\pi',\pi,\rho^a)$
and the based cases correspond to the rest of the cases in the definition of $\RR(\pi',\pi,\rho^a)$.
The base cases are:
\begin{enumerate}
\item 
\label{i1}
$\pi'$ is a use statement. In this case, $\Sigma^b$ is a subset of may-link
aliases of $\rho^a$ i.e., $\MayLnA\Inn_{{\pi'}}(\rho^a)$. 

\item 
\label{i2}
$\pi'$ is an assignment $\X=\dots$ such that $\X$ is not a prefix of $\rho^a$. Here
also $\Sigma^b$ is a subset of the may-link aliases of $\rho^a$ i.e.,
$\MayLnA\Inn_{{\pi'}}(\rho^a)$.

\item 
\label{i3}
$\pi'$ is an assignment $\X\texttt{->}\f=\dots$ such that $\nexists \rho$,
 $\rho \in \MustLnA\Inn_{\pi'}(\X.\f)$, where $\rho$ is a prefix of $\rho^a$. Here also $\Sigma^b$ is
a subset of the may-link aliases of $\rho^a$ i.e., $\MayLnA\Inn_{{\pi'}}(\rho^a)$.

\item 
\label{i4}
$\pi'$ is an assignment $\X=\Y$ such that $\rho^a=\X.\sigma$. In this case, $\Sigma^b$
is a subset of the may-link aliases of $\Y.\sigma$ i.e., $\MayLnA\Inn_{{\pi'}}(\Y.\sigma)$.

\item 
\label{i5}
$\pi'$ is an assignment $\X=\Y\texttt{->}\f$ such that $\rho^a=\X.\sigma$. In this case, $\Sigma^b$
is a subset of the may-link aliases of $\Y.\f.\sigma$ i.e., $\MayLnA\Inn_{{\pi'}}(\Y.\f.\sigma)$.

\item 
\label{i6}
 $\pi'$ is an assignment $\X\texttt{->}\f=\Y$ 
 such that $\forall \sigma, \rho^a=\rho.\sigma$, where $\rho \in
\MayLnA\Inn_{{\pi'}}(\X.\f)$. In this case, $\Sigma^b$ is a subset of the
may-link aliases of all $\Y.\sigma$ i.e., $\displaystyle \bigcup_{\forall
\sigma} \MayLnA\Inn_{{\pi'}}(\Y.\sigma)$.
\end{enumerate}

For points~\ref{i1},~\ref{i2}, and~\ref{i3} above, since $\rho^a$ is not in
$\Kill_{\pi'}(\cdot)$, access paths in $\Sigma^b$ are in some proposed liveness
graph at $\q_b$. For the remaining points, since 
$\Sigma^b=\L\Inn_{\pi'}$ (Equation~\ref{eq:prop-in}) represented with an equivalent
automaton\footnote{A set of access paths with unbounded lengths is summarized by merging fields with the
same set of use-sites. This creates a finite automaton, which is an
over-approximation of the set of computed access paths by the proposed technique
at the program point.}, we see that access paths in $\Sigma^b$ are in some proposed liveness
graph at $\q_b$.

For the inductive step, assume that the lemma holds for $\pi_1$ and $\pi_2$. 
From the definition of
$\TT$, there exists a set of access paths $\Sigma^i$ at the intermediate program point $\q_i$
between $\pi_1$ and $\pi_2$, such that $\Sigma^i=\TT(\pi_2,\pi,\rho^a)$. Now we need to transfer
all access paths $\rho^i \in \Sigma^i$ to program point $\q_b$ i.e., $\Sigma^b=
\displaystyle \bigcup_{\rho^i\in \Sigma^i} \TT(\pi_1,\pi,\rho^i)$.
Since $\rho^a$ is in some proposed liveness graph at $\q_a$, by the inductive hypothesis, access paths in
$\Sigma^i$ must be in some proposed liveness graph at $\q_i$. Further, by the induction hypothesis, 
access paths in $\Sigma^b$ must be in some proposed liveness graph at $\q_b$.
\end{proof}


{
\section{Program Characteristics for Precision Benefit}
\label{sec:repr}


{
\begin{figure}[t]
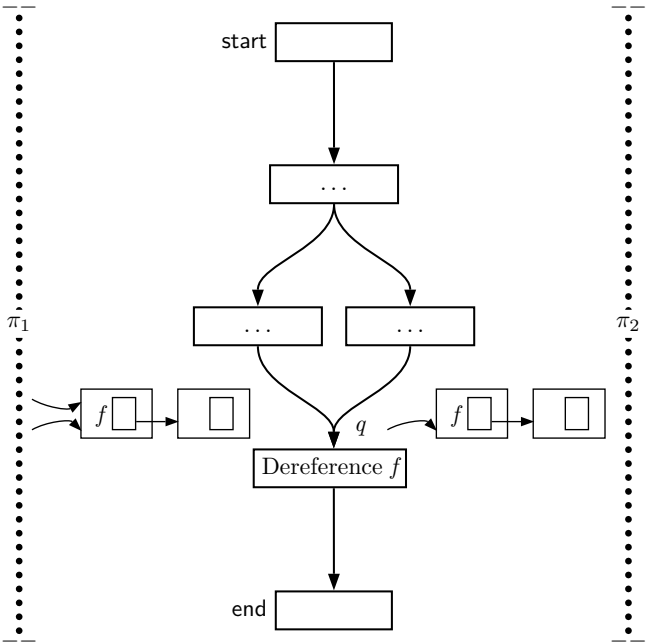

\centering
\repra
\caption{\small Representative feature \#1 in programs.}
\label{fig:repra}
\end{figure}
}

{
\begin{figure}
\footnotesize
{\begin{lstlisting}[numbers=none,lineskip={-1pt}]
struct node { 
    int data;
    struct node * next; 
};
\end{lstlisting}}
\begin{lstlisting}[lineskip={-1pt}]
main() {
    struct node *head, *tmp;
    list_alloc(tmp);
    switch(option) {
    case shallow_copy:
        head = tmp;
        break;
    case deep_copy:
        t=tmp;
        head=new;
        h=head; 
        while(t) {
            h->data=t->data;
            h->next=new;
            h=h->next;
            t=t->next;
        }
        break;
    }
    for ...
        head=head->next;
}
\end{lstlisting}
\caption{
{\small Representative feature \# 1. Shallow and deep copy of a linear linked list.  Locations
reachable by $\tmp.(\Next)^*$ are dead at the end of the deep copy and can
be garbage collected by proposed liveness analysis but not by existing. 
}}
\label{fig:algoa}
\end{figure}
}


Below are two representative features in SPEC CPU2006
benchmarks and common algorithms that show a precision benefit with the use of
proposed heap liveness analysis as compared to the use of existing heap
liveness analysis.
 
\begin{itemize}
\item {Representative feature \#1}
\begin{itemize}
\item Two different control flow
paths $\pi_1$ and $\pi_2$, both from the start to end of the program via a common
program point q. Field $\f$ is used at program point $\q$. Prefixes of field $\f$ at
$\q$ along control flow path $\pi_1$ are different from those at $\q$ along
$\pi_2$ (Figure~\ref{fig:repra}).

\item Figure~\ref{fig:link} is an example of such a program.

\item This representative feature can be found in
SVComp benchmarks and in SPEC CPU2006 benchmarks. We
found at least half a dozen such cases in h264ref; e.g., 
{\em terminate\_slice}(), {\em CheckAvailabilityOfNeighbours}(), {\em intrapred\_luma\_16x16}(),
{\em slice\_too\_big}().

\item Deep copy and shallow copy of a linked list (Figure~\ref{fig:algoa}) 
is such a common algorithm.

\end{itemize}

\item {Representative feature \#2}
\begin{itemize}
\item Two different control flow
paths $\pi_1$ and $\pi_2$, both from the start to end of the program via common
program points $\q$ and $\q'$. Field $\f$ is used at program
point $\q$ and is live at program point $\q'$. The prefixes and suffixes of
field $\f$ accessed in the memory graph after $\q'$ along $\pi_1$ are different
from those accessed in the memory graph after $\q'$ along $\pi_2$ (Figure~\ref{fig:reprb}).

\item Figure~\ref{fig:merge} is an example of such a program.

\item This representative feature is less common and harder to find. Functions
in h264ref are {\em terminate\_slice}() and {\em terminate\_macroblock}().

\item Creation of a new file in a two-level directory structure (Figure~\ref{fig:algob})
is such a common algorithm.

\end{itemize}
\end{itemize}


{
\begin{figure}[t]
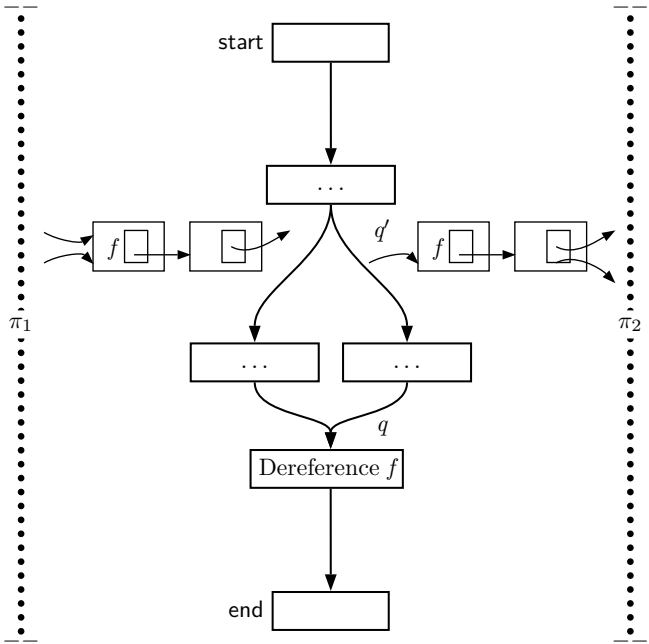

\centering
\reprb
\caption{\small Representative feature \#2 in programs.}
\label{fig:reprb}
\end{figure}
}

{
\begin{figure}
\footnotesize
{\begin{lstlisting}[numbers=none,lineskip={-1pt}]
struct user_t {
    int uid;
    struct file_t * fhead;
    struct user_t * unext; 
};
struct file_t {
    int fid;
    struct fdata_t * fdata;
    struct file_t * fnext; 
};
...
\end{lstlisting}}
\begin{lstlisting}[lineskip={-1pt}]
main() {
    ...
    create_file(uid, fid);
}
create_file(int uid, int fid){
    user_t * user = root;
    file_t * file = user->fhead;
    if(uid is allowed) {
        while(user->uid != uid)
            user = user->unext;
        file = user->fhead;
        while(file->fnext->fid != fid)
            file = file->fnext;
        file->fnext->fdata = new;
    }
    use(file->fnext->fdata);
}
\end{lstlisting}
\caption{
{\small Representative feature \# 2. An operation on two-level directory structure.  Locations
reachable by $\user.(\unext)^+.\fhead.(\fnext)^*.\data$ are dead at exit of
statement 3 and can be garbage collected by proposed liveness analysis
but not by existing.  
}}
\label{fig:algob}
\end{figure}
}

}

\end{document}